\documentclass[a4paper,11pt, twocolumn,accepted=2025-09-26]{quantumarticle}
\pdfoutput=1
\PassOptionsToPackage{
	pdfencoding=auto,
	pdfnewwindow=true,
	pdfusetitle=false,
	bookmarks=true,
	bookmarksnumbered=true,
	bookmarksopen=true,
	pdfpagemode=UseThumbs,
	bookmarksopenlevel=1,
	pdfpagelabels=false,
	breaklinks=true,  
	backref=false, 
	colorlinks=true,
}{hyperref}
\RequirePackage{etex}
\PassOptionsToPackage{skip=0pt,  font=small,labelfont=bf}{caption}
\usepackage{soul}
\usepackage{amsmath}
\usepackage{hyperref}
\usepackage{bbm}
\usepackage{amsthm,accents,amssymb,graphicx,comment,dsfont} 
\usepackage[scr=boondox]{mathalfa}
\usepackage{paralist}
\usepackage{color}
\usepackage{bbold}
\usepackage[none]{hyphenat}
\usepackage{adjustbox}
\usepackage{thm-restate}
\usepackage{geometry}
\usepackage{cite}

\newtheorem{theorem}{Theorem}

\newtheorem{prop}{Proposition}
\newtheorem{corollary}{Corollary}
\newtheorem{lemma}[theorem]{Lemma}

\newtheorem{definition}[theorem]{Definition}

\newcommand{\ket}[1]{\left| #1 \right\rangle}
\newcommand{\bra}[1]{\left\langle #1 \right|}

\newcommand{\beq}{\begin{equation}}
	\newcommand{\eeq}{\end{equation}}
\newcommand{\bea}{\begin{align}}
	\newcommand{\eea}{\end{align}}

\usepackage{graphicx}
\usepackage[usenames,dvipsnames,table]{xcolor} 
\usepackage{mathtools}
\usepackage[normalem]{ulem}
\input{preamble}

\definecolor{googleblue}{RGB}{34, 0, 204}
\definecolor{panblue}{RGB}{0,24,150}
\definecolor{carmine}{RGB}{150, 0, 24}
\hypersetup{
	unicode=true,
	bookmarksnumbered=false,
	bookmarksopen=false,
	breaklinks=false,
	colorlinks=true,
	linkcolor=carmine,
	citecolor=googleblue,
	urlcolor=panblue,
	anchorcolor=OliveGreen}

\newcommand*{\vertbar}{\rule[-1ex]{0.5pt}{2.5ex}}

\newcommand{\bel}{\color{magenta}}
\newcommand{\david}{\color{red}}
\newcommand{\blk}{\color{black}}
\newcommand{\john}{\color{blue}}
\newcommand{\jnote}[1]{{\color{blue}[JHS: #1]}}
\newcommand{\vini}{\color{cyan}}
\newcommand{\ns}{\mkern-3mu} %negative space for f/\upsim

\newcommand{\backtildeacc}{%
  \raisebox{-1.35ex}[.3ex][0pt]{\reflectbox{\normalfont\char"7E}}%
}
\newcommand{\btilde}[1]{\accentset{\scalebox{2}{\backtildeacc}}{#1}}

\newcommand{\changelinkcolor}[1]{\hypersetup{linkcolor=#1}}   

\newcommand{\btp}{\begin{tikzpicture}}
\newcommand{\etp}{\end{tikzpicture}}
\newcommand{\tikzstate}[2]{\begin{pgfonlayer}{nodelayer}
		\node [style=Wsquareadj] (1) at (-0, -0.4) {$#1$};
		\node [style=none] (2) at (-0, 0.4) {};
		\node [style=right label] (3) at (0, 0.3) {$#2$};
	\end{pgfonlayer}
	\begin{pgfonlayer}{edgelayer}
		\draw [qWire] (2.center) to (1.center);
\end{pgfonlayer}}
\newcommand{\tikzeffect}[2]{\begin{pgfonlayer}{nodelayer}
		\node [style=none] (1) at (-0, -0.4) {};
		\node [style=Wsquare] (2) at (-0, 0.4) {$#1$};
		\node [style=right label] (3) at (0, -0.3) {$#2$};
	\end{pgfonlayer}
	\begin{pgfonlayer}{edgelayer}
		\draw [qWire] (2.center) to (1.center);
\end{pgfonlayer}}
\newcommand{\tikzdiscard}[1]{\begin{pgfonlayer}{nodelayer}
		\node [style=none] (1) at (-0, -0.5) {};
		\node [style=upground] (2) at (-0, 0.5) {};
		\node [style=right label] (3) at (0, -0.4) {$#1$};
	\end{pgfonlayer}
	\begin{pgfonlayer}{edgelayer}
		\draw [qWire] (2) to (1.center);
\end{pgfonlayer}}
\allowdisplaybreaks

\usepackage{subcaption}
\usepackage[labelformat=simple]{subcaption}

\usepackage{microtype}
\microtypecontext{spacing=nonfrench}
\microtypesetup{
	expansion={true,nocompatibility},
	protrusion={true,nocompatibility},
	activate={true,nocompatibility},
	tracking=true,
	kerning=true,
	spacing={true}
}

\usepackage[all=normal,floats=tight,mathspacing=tight,wordspacing=tight,paragraphs=normal,tracking=tight,charwidths=tight,mathdisplays=normal,sections=normal,margins=normal]{savetrees}

\everypar=\expandafter{\the\everypar\loosness=-1 }
\newcommand{\baldi}[1]{{{\color{purple}#1}}}
\newcommand{\Baldi}{\color{purple}}

\usepackage{calc}

\usepackage{accents}

\usepackage{tabularx}

\begin{document}
\title{Shadows and subsystems of generalized probabilistic theories:  when tomographic incompleteness is not a loophole for contextuality proofs}
	%\author{\href{mailto:davidschmid10@gmail.com}{David Schmid}}
    \author{David Schmid}
 \email{davidschmid10@gmail.com}
 \affiliation{International Centre for Theory of Quantum Technologies, University of Gda{\'n}sk, 80-309 Gda\'nsk, Poland}
 \affiliation{Perimeter Institute for Theoretical Physics,  N2L 2Y5 Waterloo, Canada}
 %\author{\href{mailto:john.h.selby@gmail.com}{John H.~Selby}}
 \author{John H.~Selby}
	\email{john.h.selby@gmail.com}
 \affiliation{International Centre for Theory of Quantum Technologies, University of Gda{\'n}sk, 80-309 Gda\'nsk, Poland}
	\author{Vinicius P.~Rossi}
 \affiliation{International Centre for Theory of Quantum Technologies, University of Gda{\'n}sk, 80-309 Gda\'nsk, Poland}
        \author{Roberto D.~Baldijão}
\affiliation{International Centre for Theory of Quantum Technologies, University of Gda{\'n}sk, 80-309 Gda\'nsk, Poland}
        \affiliation{Perimeter Institute for Theoretical Physics,  N2L 2Y5 Waterloo, Canada}
	\author{Ana Belén Sainz}
\affiliation{International Centre for Theory of Quantum Technologies, University of Gda{\'n}sk, 80-309 Gda\'nsk, Poland}
\affiliation{Basic Research Community for Physics e.V., Germany}

%	\date{\today}
	\twocolumn[
\begin{@twocolumnfalse}
\begin{abstract}
It is commonly believed that failures of tomographic completeness undermine assessments of nonclassicality in noncontextuality experiments. In this work, we study how such failures can indeed lead to mistaken assessments of nonclassicality. We then show that proofs of the failure of noncontextuality are robust to a very broad class of failures of tomographic completeness, including the kinds of failures that are likely to occur in real experiments. We do so by showing that such proofs actually rely on a much weaker assumption that we term {\em relative tomographic completeness}: namely, that one's experimental procedures are tomographic {\em for each other}. Thus, the failure of noncontextuality can be established even with coarse-grained, effective, emergent, or virtual degrees of freedom. This also implies that the existence of a deeper theory of nature (beyond that being probed in one's experiment) does not in and of itself pose any challenge to proofs of nonclassicality.
To prove these results, we first introduce a number of useful new concepts within the framework of generalized probabilistic theories (GPTs). Most notably, we introduce the notion of a {\em GPT subsystem}, generalizing a range of preexisting notions of subsystems (including those arising from tensor products, direct sums, decoherence processes, virtual encodings, and more). We also introduce the notion of a {\em shadow} of a GPT fragment, which captures the information lost when one's states and effects are unknowingly not tomographic for one another. 
\end{abstract}
  \end{@twocolumnfalse}
]

	\maketitle
		\setcounter{tocdepth}{1}
\clearpage
\thispagestyle{empty}
%\newgeometry{top=1.5cm,bottom=1.5cm,left=2.5cm,right=2.5cm}   %%%%%%%%%% squeeze the table of contents on one page
\tableofcontents
%\restoregeometry

\newpage

	\section{Introduction}

A gold standard universally applicable method for demonstrating that a given theory, experiment, or phenomenon resists classical explanation is to prove that it cannot be reproduced in any generalized noncontextual ontological model~\cite{gencontext}. This approach has been motivated extensively in the literature~\cite{gencontext,Leibniz,schmid2021guiding,negativity,Schmid2024structuretheorem,SchmidGPT}, most notably by a version of Leibniz's principle~\cite{Leibniz} and by its equivalence to the existence of a positive quasiprobabilistic representation~\cite{negativity,Schmid2024structuretheorem} and to the natural notion of classical-explainability in the framework of generalized probabilistic theories~\cite{SchmidGPT,Schmid2024structuretheorem}. Studies characterizing exactly which phenomena constitute proofs of nonclassicality---of the failure of generalized noncontextuality---
have been carried out in the contexts of computation~\cite{Schmid2022Stabilizer,shahandeh2021quantum}, state discrimination~\cite{schmid2018contextual,flatt2021contextual,mukherjee2021discriminating,Shin2021}, interference~\cite{Catani2023whyinterference,catani2022reply,catani2023aspects}, compatibility~\cite{selby2023incompatibility,selby2023accessible,PhysRevResearch.2.013011}, uncertainty relations~\cite{catani2022nonclassical}, metrology~\cite{contextmetrology}, thermodynamics~\cite{contextmetrology,comar2024contextuality}, weak values~\cite{AWV, KLP19}, coherence~\cite{rossi2023contextuality,Wagner2024coherence}, quantum Darwinism~\cite{baldijao2021noncontextuality}, information processing and communication~\cite{POM,RAC,RAC2,Saha_2019,Yadavalli2020,PhysRevLett.119.220402,fonseca2024robustness}, cloning~\cite{cloningcontext}, broadcasting~\cite{jokinen2024nobroadcasting}, and Bell~\cite{Wright2023invertible,schmid2020unscrambling} and Kochen-Specker scenarios~\cite{operationalks,kunjwal2018from,Kunjwal16,Kunjwal19,Kunjwal20,specker,Gonda2018almostquantum}. In each of these contexts, contextuality constitutes a resource for generating novel kinds of correlations, which in many cases translates into advantages for information processing and other tasks. 
Henceforth, we will refer to generalized noncontextuality simply as noncontextuality, and throughout this work we will use the terms \emph{contextuality} and \emph{non-classicality} interchangeably.

It is often claimed that assessments of noncontextuality rest on the assumption that one's laboratory procedures are {\em tomographically complete}~\cite{gencontext,mazurek2021experimentally,schmid2024addressing}: that is, that they allow one to fully characterize the properties of the physical system in question (the causal mediary between one's preparations and measurements). However, we will show here that assessments of noncontextuality  rest on a weaker assumption, namely that one's preparations and measurements are sufficient to fully characterize  {\em each other}, regardless of whether they are tomographic for all properties of the true physical system. We term this the {\em relative tomographic completeness} assumption.

Consequently, we prove that assessments of nonclassicality will be valid even if one's laboratory operations are far from tomographically complete for any given systems in one's experiment or theory, provided that the preparations and measurements used in one's proof of nonclassicality probe {\em the same degree of freedom}. 

To do this, we must first formalize what we mean by `degree of freedom'. 
Inspired by Refs.~\cite{SchmidGPT,muller2023testing}, we  do so within a framework for possible theories known as the framework of generalized probabilistic theories (GPTs)~\cite{Hardy,GPT_Barrett,janotta2013generalized}. 
This framework captures quantum theory, classical theory, and a vast array of other theories as special cases, and allows for a unified study of all of these. While the notion of a system within a GPT is completely standard, we will here introduce the notion of a GPT {\em subsystem}, and then will demonstrate how this captures the relevant notion of a `degree of freedom' for the study of noncontextuality. Our starting point will be a formalisation of GPTs which makes their representation more flexible, and that subsumes standard representations of GPTs (see Sec.~\ref{sec:subsuming}).

Our notion of GPT subsystems is extremely general. It captures tensor product and direct sum subsystems in the usual senses defined in quantum theory (and in arbitrary GPTs), but it also captures such wide-ranging examples as virtual subsystems, stabilizer subsystems, rebit subsystems, decohered subsystems, and so on. Indeed, according to our definition, any fragment of a GPT which itself satisfies all the mathematical properties of a GPT constitutes a GPT subsystem of that GPT. This notion aims to capture what it means for a GPT system to ``live inside'' another, or to be ``explainable by'' another---in both cases, with all its GPT structure intact: if $A$ is a GPT subsystem of $B$, then $B$ can explain $A$ through the notion of GPT explainability introduced by Ref.~\cite{garner2020characterization}.
The main concepts that we introduce pertaining to GPT systems and subsystems, and how they relate to each other, are summarised in Fig.~\ref{fig:sumfig} in Sec.~\ref{subse:sum}.

We then show that the only assumption needed for accurate assessments of nonclassicality is that the states and effects in one's experiment are tomographic for a GPT subsystem of the true GPT. This is simply another way of stating the relative tomographic completeness assumption.

If the states and effects do not satisfy this property, and so are not relatively tomographic, then one can indeed be led to incorrectly conclude that an experiment exhibits nonclassicality. Indeed, this is the sole nontrivial loophole in a state-of-the-art proof of the failure of generalized noncontextuality~\cite{mazurek2016experimental,mazurek2021experimentally}. (Note in particular that there are no loopholes due to detector inefficiencies or superdeterminism~\cite{selby2023incompatibility}.) Yet this assumption, and what happens when it fails, has only been studied explicitly in a single prior work, Ref.~\cite{PuseydelRio} (which we comment on further in the conclusions). 

In this work, we provide rigorous tools for studying this potential loophole, once again by introducing relevant concepts within the framework of generalized probabilistic theories.

Consider an experiment that accesses only some fragment of the preparations and measurements on one's physical system. If a given pair of distinct states (measurement effects) in one's experiment cannot be tomographically characterized by the effects (states) in the experiment, then those states (effects) will appear to be exactly identical within the context of the given experiment. This will then lead the experimenter to mistakenly treat distinct states (effects) as identical.
Consequently, the GPT description one will give to the experiment will be incorrect---it will in fact be given by a particular informationally lossy function of the true GPT fragment describing the scenario. We will call this informationally lossy map the {\em shadow map}. 

If (and only if) the assumption of relative tomography is not satisfied, the shadow map acts nontrivially, distorting the true GPT description of one's experiment. This distortion {\em can}, but does not necessarily, cause one to reach incorrect conclusions about classicality of the experiment. In the final portion of this work, we begin the project of characterizing under what conditions the failure of relative tomography does indeed lead to mistaken assessments of nonclassicality.

\section{Generalized probabilistic theory preliminaries}\label{se:gptp}
	
\subsection{State spaces, effect spaces, systems, and embeddings}

\emph{Generalised probabilistic theories}~\cite{Hardy,janotta2013generalized,GPT_Barrett} (GPTs) provide a framework for studying operational theories at the level of the operational statistics they can generate. When defined in full generality, GPTs have a rich compositional structure, but we will focus here on preparations and measurements on a single system within the theory. Our presentation is not entirely standard, but it has some advantages over prior presentations and introduces some useful new concepts. Like most presentations, we will here consider only finite-dimensional systems. 

To begin, we will consider the structure of states in GPTs.

\begin{definition}[GPT state spaces]\label{def:statespace} 
The states $s$ for a given GPT system $G$ are represented by vectors in a (finite dimensional) real vector space $U_G$. Among these state-vectors, there is
a privileged compact convex subset $\bar{\Omega}_G$ of normalised states, where it must be the case that 
$\mathbf{0}\not\in\mathsf{AffSpan}[\bar{\Omega}_G]$, where $\mathbf{0}$ is the null state (the zero vector). 
The \emph{state space} $\Omega_G$ of a GPT  also includes subnormalized states, and so is the convex hull of the normalised states with the null state, 
\beq\label{eq:normalisable} 
\Omega_G:=\mathsf{Conv}[\mathbf{0}\cup\bar{\Omega}_G].
\eeq
\end{definition}

Note that (unlike in many approaches to GPTs) we will not assume that $\mathsf{Span}[\Omega_G]=U_G$, as this assumption is naturally violated when one starts to consider GPT subsystems, as we will do here.

We depict state $s$ of GPT system $G$ diagrammatically by
\beq
\btp
			\tikzstate{s}{G} \etp \in \Omega_G.
\eeq

We are here viewing GPT state spaces as objects in their own right, rather than merely as part of the specification of a GPT. That is, they are not merely a set of labels for computing probabilities, but more importantly contain a representation of the operational identities among states. An {\em operational identity} among GPT processes is simply a linear equality; for states on a single system\footnote{In general, operational identities can have much more structure; see Refs.~\cite{schmid2024addressing,schmid2020unscrambling}.}, the general form of an operational identity is 
\begin{equation}
     \sum_{s\in\Omega_G} \alpha_s \ s =  \mathbf{0} 
\end{equation}
for $\alpha_s\in\mathds{R}$. The set of all such operational identities\footnote{This set can be easily computed~\cite{schmid2024noncontextuality} and can described as a subspace of $\mathds{R}^{\Omega_G}$.} encodes (and is encoded by)  the convex geometry of the state space, which in turn is a critical aspect of the structure of a given theory. As a simple example of the importance of this convex structure, note that mixed states in classical probability theory have a unique decomposition into pure states, while  mixed quantum states admit of infinitely many decompositions into pure states. 

We often want to consider more than one state space and maps between them. This may be because we are viewing them as state spaces belonging to two different systems within the same GPT and are considering physical transformations between them. Or, it could be that we are viewing them as state spaces belonging to two entirely different GPTs, in which case such maps would not be physical transformations, but rather representations of the states of one system within the state space of another.

\begin{definition}[GPT state maps]\label{def:statemaps}
A state map between two state spaces $\Omega_G$ and $\Omega_H$ is a linear map 
\beq \iota:\mathsf{Span}[\Omega_G] \to \mathsf{Span}[\Omega_H]
\eeq such that $\iota({\Omega}_G) \subseteq {\Omega}_H$
\footnote{It is clear that GPT state maps are closed under sequential composition and that the identity linear map is a state map, so one can define the category $\mathtt{State}$ that has  state spaces as objects and state maps as morphisms. 
}. 

A state map is said to be \emph{faithful} if and only if it is injective, to be \emph{full} (or surjective) if and only if it satisfies $\iota(\Omega_G)=\Omega_H$, and \emph{invertible} if and only if it has a linear inverse which is itself a state map, i.e., such that $\iota^{-1}({\Omega}_H)\subseteq{\Omega}_G$.  Note that if a state map is both full and faithful then it is necessarily invertible. 
\end{definition}

Diagrammatically we denote state maps as
\beq		\btp
		\begin{pgfonlayer}{nodelayer}
			\node [style=none] (1) at (-0, -0.8) {};
			\node [style=none] (2) at (-0,0.8) {};
			\node [style=right label] (4) at (0, -0.8) {$G$};
			\node [style=small box, inner sep=1pt, fill=quantumviolet!30] (5) at (0,0) {$\iota$};
			\node [style=right label] (6) at (0,0.8) {$H$};
		\end{pgfonlayer}
		\begin{pgfonlayer}{edgelayer}
			\draw [qWire] (2.center) to (5);
			\draw [qWire] (5) to (1.center);
		\end{pgfonlayer}
		\etp
\eeq
where we use the violet box to highlight the fact that we are not (necessarily) thinking of these as physical transformations. We represent their action on states via
\beq
\btp
		\begin{pgfonlayer}{nodelayer}
			\node [style=none] (1) at (-0,-0.8) {};
			\node [style=none] (2) at (-0,0.8) {};
			\node [style=right label] (4) at (0,-0.8) {$G$};
			\node [style=small box, inner sep=1pt, fill=quantumviolet!30] (5) at (0,0) {$\iota$};
			\node [style=right label] (6) at (0,0.8) {$H$};
		\end{pgfonlayer}
		\begin{pgfonlayer}{edgelayer}
			\draw [qWire] (2.center) to (5);
			\draw [qWire] (5) to (1.center);
		\end{pgfonlayer}
		\etp\ ::\ \btp
			\tikzstate{s}{G}
   \etp
   \ \mapsto\ \btp
		\begin{pgfonlayer}{nodelayer}
			\node [style=Wsquareadj] (1) at (-0, -1) {$s$};
			\node [style=none] (2) at (-0, 1) {};
			\node [style=right label] (4) at (0, -0.4) {$G$};
			\node [style=small box, inner sep=1pt, fill=quantumviolet!30] (5) at (0,0.2) {$\iota$};
			\node [style=right label] (6) at (0,1) {$H$};
		\end{pgfonlayer}
		\begin{pgfonlayer}{edgelayer}
			\draw [qWire] (2.center) to (5);
			\draw [qWire] (5) to (1.center);
		\end{pgfonlayer}
		\etp \in \Omega_H,
\eeq
where one reads this diagram from bottom to top as composition of linear maps.

Because state maps are linear, they preserve operational identities. In addition, new operational identities are introduced if and only if one's state map is not faithful; that is, for nonfaithful $\iota$, $\iota(\Omega_G)$ will  exhibit equalities that are not present in $\Omega_G$.  

Note also that invertibility of a state map $\iota$ is a strictly stronger condition than invertibility of $\iota$ as a linear map, due to the condition that the inverse of $\iota$ must also be a state map.  For instance, consider embedding a Bloch ball into another Bloch ball by a partially  depolarizing map. This state map  is invertible {\em as a linear map} (if it is not totally depolarizing), but its inverse is not a valid quantum channel (as it can map quantum states to outside the Bloch ball). Consequently the state map is not an invertible state map, since its linear inverse is not a state map. Similarly, note that surjectivity of a state map is a strictly stronger condition than surjectivity as a linear map, as the latter would merely demand that $\iota(\mathsf{Span}[\Omega_G])=\mathsf{Span}[\Omega_H]$.

Two state spaces are isomorphic if there is an invertible state map between them. This is the physically meaningful notion of equivalence for state spaces of GPTs---one that demands that the convex structure (encoding the operational identities) is the same in the two state spaces. It is this notion of equivalence that ensures that it is sensible to define GPT state spaces in a way that is independent of the vector space dimension, as we have done. This is because embedding some state space $\Omega_G\subset U_G$ within a higher dimensional vector space $W$ via some injective linear map $I_\Omega:U_G\to W$ gives a new state space $I_\Omega(\Omega_G)\subset W$ which is isomorphic to the original one. Up to this isomorphism, state spaces are insensitive to the vector space in which they are being represented. So, for example, the state space of a classical bit is the same whether or not one views it as living in a $2$-dimensional vector space, or as embedded in (say) the $4$-dimensional vector space of a qubit (we will see such an example later in Figure~\ref{fig:exSimpP}.)

Next we consider the structure of effects in GPTs.

\begin{definition}[GPT effect spaces]\label{def:effectspace}
The effects $e$ for a given GPT system $G$ are represented by vectors in a (finite dimensional) real vector space $V_G$. The \emph{effect space}, $\mathcal{E}_G$, is a set of such effect vectors that is necessarily convex and compact, that contains the null effect (the zero vector $\mathbf{0}$) and a privileged unit effect $u_G$, and that satisfies
\beq\label{eq:measurementinduced}
\forall e\in\mathcal{E}_G \ \exists e^\perp \in \mathcal{E}_G \ \text{s.t.} \ e+e^\perp = u_G.
\eeq
\end{definition}

Note that (unlike in many approaches to GPTs) we will not assume that $\mathsf{Span}[\mathcal{E}_G]=V_G$, as this assumption is naturally violated when one starts to consider GPT subsystems, as we will do here.

We depict effect $e$ of GPT system $G$ diagrammatically by
\beq
\btp
			\tikzeffect{e}{G} \etp \in \mathcal{E}_G.
\eeq

Just as for state spaces, GPT effect spaces are not merely a tool for computing probabilities; rather, they encode operational identities among effects, which in turn are critical to the structure of a given theory. For effects on a single system, the general form of an operational identity is
\begin{equation}
     \sum_{e\in\mathcal{E}_G} \alpha_e \ e = \mathbf{0} 
\end{equation}
for $\alpha_e\in\mathds{R}$. 
Just as for states, the set of all such operational identities\footnote{ This set can be easily computed~\cite{schmid2024noncontextuality} and can described as a subspace of $\mathds{R}^{\mathcal{E}_G}$.} encodes (and is encoded by)  the convex geometry of the effect space, which in turn is a critical aspect of the structure of a given theory. 

Like with states, we will often want to consider more than one effect space and maps between them, where again these may be physical maps (e.g., representing pre-composition with some physical transformation), or they may just be mathematical maps (e.g., giving a representation of the effects of one theory within another).

\begin{definition}[GPT effect maps]\label{def:effectmaps} Effect maps between two effect spaces $\mathcal{E}_G$ and $\mathcal{E}_H$ are linear maps 
\beq\kappa:\mathsf{Span}[\mathcal{E}_G]\to \mathsf{Span}[\mathcal{E}_H]\eeq
such that $\kappa(\mathcal{E}_G)\subseteq\mathcal{E}_H$  and that preserve the unit effect, so $\kappa(u_G)=u_H$\footnote{Like state maps, effect maps are closed under sequential composition, and the identity linear map is an effect map, so one can define the category $\mathtt{Effect}$ that has  effect spaces as objects and effect maps as morphisms. }.

An effect maps is said to be \emph{faithful} if and only if it is injective, to be \emph{full} (or surjective) if and only if it satisfies $\kappa(\mathcal{E}_G)=\mathcal{E}_H$,  and to be
\emph{invertible} if and only if it has a linear inverse which is itself an effect map, i.e., such that $\kappa^{-1}(\mathcal{E}_H)\subseteq \mathcal{E}_G$.   Note that a full and faithful effect map is necessarily invertible. 
\end{definition}

Diagrammatically, we denote effect maps as
\beq		\btp
		\begin{pgfonlayer}{nodelayer}
			\node [style=none] (1) at (-0, -0.8) {};
			\node [style=none] (2) at (-0,0.8) {};
			\node [style=right label] (4) at (0, -0.8) {$H$};
			\node [style=small box, inner sep=1pt, fill=quantumviolet!30] (5) at (0,0) {$\kappa$};
			\node [style=right label] (6) at (0, 0.8) {$G$};
		\end{pgfonlayer}
		\begin{pgfonlayer}{edgelayer}
			\draw [qWire] (2.center) to (5);
			\draw [qWire] (5) to (1.center);
		\end{pgfonlayer}
		\etp,
\eeq
and we represent their action on effects via
\beq
\btp
		\begin{pgfonlayer}{nodelayer}
			\node [style=none] (1) at (-0, -0.8) {};
			\node [style=none] (2) at (-0, 0.8) {};
			\node [style=right label] (4) at (0, -0.8) {$H$};
			\node [style=small box, inner sep=1pt, fill=quantumviolet!30] (5) at (0,0) {$\kappa$};
			\node [style=right label] (6) at (0, 0.8) {$G$};
		\end{pgfonlayer}
		\begin{pgfonlayer}{edgelayer}
			\draw [qWire] (2.center) to (5);
			\draw [qWire] (5) to (1.center);
		\end{pgfonlayer}
		\etp\ ::\ \btp
			\tikzeffect{e}{G}
   \etp
   \ \mapsto\ \btp\begin{pgfonlayer}{nodelayer}
			\node [style=none] (1) at (-0, -1) {};
			\node [style=Wsquare] (2) at (-0, 1) {$e$};
			\node [style=right label] (4) at (0, -0.9) {$H$};
			\node [style=small box, inner sep=1pt, fill=quantumviolet!30] (5) at (0,-0.2) {$\kappa$};
			\node [style=right label] (6) at (0, 0.5) {$G$};
		\end{pgfonlayer}
		\begin{pgfonlayer}{edgelayer}
			\draw [qWire] (2.center) to (5);
			\draw [qWire] (5) to (1.center);
		\end{pgfonlayer}
		\etp \in \mathcal{E}_H,
\eeq
where one reads these diagrams from top to bottom as composition of linear maps.

Just as for state maps, the linearity of effect maps implies that they preserve operational identities, new operational identities will be introduced if and only if one's effect map is not faithful, and invertibility and surjectivity of an effect map $\kappa$ are strictly stronger conditions than invertibility or surjectivity of $\kappa$ as a linear map.

Two effect spaces are equivalent (isomorphic) if and only if there is an invertible effect map between them. Like with state spaces, this means that embedding some effect space $\mathcal{E}_G\subset V_G$ within a higher dimensional vector space $W'$ via some injective linear map $I_\mathcal{E}:V_G\to W'$ gives a new effect space $I_\mathcal{E}(\mathcal{E}_G)\subset W'$ which is isomorphic to the original one, so effect spaces are also insensitive to which vector space they live in.

The final component of the GPT framework which we need is a probability rule which constitutes the empirical predictions of the theory.
\begin{definition}[Probability rules]\label{def:prob}
A probability rule for a GPT system $G$ can be represented by a bilinear map $p_G:\mathsf{Span}[\Omega_G]\times\mathsf{Span}[\mathcal{E}_G]\to \mathds{R}$ such that $p_G(\Omega_G,\mathcal{E}_G)\subseteq [0,1]$ and such that $p_G(\bar{\Omega}_G,u_G)=1$. 

Moreover, the probability rule is said to be \emph{tomographic} if and only if it perfectly characterises states and effects, so that
\begin{align}
    \label{eq:tomographic}
p_G(s,e)=p_G(s',e),\quad  \forall  e\in \mathcal{E}_G \implies s=s'\\ \nonumber \textrm{ and } p_G(s,e)=p_G(s,e'),\quad \forall s\in\Omega_G \implies e=e'.
\end{align}
\end{definition}

We represent the probability rule diagrammatically as
\beq
\btp
\begin{pgfonlayer}{nodelayer}
	\node [style=none] (1) at (-0, -1) {};
	\node [style=none] (2) at (-0, 1) {};
	\node [style=right label] (4) at (0, -1) {$G$};
        \node [style=empty circle,fill=white] (5) at (0,0) {$p_G$};
	\node [style=right label] (6) at (0, 1) {$G$};
\end{pgfonlayer}
\begin{pgfonlayer}{edgelayer}
	\draw [qWire] (2) to (1.center);
\end{pgfonlayer}
\etp.
\eeq
We represent it as a dashed circle because it should not (necessarily) be interpreted as 
a state map or an effect map, but a map from state-effect pairs to the reals:
\beq
\btp
\begin{pgfonlayer}{nodelayer}
	\node [style=none] (1) at (-0, -1) {};
	\node [style=none] (2) at (-0, 1) {};
	\node [style=right label] (4) at (0, -1) {$G$};
        \node [style=empty circle,fill=white] (5) at (0,0) {$p_G$};
	\node [style=right label] (6) at (0, 1) {$G$};
\end{pgfonlayer} 
\begin{pgfonlayer}{edgelayer}
	\draw [qWire] (2) to (1.center);
\end{pgfonlayer}
\etp ::\ \left(\btp
			\tikzstate{s}{G} \etp  , \btp
		\tikzeffect{e}{G}
		\etp\right)\ \  \mapsto\quad
\btp
		\begin{pgfonlayer}{nodelayer}
			\node [style=Wsquareadj] (1) at (-0, -1.5) {$s$};
			\node [style=Wsquare] (2) at (-0, 1.5) {$e$};
			\node [style=right label] (4) at (0, 1) {$G$};
			\node [style=empty circle, fill=white] (8) at (0,0) {$p_G$};
			\node [style=right label] (6) at (0, -0.8) {$G$};
		\end{pgfonlayer}
		\begin{pgfonlayer}{edgelayer}
			\draw [qWire] (2) to (1.center);
		\end{pgfonlayer}
		\etp \in [0,1]
\eeq

Putting the above together, we can define a system in a GPT as follows.
\begin{definition}[GPT system]\label{def:GPTsystem}
A GPT system $G$ is a triple of a state space $\Omega_G$, an effect space $\mathcal{E}_G$, and a tomographic probability rule $p_G$:
\begin{equation}
 \label{eq:GPT}
		G:=\left(\left\{\btp
			\tikzstate{s}{G}
		\etp\right\}_{s\in\Omega_G},
		\left\{\btp
		\tikzeffect{e}{G}
		\etp\right\}_{e\in\mathcal{E}_G},\quad
\btp
\begin{pgfonlayer}{nodelayer}
	\node [style=none] (1) at (-0, -1) {};
	\node [style=none] (2) at (-0, 1) {};
	\node [style=right label] (4) at (0, -1) {$G$};
        \node [style=empty circle,fill=white] (5) at (0,0) {$p_G$};
	\node [style=right label] (6) at (0, 1) {$G$};
\end{pgfonlayer}
\begin{pgfonlayer}{edgelayer}
	\draw [qWire] (2) to (1.center);
\end{pgfonlayer}
\etp\right) \,.
	\end{equation}
\end{definition}

We can then define the notion of a GPT system embedding as follows.

\begin{definition}[GPT system embedding \cite{muller2023testing}]\label{defn:GPTembedding}
A \emph{GPT system embedding}\footnote{ It is straightforward to see that we can sequentially compose such embeddings by element-wise composition of the pair, and that any GPT system embeds into itself by the identity state and effect maps. We can therefore define the category  $\mathtt{GPT-System}$ that has  GPT systems as objects and GPT embeddings as morphisms. } from $G$ to $H$ is a pair of a state map $\iota:\Omega_G\to\Omega_H$ and an effect map $\kappa:\mathcal{E}_G\to\mathcal{E}_H$
which taken together preserve probabilities, so that $p_G(e,s)=p_H(\iota(s),\kappa(e))$ for all $s\in\Omega_G$ and $e\in\mathcal{E}_H$. Diagrammatically, it is a state and effect map satisfying
\beq
\btp
		\begin{pgfonlayer}{nodelayer}
			\node [style=none] (1) at (-0, -1.5) {};
			\node [style=none] (2) at (-0, 1.5) {};
			\node [style=right label] (4) at (0, 1) {$G$};
			\node [style=empty circle, fill=white] (8) at (0,0) {$p_G$};
			\node [style=right label] (6) at (0, -0.8) {$G$};
		\end{pgfonlayer}
		\begin{pgfonlayer}{edgelayer}
			\draw [qWire] (2) to (1.center);
		\end{pgfonlayer}
		\etp \ =\ 
		\btp
		\begin{pgfonlayer}{nodelayer}
			\node [style=none] (1) at (-0, -2.6) {};
			\node [style=none] (2) at (-0, 2.6) {};
			\node [style=right label] (4) at (0, 2.4) {$G$};
			\node [style=empty circle,fill=white] (8) at (0,0) {$p_H$};
			\node [style=right label] (6) at (0, -2.2) {$G$};
              \node [style=small box,inner sep=1pt,fill=quantumviolet!30] (10) at (0,1.65) {$\kappa$};
              \node [style=right label] (11) at (0,1.1) {$H$};
              \node [style=right label] (12) at (0,-0.8) {$H$};
              \node [style=small box,inner sep=1pt,fill=quantumviolet!30] (13) at (0,-1.55) {$\iota$};
		\end{pgfonlayer}
		\begin{pgfonlayer}{edgelayer}
			\draw [qWire] (2) to (1.center);
		\end{pgfonlayer}
		\etp.
\eeq
Note that both $\iota$ and $\kappa$ must be faithful as a consequence of the fact that the GPT system has a tomographic probability rule. The embedding is said to be \emph{invertible} if both $\iota$ and $\kappa$ are moreover invertible. 
\end{definition}

In the language of Ref.~\cite{muller2023testing}, this is an exact unital embedding of one GPT into another.

If $\iota$ and $\kappa$ are invertible state and effect maps from $G$ to $H$, then (by definition) their inverses are state and effect maps from $H$ to $G$. It follows that $(\iota^{-1},\kappa^{-1})$ gives an embedding of $H$ into $G$, since these maps necessarily also preserve probabilities: 
\beq \label{twodirections}
    \btp
		\begin{pgfonlayer}{nodelayer}
			\node [style=none] (1) at (-0, -1.5) {};
			\node [style=none] (2) at (-0, 1.5) {};
			\node [style=right label] (4) at (0, 1) {$G$};
			\node [style=empty circle, fill=white] (8) at (0,0) {$p_G$};
			\node [style=right label] (6) at (0, -0.9) {$G$};
		\end{pgfonlayer}
		\begin{pgfonlayer}{edgelayer}
			\draw [qWire] (2) to (1.center);
		\end{pgfonlayer}
		\etp
  \ =\ 
  \btp
		\begin{pgfonlayer}{nodelayer}
			\node [style=none] (1) at (-0, -3) {};
			\node [style=none] (2) at (-0, 3) {};
			\node [style=right label] (4) at (0, 2.4) {$G$};
			\node [style=empty circle, fill=white] (8) at (0,0) {$p_H$};
			\node [style=right label] (6) at (0, -2.2) {$G$};
              \node [style=small box,inner sep=1pt,fill=quantumviolet!30] (10) at (0,1.65) {$\kappa$};
              \node [style=right label] (11) at (0,0.9) {$H$};
              \node [style=right label] (12) at (0,-0.9) {$H$};
              \node [style=small box,inner sep=1pt,fill=quantumviolet!30] (13) at (0,-1.55) {$\iota$};
		\end{pgfonlayer}
		\begin{pgfonlayer}{edgelayer}
			\draw [qWire] (2) to (1.center);
		\end{pgfonlayer}
		\etp
  \ \implies\ 
  \btp
		\begin{pgfonlayer}{nodelayer}
			\node [style=none] (1) at (-0, -3) {};
			\node [style=none] (2) at (-0, 3) {};
			\node [style=right label] (4) at (0, 2.4) {$H$};
			\node [style=empty circle, fill=white] (8) at (0,0) {$p_G$};
			\node [style=right label] (6) at (0, -2.2) {$H$};
              \node [style=small box,inner sep=1.2pt,fill=quantumviolet!30] (10) at (0,1.65) {$\kappa^{-1}$};
              \node [style=right label] (11) at (0,0.9) {$G$};
              \node [style=right label] (12) at (0,-0.9) {$G$};
              \node [style=small box,inner sep=1pt,fill=quantumviolet!30] (13) at (0,-1.55) {$\iota^{-1}$};
		\end{pgfonlayer}
		\begin{pgfonlayer}{edgelayer}
			\draw [qWire] (2) to (1.center);
		\end{pgfonlayer}
		\etp
  \ =\ 
  \btp
		\begin{pgfonlayer}{nodelayer}
			\node [style=none] (1) at (-0, -1.5) {};
			\node [style=none] (2) at (-0, 1.5) {};
			\node [style=right label] (4) at (0, 1) {$H$};
			\node [style=empty circle, fill=white] (8) at (0,0) {$p_H$};
			\node [style=right label] (6) at (0, -0.9) {$H$};
		\end{pgfonlayer}
		\begin{pgfonlayer}{edgelayer}
			\draw [qWire] (2) to (1.center);
		\end{pgfonlayer}
		\etp.
\eeq
Consequently for any invertible GPT embedding from $G$ into $H$ with state and effect maps $(\iota,\kappa)$, the inverse maps $(\iota^{-1},\kappa^{-1})$ give a GPT embedding of $H$ into $G$.

It is straightforward to check that GPT embeddings are transitive~\cite{muller2023testing}: if we have an embedding of $G$ into $H$ by $(\iota,\kappa)$ and $H$ into $K$ by $(\delta,\gamma)$, then the sequential composition of these, $(\delta\circ \iota, \gamma\circ \kappa)$, will be an embedding of $G$ into $K$. 

Two GPT systems are said to be isomorphic if and only if there is an invertible embedding between them.
\begin{definition}[Equivalence of GPT systems]\label{defn:GPTsystemequivalence}
Two GPT systems are equivalent (isomorphic) if and only if there is an invertible GPT system embedding (as in Definition~\ref{defn:GPTembedding}) between them.
\end{definition}

This definition of equivalence subsumes the standard one found in the literature, as we show in the next section. 

Note that if we simply embed the states and effects of a GPT system within some higher dimensional vector spaces via injective linear maps $I_\Omega:U_G\to W$ and $I_\mathcal{E}:V_G\to W'$, this defines an isomorphic GPT system, $(I_\Omega(\Omega_G),I_\mathcal{E}(\mathcal{E}_G), p_G(P_\Omega( \cdot ),P_\mathcal{E}( \cdot )))$ where $P_\Omega$ is any left inverse of $I_\Omega$ and similarly for $P_\mathcal{E}$. This shows once again how our definitions are insensitive to the particular vector spaces in which one chooses to represent the states and effects.

\subsection{Subsuming the standard notion of a GPT}\label{sec:subsuming}

 Unlike in standard presentations of GPTs, our definition of a GPT system involves an explicit probability rule; moreover (as mentioned earlier), we do not assume that the states and effects on the system span the vector spaces in which they live. 
Our definition of a GPT system is not essentially distinct from that found in the literature, but at the representational level it is more flexible. For instance, it subsumes two common but distinct kinds of representations of GPTs: those that represent effects as linear functionals on states, and those that represent effects as living in the same vector space as states. This is convenient, as it enables one to switch freely between these two descriptions, each of which is useful in certain contexts. We will see further advantages of our definitions when we come to the representation of GPT fragments (see the next section and the comment about accessible GPT fragments therein).

\begin{theorem}
Every GPT system in our framework is isomorphic to one in each of the standard forms found in the literature. Conversely, every GPT in one of the standard forms found in the literature is an instance of a GPT system in our framework. 
\end{theorem}
\begin{proof}
First, let us show how to convert any given GPT system in our framework to a GPT in the standard forms found in the literature~\cite{Hardy,GPT_Barrett,Masanes_2011,plavala2021general}. If one is given a GPT system whose states do not span the vector space $U_G$ in which they live and whose effects do not span the vector space $V_G$ in which they live, one first transforms this into an isomorphic GPT system in which the states live in a smaller vector space which they {\em do} span, and in which the effects live in a smaller vector space which they {\em do} span.
One can do this, for instance, using the projection maps $P_\Omega$ and $P_\mathcal{E}$, where $P_\Omega$ is a projection $U_G\to \mathsf{Span}[\Omega_G]$ and where $P_\mathcal{E}$ is a projection $V_G \to \mathsf{Span}[\mathcal{E}_G]$; these give a GPT isomorphism between the two GPT systems. 

One can then lump the probability rule together with the effects to transform them into linear functionals on $\mathsf{Span}[\Omega_G]$, i.e., for each $e\in \mathcal{E}_G$ we can define a linear functional
\begin{align}
  \btp
        \begin{pgfonlayer}{nodelayer}
            \node [style=Wsquare] (0) at (0,1) {$e$};
            \node [style=none] (1) at (0,-1) {};
            \node [style=empty circle,fill=white] (2) at (0,-0.2) {$p_G$};
            \node [style=right label] (3) at (0,0.5) {$G$};
            \node [style=right label] (4) at (0,-1) {$G$};
        \end{pgfonlayer}
        \begin{pgfonlayer}{edgelayer}
            \draw [qWire] (0) to (1.center);
        \end{pgfonlayer}
    \etp  : \Omega_G \to [0,1]
    \ \ :: \btp \tikzstate{s}{G} \etp \ \  \mapsto\quad
\btp
		\begin{pgfonlayer}{nodelayer}
			\node [style=Wsquareadj] (1) at (-0, -1.5) {$s$};
			\node [style=Wsquare] (2) at (-0, 1.5) {$e$};
			\node [style=right label] (4) at (0, 1) {$G$};
			\node [style=empty circle, fill=white] (8) at (0,0) {$p_G$};
			\node [style=right label] (6) at (0, -0.8) {$G$};
		\end{pgfonlayer}
		\begin{pgfonlayer}{edgelayer}
			\draw [qWire] (2) to (1.center);
		\end{pgfonlayer}
		\etp.
\end{align} 
This recovers the common view~\cite{Masanes_2011,plavala2021general} of effects as living in the vector space dual to that for states---i.e., the view wherein effects are {\em defined} as linear functionals on states. 

One can then transform this representation into the other common representation~\cite{Hardy,GPT_Barrett} of GPTs, where states and effects are vectors in the same vector space (as opposed to dual vector spaces). The equivalence between these two kinds of representations is well known, and relies on choosing an (arbitrary) inner product and using the Riesz representation theorem to get an inner product representation. 

Conversely, a given GPT in either of these common forms is already an instance of a GPT system as we have defined it. 
A GPT wherein effects are viewed as linear functionals is a special case of a GPT system as in our Definition~\ref{def:GPTsystem}, where the probability rule is given by the evaluation map (i.e., by directly applying the linear functional to the state: $p(s,e)=e(s)$). A GPT wherein states and effects live in the same vector space is also a special case of a GPT system as in our Definition~\ref{def:GPTsystem}, where the probability rule is given by the inner product: $p(s,e)=\left<e,s\right>$. 
\end{proof}

Our notion of equivalence for GPT systems (Definition~\ref{defn:GPTsystemequivalence}) also subsumes the standard one in the literature: that two GPTs are equivalent if one is generated from the other by acting an invertible linear map on its state space and the inverse of that map on the effect space, while keeping the probability rule fixed.

\subsection{Fragments and fragment embeddings}\label{sec:frags}

A GPT, viewed as a theory, is understood as describing every physically possible process in some (possibly hypothetical) world. However, sometimes one is interested in giving a GPT description of some particular experiment or particular processes within the world. In such cases, one needs a more general notion which allows, for example, for the possibility that the set of states and set of effects that are of interest might not be tomographic for each other. 
This leads to the notion of a {\em GPT fragment}. A fragment that is not itself a GPT system is typically defined relative to some background GPT system.

\begin{definition}[Fragment of a GPT]\label{def:GPTfragment}
A {\em fragment} $f$ of a GPT system $G:=(\Omega_G,\mathcal{E}_G, p_G)$ has a GPT state space defined by 
a subset $\bar{\Omega}_{f}\subseteq \bar{\Omega}_G$ of the normalized states in $G$. As with all GPT state spaces, the full state space is
$\Omega_f:=\mathsf{Conv}[\mathbf{0}\cup\bar{\Omega}_f]$.
It has a GPT effect space defined by a subset $\mathcal{E}_{f}\subseteq \mathcal{E}_G$ of effects from the full GPT. As with all effect spaces, $\mathcal{E}_{f}$ is convex and compact, must include the null effect $\mathbf{0}$ and the unit effect $u_f:=u_G$, and satisfy
$\forall e\in\mathcal{E}_f \ \exists e^\perp \in \mathcal{E}_f \ \text{s.t.} \ e+e^\perp = u_f$.
Probabilities in the fragment are computed using the probability rule $p_G$, but with its domain restricted to (the span of) the states and effects in the fragment; that is, one defines the probability rule by defining $p_f(s,e):=p_G(s,e)$ for all $s\in \Omega_f$ and $e\in \mathcal{E}_f$ and extending this to a bilinear map $p_f:\mathsf{Span}[\Omega_f]\times \mathsf{Span}[\mathcal{E}_f]\to \mathds{R}$. A GPT fragment is therefore a triple $f:=(\Omega_f,\mathcal{E}_f,p_f)$. A fragment is said to be relatively tomographic if its probability rule is tomographic.
 \end{definition}

In the present work, the {\em only} difference between a GPT fragment and a proper GPT is that for the former, we do not require that the probability rule $p_f$ is tomographic (i.e., it does not necessarily satisfy Eq.~\eqref{eq:tomographic}). Thus, every GPT fragment that is relatively tomographic satisfies all the properties of a proper GPT. Note that in this special case, the span of the states is isomorphic to the span of the effects:
\begin{align}
    \mathsf{Span}[\Omega_f] \cong \mathsf{Span}[\mathcal{E}_f].
\end{align} 
For generic fragments, however, this may not be the case.

It is sometimes useful to consider a broader notion of a fragment that does not demand that every state in the fragment has a normalized counterpart that is also in the fragment~\cite{selby2023accessible}, or in which the effect space does not contain the unit or complementary effects, or in which the set of states and set of effects are not convexly closed. Such generality is necessary to describe a real experiment, where (for example) one can only implement a finite number of states and effects.
However, we do not consider this extra generality here, as it obfuscates our theoretical analyses without having any meaningful impact on any of our results.
(Moreover, any fragment in this broader sense can be uniquely closed into a fragment in the sense in our work, and one can infer the exact probabilities generated by the closed fragment from the original fragment.) 

Clearly, GPT systems are the special case of GPT fragments that are relatively tomographic. Moreover, one can lift the definition of GPT system embeddings to the more general case of GPT fragment embeddings. The only difference is that when the fragment to be embedded is not relatively tomographic, it does not necessarily follow that the embedding maps $\iota$ and $\kappa$ are faithful.

\begin{definition}[GPT fragment embeddings]\label{def:fragembeddings}
A \emph{GPT fragment embedding} from GPT fragment $g$ to GPT fragment $h$ is a pair of a state map $\iota:\Omega_g\to\Omega_h$ and an effect map $\kappa:\mathcal{E}_g\to\mathcal{E}_h$ which together preserve probabilities: that is, $p_g(e,s)=p_h(\iota(s),\kappa(e))$ for all $s\in\Omega_g$ and $e\in\mathcal{E}_h$\footnote{Like with GPT systems, we can form a category of GPT fragments and GPT fragment embeddings which we call $\mathtt{GPT-Fragment}$. Clearly, $\mathtt{GPT-System}$ is a full subcategory of $\mathtt{GPT-Fragment}$.}. Diagrammatically, it is a pair satisfying
\beq
\btp
		\begin{pgfonlayer}{nodelayer}
			\node [style=none] (1) at (-0, -1.5) {};
			\node [style=none] (2) at (-0, 1.5) {};
			\node [style=right label] (4) at (0, 1) {$g$};
			\node [style=empty circle, fill=white] (8) at (0,0) {$p_g$};
			\node [style=right label] (6) at (0, -0.8) {$g$};
		\end{pgfonlayer}
		\begin{pgfonlayer}{edgelayer}
			\draw [qWire] (2) to (1.center);
		\end{pgfonlayer}
		\etp \ =\ 
		\btp
		\begin{pgfonlayer}{nodelayer}
			\node [style=none] (1) at (-0, -2.6) {};
			\node [style=none] (2) at (-0, 2.6) {};
			\node [style=right label] (4) at (0, 2.4) {$g$};
			\node [style=empty circle,fill=white] (8) at (0,0) {$p_h$};
			\node [style=right label] (6) at (0, -2.2) {$h$};
              \node [style=small box,inner sep=1pt,fill=quantumviolet!30] (10) at (0,1.65) {$\kappa$};
              \node [style=right label] (11) at (0,1.1) {$h$};
              \node [style=right label] (12) at (0,-0.8) {$h$};
              \node [style=small box,inner sep=1pt,fill=quantumviolet!30] (13) at (0,-1.55) {$\iota$};
		\end{pgfonlayer}
		\begin{pgfonlayer}{edgelayer}
			\draw [qWire] (2) to (1.center);
		\end{pgfonlayer}
		\etp.
\eeq
The embedding is said to be \emph{faithful} if and only if both $\iota$ and $\kappa$ are faithful \footnote{One might consider changing the definition of embedding to also demand injectivity. For embeddings of GPT subsystems, as we will see, injectivity follows from the existence of an embedding, and need not be assumed. For fragments of GPTs, however, it does not follow automatically, and a strong motivation for not including injectivity in the definition of embedding is that an ontological model of a GPT fragment need not be injective, and we take existence of an ontological model to coincide with simplex-embeddability.}, is said to be \emph{full} if and only if both $\iota$ and $\kappa$ are full, and is said to be \emph{invertible} if and only if both $\iota$ and $\kappa$ are  invertible. 
\end{definition}

This generalizes the notion of an exact unital embedding of Ref.~\cite{muller2023testing} to apply to GPT fragments that are not standard GPTs; additionally, it generalizes the notion of simplex-embeddability of Ref.~\cite{SchmidGPT} to apply to embeddings into GPTs that are not simplicial.

For relatively tomographic fragments (i.e., for proper GPT systems), embeddings are necessarily faithful, as there is no way to introduce new operational identities while reproducing the empirical probabilities. For embeddings of fragments that are not relatively tomographic, however, both faithful and nonfaithful embeddings are possible. 
Embeddings are linear maps, so they always preserve operational identities. 
Faithful embeddings moreover do not introduce new operational identities, so that
\begin{align}
\sum_i\alpha_i{s_i}=\mathbf{0}\iff \sum_i\alpha_i\iota({s_i)}=\mathbf{0}\\
\sum_i\beta_i{e_i}=\mathbf{0}\iff \sum_i\beta_i\kappa({e_i)}=\mathbf{0}
\end{align}
Nonfaithful embeddings necessarily {\em do} introduce new operational identities, however. If it is the embedding map for states that fails to be injective, then one will have at least one pair of distinct states, $s\neq s'$, for which $\iota(s)=\iota(s')$.

We emphasize this point because these facts will be crucial for understanding shadow maps and their implications for assessing nonclassicality. As an extreme example, every fragment with a trivial probability rule 
(i.e., one which is constant on normalized states) is embeddable into a trivial one with only one state (the embedding can simply map every normalized state in the original fragment to this unique state).

Finally, the notion of equivalence for GPT fragments is the following  (analogous to that for GPT systems).

\begin{definition}[Equivalence of GPT fragments]\label{defn:GPTfragequivalence}
Two GPT fragments are equivalent (isomorphic) if and only if there is an invertible GPT fragment embedding (as defined in Definition~\ref{def:fragembeddings}) between them.
\end{definition}

Finally, note that a GPT fragment embedding necessarily takes normalized states to normalized states, as a consequence of probability preservation together with the preservation of the unit effect.
\bigskip 

In earlier works~\cite{selby2023accessible,selby2024linear}, two related concepts were introduced: {\em fragments} of a GPT and {\em accessible GPT fragments}. 
The distinction arises because the states and effects in a fragment of a given GPT need not span the vector space of the GPT, in which case one has a choice: one can either represent the processes in the fragment as living within the space they span, or one can represent them in the (larger) vector space of the original GPT. The former view was introduced in Ref.~\cite{selby2023accessible} under the name accessible GPT fragments, whereas the latter is simply called a fragment.

Because we have represented GPT state spaces and effects spaces in a way that is insensitive to the vector space in which they are represented, it follows that our approach (e.g., in Definition~\ref{def:GPTfragment})  subsumes both representations and treats them as equivalent (in the sense of Definition~\ref{defn:GPTfragequivalence}).

	\subsection{Nonclassicality as simplex-embeddability}

Assessing the noncontextuality~\cite{gencontext} of an operational prepare-measure scenario is equivalent to assessing whether or not the GPT representation of that scenario can be embedded into a simplicial GPT, as was shown in Ref.~\cite{SchmidGPT}.
Consequently, we can take simplex-embeddability as our notion of classical-explainability for a GPT. In doing so, we inherit all of the prior foundational motivations that have been given for generalized noncontextuality~\cite{gencontext,Leibniz,schmid2021guiding}.

To formalize this notion of classical explainability, we must first formalize the notion of a simplicial GPT system.

\begin{definition}[Simplicial GPT systems]
    A {\em simplicial GPT system} is defined by the tuple $\Lambda:=(\Delta,\Delta^*,p_\Lambda)$, where $\Delta$ is a state space which forms a simplex living in a finite-dimensional vector space $\mathds{R}^d$; $\Delta^*$ is the effect space which forms the logical dual of the simplex (that is, the set of vectors whose inner product with vectors in the simplex is between 0 and 1); and $p_\Lambda$ is the Euclidean inner product in $\mathds{R}^d$.
\end{definition}

Simplicial GPT systems represent strictly classical systems---those for which every state has a unique interpretation as a state of knowledge about some fundamental ontic states (which can be associated with the extremal states---the vertices of the simplex), and for which every logically possible measurement is allowed (including a single maximally informative measurement). They also embody the notion of classicality introduced by Leggett and Garg under the name macroscopic realism~\cite{LG,Schmid2024reviewreformulation}.

A given GPT system or GPT fragment is classically explainable if and only if it can be embedded into a simplicial GPT system.
\begin{definition}[Simplex embedding for  a GPT fragment] A GPT fragment $f$ admits of a simplex embedding if there exists a simplicial GPT system $\Lambda$ and a GPT fragment embedding $(\iota,\kappa):f\to \Lambda$. 
\end{definition} 

It follows from this definition that every fragment of a simplex-embeddable GPT system will also admit of a simplex embedding. On the other hand, impossibility of simplex embedding for a fragment necessarily implies impossibility of simplex embedding for the full GPT system of which the fragment is part.

This is the natural notion of classical explainability within the framework of GPTs.

	\section{GPT subsystems} \label{ncofcomp} 

The notion of GPT system embeddings leads to a natural and very general notion of {\em  subsystems} within the framework of GPTs. This aims to capture what it means for a GPT system to ``live inside'' another, or to be ``explainable by'' another---in both cases, with all its GPT structure intact.

\begin{definition}[GPT subsystems]
 \label{def:subsystem}
		A subsystem of a GPT system $G$ is a GPT system $F$, denoted $F\subseteq G$, such that there exists a GPT embedding (in the sense of Definition~\ref{def:fragembeddings}) of $F$ into $G$. 
\end{definition}

By definition, every GPT subsystem is itself a GPT system, and so constitutes a relatively tomographic fragment.

Two obvious but useful facts about GPT subsystems are the following, the first of which follows from transitivity of GPT embeddings:

	\begin{prop}
 \label{prop: transitivity}
		The notion of a GPT subsystem is transitive, so that $F\subseteq G$ and $G\subseteq H$ implies $F\subseteq H$.
	\end{prop} 
	
Second, from the definition of simplex-embedding, it follows that:

	\begin{prop}
A GPT system $F$ is simplex embeddable if and only if it is a GPT subsystem of a simplicial system; that is, if and only if $F\subseteq \Lambda$ for some simplicial GPT system $\Lambda$. 
	\end{prop}

This definition of a GPT subsystem is general enough to subsume a variety of interesting cases: 
\begin{enumerate}
\item a component $F$ of a composite system $F\otimes G$ is a GPT subsystem of $F\otimes G$,
\item a component $F$ of a direct sum system $F\oplus G$ is a GPT subsystem of $F\oplus G$,
\item virtual quantum systems (so that, e.g., a logical qubit is a GPT subsystem of the physical qubits it is encoded in),
\item stabilizer quantum systems~\cite{gottesman1997stabilizer,gottesman1998heisenberg} are GPT subsystems of quantum systems of the same dimension (so that, e.g., a stabilizer qubit is a GPT subsystem of a qubit),
\item restricting the states and effects on a GPT system to those living in a linear subspace defines a GPT subsystem (so that, e.g., a rebit~\cite{stueckelberg1960quantum,caves2001entanglement} system is a GPT subsystem of a qubit),
\item the set of all processes that are the image of a decoherence process on a GPT system (so that, e.g., a classical bit is a GPT subsystem of a quantum bit).
\item any fragment of a GPT system $G$ which is itself a valid GPT system constitutes a GPT subsystem of $G$. 
\end{enumerate}

We expand on these examples and prove that they are special cases of GPT subsystems below. Many of these examples are pictured in Figure~\ref{subsysExs}.

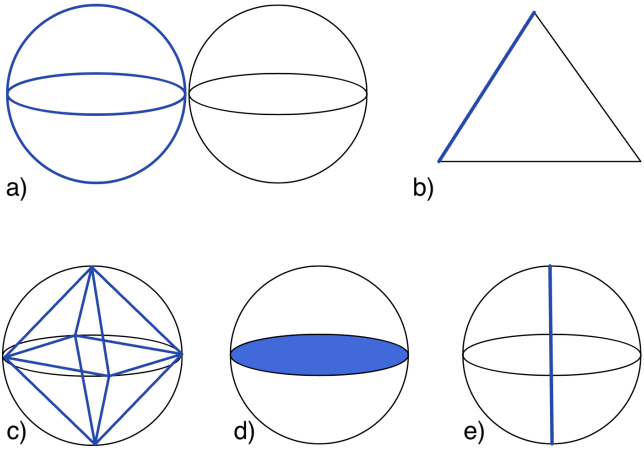
\begin{figure*}[htb!]
   \centering

\tikzset{every picture/.style={line width=0.5pt}} %set default line width to 0.75pt        

\begin{tikzpicture}[x=0.75pt,y=0.75pt,yscale=-1,xscale=1]
	%uncomment if require: \path (0,300); %set diagram left start at 0, and has height of 300
	
	%Shape: Ellipse [id:dp5251685626939022] 
	\draw   (262.13,54.89) .. controls (262.13,34.96) and (278.34,18.8) .. (298.33,18.8) .. controls (318.33,18.8) and (334.54,34.96) .. (334.54,54.89) .. controls (334.54,74.82) and (318.33,90.97) .. (298.33,90.97) .. controls (278.34,90.97) and (262.13,74.82) .. (262.13,54.89) -- cycle ;
	%Curve Lines [id:da24398327506478423] 
	\draw  [dash pattern={on 0.84pt off 2.51pt}]  (262.13,54.89) .. controls (264.7,42.18) and (330.94,41.16) .. (334.54,54.89) ;
	%Curve Lines [id:da2659187093165445] 
	\draw    (262.13,54.89) .. controls (265.21,67.86) and (331.97,69.4) .. (334.54,54.89) ;
	%Shape: Ellipse [id:dp2164174157162806] 
	\draw  [fill={rgb, 255:red, 74; green, 144; blue, 226 }  ,fill opacity=0.75 ] (185.61,55.91) .. controls (185.61,35.98) and (201.82,19.83) .. (221.81,19.83) .. controls (241.81,19.83) and (258.02,35.98) .. (258.02,55.91) .. controls (258.02,75.84) and (241.81,92) .. (221.81,92) .. controls (201.82,92) and (185.61,75.84) .. (185.61,55.91) -- cycle ;
	%Curve Lines [id:da736673035718332] 
	\draw  [dash pattern={on 0.84pt off 2.51pt}]  (185.61,55.91) .. controls (188.18,43.21) and (254.42,42.18) .. (258.02,55.91) ;
	%Curve Lines [id:da386458162060341] 
	\draw    (185.61,55.91) .. controls (188.69,68.89) and (255.45,70.43) .. (258.02,55.91) ;
	%Shape: Ellipse [id:dp38472902753127924] 
	\draw  [fill={rgb, 255:red, 74; green, 144; blue, 226 }  ,fill opacity=0.75 ] (63.33,55.91) .. controls (63.33,35.98) and (79.54,19.83) .. (99.54,19.83) .. controls (119.54,19.83) and (135.74,35.98) .. (135.74,55.91) .. controls (135.74,75.84) and (119.54,92) .. (99.54,92) .. controls (79.54,92) and (63.33,75.84) .. (63.33,55.91) -- cycle ;
	%Curve Lines [id:da8244415163027707] 
	\draw  [dash pattern={on 0.84pt off 2.51pt}]  (63.33,55.91) .. controls (65.9,43.21) and (132.15,42.18) .. (135.74,55.91) ;
	%Curve Lines [id:da7590840398916963] 
	\draw    (63.33,55.91) .. controls (66.41,68.89) and (133.18,70.43) .. (135.74,55.91) ;
	%Curve Lines [id:da6753854197284568] 
	\draw [line width=0.75]    (157.04,61.01) .. controls (149.74,61.25) and (149.73,53.21) .. (157.04,53.95) ;
	%Straight Lines [id:da6473639889363536] 
	\draw [line width=0.75]    (157.04,61.01) -- (169.65,61.02) ;
	\draw [shift={(171.65,61.02)}, rotate = 180.03] [color={rgb, 255:red, 0; green, 0; blue, 0 }  ][line width=0.75]    (10.93,-3.29) .. controls (6.95,-1.4) and (3.31,-0.3) .. (0,0) .. controls (3.31,0.3) and (6.95,1.4) .. (10.93,3.29)   ;
	
	%Shape: Rectangle [id:dp4188770791018548] 
	\draw  [draw opacity=0][fill={rgb, 255:red, 74; green, 144; blue, 226 }  ,fill opacity=0.75 ] (184.47,127) -- (220.19,162.02) -- (183.86,199.06) -- (148.15,164.04) -- cycle ;
	%Straight Lines [id:da2837992806933154] 
	\draw [fill={rgb, 255:red, 74; green, 144; blue, 226 }  ,fill opacity=0.75 ] [dash pattern={on 0.84pt off 2.51pt}]  (207.26,154.55) -- (220.33,163.03) ;
	%Straight Lines [id:da5053617520095802] 
	\draw [fill={rgb, 255:red, 74; green, 144; blue, 226 }  ,fill opacity=0.75 ]   (148,163.03) -- (161.07,170.7) ;
	%Straight Lines [id:da519026553794865] 
	\draw [fill={rgb, 255:red, 74; green, 144; blue, 226 }  ,fill opacity=0.75 ]   (161.07,170.7) -- (220.33,163.03) ;
	%Straight Lines [id:da7676201765485525] 
	\draw [fill={rgb, 255:red, 74; green, 144; blue, 226 }  ,fill opacity=0.75 ] [dash pattern={on 0.84pt off 2.51pt}]  (148,163.03) -- (207.26,154.55) ;
	%Straight Lines [id:da01243415635496703] 
	\draw [fill={rgb, 255:red, 74; green, 144; blue, 226 }  ,fill opacity=0.75 ]   (161.07,170.7) -- (184.47,127) ;
	%Straight Lines [id:da3478705429630219] 
	\draw [fill={rgb, 255:red, 74; green, 144; blue, 226 }  ,fill opacity=0.75 ] [dash pattern={on 0.84pt off 2.51pt}]  (183.86,199.06) -- (207.26,154.55) ;
	%Straight Lines [id:da18853848166471043] 
	\draw [fill={rgb, 255:red, 74; green, 144; blue, 226 }  ,fill opacity=0.75 ]   (161.07,170.7) -- (183.86,199.06) ;
	%Straight Lines [id:da9680538916005651] 
	\draw [fill={rgb, 255:red, 74; green, 144; blue, 226 }  ,fill opacity=0.75 ]   (148,163.03) -- (184.47,127) ;
	%Straight Lines [id:da19058405950649115] 
	\draw [fill={rgb, 255:red, 74; green, 144; blue, 226 }  ,fill opacity=0.75 ]   (184.47,127) -- (220.33,163.03) ;
	%Straight Lines [id:da0024267172104178725] 
	\draw [fill={rgb, 255:red, 74; green, 144; blue, 226 }  ,fill opacity=0.75 ] [dash pattern={on 0.84pt off 2.51pt}]  (184.47,127) -- (207.26,154.55) ;
	%Straight Lines [id:da47398713650583235] 
	\draw [fill={rgb, 255:red, 74; green, 144; blue, 226 }  ,fill opacity=0.75 ]   (220.33,163.03) -- (183.86,199.06) ;
	%Straight Lines [id:da9169913024961485] 
	\draw [fill={rgb, 255:red, 74; green, 144; blue, 226 }  ,fill opacity=0.75 ]   (148,163.03) -- (183.86,199.06) ;
	
	%Shape: Ellipse [id:dp6904925474003195] 
	\draw   (148.13,162.89) .. controls (148.13,142.96) and (164.34,126.8) .. (184.33,126.8) .. controls (204.33,126.8) and (220.54,142.96) .. (220.54,162.89) .. controls (220.54,182.82) and (204.33,198.97) .. (184.33,198.97) .. controls (164.34,198.97) and (148.13,182.82) .. (148.13,162.89) -- cycle ;
	%Curve Lines [id:da9349637600863591] 
	\draw  [dash pattern={on 0.84pt off 2.51pt}]  (148.13,162.89) .. controls (150.7,150.18) and (216.94,149.16) .. (220.54,162.89) ;
	%Curve Lines [id:da7525177743390872] 
	\draw    (148.13,162.89) .. controls (151.21,175.86) and (217.97,177.4) .. (220.54,162.89) ;
	%Shape: Rectangle [id:dp47362648330723045] 
	\draw  [draw opacity=0][fill={rgb, 255:red, 74; green, 144; blue, 226 }  ,fill opacity=0.75 ] (62.47,126) -- (98.19,161.02) -- (61.86,198.06) -- (26.15,163.04) -- cycle ;
	%Straight Lines [id:da5432741959865931] 
	\draw [fill={rgb, 255:red, 74; green, 144; blue, 226 }  ,fill opacity=0.75 ] [dash pattern={on 0.84pt off 2.51pt}]  (85.26,153.55) -- (98.33,162.03) ;
	%Straight Lines [id:da9848249024773028] 
	\draw [fill={rgb, 255:red, 74; green, 144; blue, 226 }  ,fill opacity=0.75 ]   (26,162.03) -- (39.07,169.7) ;
	%Straight Lines [id:da35778464964554824] 
	\draw [fill={rgb, 255:red, 74; green, 144; blue, 226 }  ,fill opacity=0.75 ]   (39.07,169.7) -- (98.33,162.03) ;
	%Straight Lines [id:da32676159784805126] 
	\draw [fill={rgb, 255:red, 74; green, 144; blue, 226 }  ,fill opacity=0.75 ] [dash pattern={on 0.84pt off 2.51pt}]  (26,162.03) -- (85.26,153.55) ;
	%Straight Lines [id:da6704533861842703] 
	\draw [fill={rgb, 255:red, 74; green, 144; blue, 226 }  ,fill opacity=0.75 ]   (39.07,169.7) -- (62.47,126) ;
	%Straight Lines [id:da37449701496193044] 
	\draw [fill={rgb, 255:red, 74; green, 144; blue, 226 }  ,fill opacity=0.75 ] [dash pattern={on 0.84pt off 2.51pt}]  (61.86,198.06) -- (85.26,153.55) ;
	%Straight Lines [id:da8748240673628406] 
	\draw [fill={rgb, 255:red, 74; green, 144; blue, 226 }  ,fill opacity=0.75 ]   (39.07,169.7) -- (61.86,198.06) ;
	%Straight Lines [id:da12298841743741218] 
	\draw [fill={rgb, 255:red, 74; green, 144; blue, 226 }  ,fill opacity=0.75 ]   (26,162.03) -- (62.47,126) ;
	%Straight Lines [id:da5638690371152644] 
	\draw [fill={rgb, 255:red, 74; green, 144; blue, 226 }  ,fill opacity=0.75 ]   (62.47,126) -- (98.33,162.03) ;
	%Straight Lines [id:da9070964339528934] 
	\draw [fill={rgb, 255:red, 74; green, 144; blue, 226 }  ,fill opacity=0.75 ] [dash pattern={on 0.84pt off 2.51pt}]  (62.47,126) -- (85.26,153.55) ;
	%Straight Lines [id:da17682532461652145] 
	\draw [fill={rgb, 255:red, 74; green, 144; blue, 226 }  ,fill opacity=0.75 ]   (98.33,162.03) -- (61.86,198.06) ;
	%Straight Lines [id:da6264716857458844] 
	\draw [fill={rgb, 255:red, 74; green, 144; blue, 226 }  ,fill opacity=0.75 ]   (26,162.03) -- (61.86,198.06) ;
	
	%Curve Lines [id:da18125282127354958] 
	\draw [line width=0.75]    (120.04,168.01) .. controls (112.74,168.25) and (112.73,160.21) .. (120.04,160.95) ;
	%Straight Lines [id:da9776722476038412] 
	\draw [line width=0.75]    (120.04,168.01) -- (132.65,168.02) ;
	\draw [shift={(134.65,168.02)}, rotate = 180.03] [color={rgb, 255:red, 0; green, 0; blue, 0 }  ][line width=0.75]    (10.93,-3.29) .. controls (6.95,-1.4) and (3.31,-0.3) .. (0,0) .. controls (3.31,0.3) and (6.95,1.4) .. (10.93,3.29)   ;
	
	%Shape: Ellipse [id:dp24941929541913388] 
	\draw   (368.13,162.89) .. controls (368.13,142.96) and (384.34,126.8) .. (404.33,126.8) .. controls (424.33,126.8) and (440.54,142.96) .. (440.54,162.89) .. controls (440.54,182.82) and (424.33,198.97) .. (404.33,198.97) .. controls (384.34,198.97) and (368.13,182.82) .. (368.13,162.89) -- cycle ;
	%Curve Lines [id:da7852577027996687] 
	\draw [fill={rgb, 255:red, 74; green, 144; blue, 226 }  ,fill opacity=0.75 ] [dash pattern={on 0.84pt off 2.51pt}]  (368.13,162.89) .. controls (370.7,150.18) and (436.94,149.16) .. (440.54,162.89) ;
	%Curve Lines [id:da1625869519274643] 
	\draw [fill={rgb, 255:red, 74; green, 144; blue, 226 }  ,fill opacity=0.75 ]   (368.13,162.89) .. controls (371.21,175.86) and (437.97,177.4) .. (440.54,162.89) ;
	%Curve Lines [id:da9591095686986436] 
	\draw [fill={rgb, 255:red, 74; green, 144; blue, 226 }  ,fill opacity=0.75 ]   (248.13,162.89) .. controls (251.21,175.86) and (317.97,177.4) .. (320.54,162.89) ;
	%Curve Lines [id:da27638513041294077] 
	\draw [fill={rgb, 255:red, 74; green, 144; blue, 226 }  ,fill opacity=0.75 ]   (248.13,162.89) .. controls (250.7,150.18) and (316.94,149.16) .. (320.54,162.89) ;
	%Curve Lines [id:da2918777243029076] 
	\draw [line width=0.75]    (339.04,167.01) .. controls (331.74,167.25) and (331.73,159.21) .. (339.04,159.95) ;
	%Straight Lines [id:da3034526198687604] 
	\draw [line width=0.75]    (339.04,167.01) -- (351.65,167.02) ;
	\draw [shift={(353.65,167.02)}, rotate = 180.03] [color={rgb, 255:red, 0; green, 0; blue, 0 }  ][line width=0.75]    (10.93,-3.29) .. controls (6.95,-1.4) and (3.31,-0.3) .. (0,0) .. controls (3.31,0.3) and (6.95,1.4) .. (10.93,3.29)   ;
	
	%Straight Lines [id:da5604199432305889] 
	\draw [color={rgb, 255:red, 74; green, 144; blue, 226 }  ,draw opacity=1 ][line width=1.5]    (494.33,126.8) -- (494.33,198.97) ;
	%Shape: Circle [id:dp19101229597421066] 
	\draw  [fill={rgb, 255:red, 74; green, 144; blue, 226 }  ,fill opacity=1 ] (491.33,198.97) .. controls (491.33,197.32) and (492.68,195.97) .. (494.33,195.97) .. controls (495.99,195.97) and (497.33,197.32) .. (497.33,198.97) .. controls (497.33,200.63) and (495.99,201.97) .. (494.33,201.97) .. controls (492.68,201.97) and (491.33,200.63) .. (491.33,198.97) -- cycle ;
	%Shape: Circle [id:dp9523629928670005] 
	\draw  [fill={rgb, 255:red, 74; green, 144; blue, 226 }  ,fill opacity=1 ] (491.33,128.97) .. controls (491.33,127.32) and (492.68,125.97) .. (494.33,125.97) .. controls (495.99,125.97) and (497.33,127.32) .. (497.33,128.97) .. controls (497.33,130.63) and (495.99,131.97) .. (494.33,131.97) .. controls (492.68,131.97) and (491.33,130.63) .. (491.33,128.97) -- cycle ;
	%Shape: Ellipse [id:dp32291804588382267] 
	\draw   (568.13,162.89) .. controls (568.13,142.96) and (584.34,126.8) .. (604.33,126.8) .. controls (624.33,126.8) and (640.54,142.96) .. (640.54,162.89) .. controls (640.54,182.82) and (624.33,198.97) .. (604.33,198.97) .. controls (584.34,198.97) and (568.13,182.82) .. (568.13,162.89) -- cycle ;
	%Curve Lines [id:da41059637606288346] 
	\draw  [dash pattern={on 0.84pt off 2.51pt}]  (568.13,162.89) .. controls (570.7,150.18) and (636.94,149.16) .. (640.54,162.89) ;
	%Curve Lines [id:da6825453737061422] 
	\draw    (568.13,162.89) .. controls (571.21,175.86) and (637.97,177.4) .. (640.54,162.89) ;
	%Curve Lines [id:da3607470153531829] 
	\draw [line width=0.75]    (532.04,168.01) .. controls (524.74,168.25) and (524.73,160.21) .. (532.04,160.95) ;
	%Straight Lines [id:da5144624205353631] 
	\draw [line width=0.75]    (532.04,168.01) -- (544.65,168.02) ;
	\draw [shift={(546.65,168.02)}, rotate = 180.03] [color={rgb, 255:red, 0; green, 0; blue, 0 }  ][line width=0.75]    (10.93,-3.29) .. controls (6.95,-1.4) and (3.31,-0.3) .. (0,0) .. controls (3.31,0.3) and (6.95,1.4) .. (10.93,3.29)   ;
	
	%Straight Lines [id:da13624939153705506] 
	\draw [color={rgb, 255:red, 74; green, 144; blue, 226 }  ,draw opacity=1 ][line width=1.5]    (604.33,127.97) -- (604.33,200.14) ;
	%Shape: Circle [id:dp6854931840287375] 
	\draw  [fill={rgb, 255:red, 74; green, 144; blue, 226 }  ,fill opacity=1 ] (601.33,198.97) .. controls (601.33,197.32) and (602.68,195.97) .. (604.33,195.97) .. controls (605.99,195.97) and (607.33,197.32) .. (607.33,198.97) .. controls (607.33,200.63) and (605.99,201.97) .. (604.33,201.97) .. controls (602.68,201.97) and (601.33,200.63) .. (601.33,198.97) -- cycle ;
	%Shape: Circle [id:dp7734252551494067] 
	\draw  [fill={rgb, 255:red, 74; green, 144; blue, 226 }  ,fill opacity=1 ] (601.33,127.97) .. controls (601.33,126.32) and (602.68,124.97) .. (604.33,124.97) .. controls (605.99,124.97) and (607.33,126.32) .. (607.33,127.97) .. controls (607.33,129.63) and (605.99,130.97) .. (604.33,130.97) .. controls (602.68,130.97) and (601.33,129.63) .. (601.33,127.97) -- cycle ;
	%Straight Lines [id:da34807327700563506] 
	\draw [color={rgb, 255:red, 74; green, 144; blue, 226 }  ,draw opacity=1 ][line width=1.5]    (438.33,18.8) -- (438.33,90.97) ;
	%Shape: Circle [id:dp3803196356257841] 
	\draw  [fill={rgb, 255:red, 74; green, 144; blue, 226 }  ,fill opacity=1 ] (435.33,90.97) .. controls (435.33,89.32) and (436.68,87.97) .. (438.33,87.97) .. controls (439.99,87.97) and (441.33,89.32) .. (441.33,90.97) .. controls (441.33,92.63) and (439.99,93.97) .. (438.33,93.97) .. controls (436.68,93.97) and (435.33,92.63) .. (435.33,90.97) -- cycle ;
	%Shape: Circle [id:dp7766240828627665] 
	\draw  [fill={rgb, 255:red, 74; green, 144; blue, 226 }  ,fill opacity=1 ] (435.33,20.97) .. controls (435.33,19.32) and (436.68,17.97) .. (438.33,17.97) .. controls (439.99,17.97) and (441.33,19.32) .. (441.33,20.97) .. controls (441.33,22.63) and (439.99,23.97) .. (438.33,23.97) .. controls (436.68,23.97) and (435.33,22.63) .. (435.33,20.97) -- cycle ;
	%Shape: Triangle [id:dp07550831216718568] 
	\draw   (528.33,20.97) -- (589.53,51.46) -- (528.36,89.34) -- cycle ;
	%Straight Lines [id:da48337004956736374] 
	\draw [color={rgb, 255:red, 74; green, 144; blue, 226 }  ,draw opacity=1 ][line width=1.5]    (528.33,18.8) -- (528.33,90.97) ;
	%Shape: Circle [id:dp6107889832027505] 
	\draw  [fill={rgb, 255:red, 74; green, 144; blue, 226 }  ,fill opacity=1 ] (525.33,90.97) .. controls (525.33,89.32) and (526.68,87.97) .. (528.33,87.97) .. controls (529.99,87.97) and (531.33,89.32) .. (531.33,90.97) .. controls (531.33,92.63) and (529.99,93.97) .. (528.33,93.97) .. controls (526.68,93.97) and (525.33,92.63) .. (525.33,90.97) -- cycle ;
	%Shape: Circle [id:dp5522779722249143] 
	\draw  [fill={rgb, 255:red, 74; green, 144; blue, 226 }  ,fill opacity=1 ] (525.33,20.97) .. controls (525.33,19.32) and (526.68,17.97) .. (528.33,17.97) .. controls (529.99,17.97) and (531.33,19.32) .. (531.33,20.97) .. controls (531.33,22.63) and (529.99,23.97) .. (528.33,23.97) .. controls (526.68,23.97) and (525.33,22.63) .. (525.33,20.97) -- cycle ;
	%Curve Lines [id:da37853275550390375] 
	\draw [line width=0.75]    (475.04,59.01) .. controls (467.74,59.25) and (467.73,51.21) .. (475.04,51.95) ;
	%Straight Lines [id:da6073334382048379] 
	\draw [line width=0.75]    (475.04,59.01) -- (487.65,59.02) ;
	\draw [shift={(489.65,59.02)}, rotate = 180.03] [color={rgb, 255:red, 0; green, 0; blue, 0 }  ][line width=0.75]    (10.93,-3.29) .. controls (6.95,-1.4) and (3.31,-0.3) .. (0,0) .. controls (3.31,0.3) and (6.95,1.4) .. (10.93,3.29)   ;

	% Text Node
	\draw (47,8) node [anchor=north west][inner sep=0.75pt]   [align=left] {a)};
	% Text Node
	\draw (407,8) node [anchor=north west][inner sep=0.75pt]   [align=left] {b)};
	% Text Node
	\draw (17,118) node [anchor=north west][inner sep=0.75pt]   [align=left] {c)};
	% Text Node
	\draw (242,118) node [anchor=north west][inner sep=0.75pt]   [align=left] {d)};
	% Text Node
	\draw (462,118) node [anchor=north west][inner sep=0.75pt]   [align=left] {e)};

\end{tikzpicture}
\bigskip
   \caption{Five examples of GPT subsystems (in blue) of larger systems (in black). a) One qubit is a GPT subsystem of two qubits. b) A two-level classical system is a GPT subsystem of a three-level classical system. c) The convex hull of the stabilizer states and effects forms a GPT subsystem of a qubit. d) The rebit is a GPT subsystem of a qubit. e) A classical bit is a GPT subsystem of a qubit.   }
   \label{subsysExs}
\end{figure*}

One might wonder if the notion we have defined is {\em too} general, so much so that the terminology of `subsystem' might be inappropriate. For instance, only examples 1, 2, 3, and (arguably) 6 above are conventionally considered subsystems. In the context of quantum theory, for example, the notion of a subsystem typically comes with algebraic structure. However, in the context of GPTs, a system generally has no structure beyond convex structure, and so it is reasonable to require that a ``GPT subsystem'' require only linear embeddability but nothing specific beyond that (such as preservation of any particular algebraic structure that may be present in special cases, like in quantum theory). 
As an example (that we return to later in this section), the rebit is typically not considered to be a {\em quantum} subsystem of a qubit, but it makes sense to consider it a {\em GPT} subsystem---and according to our definition, it is.

This notion of GPT subsystem is also motivated by its usefulness in the study of noncontextuality, as we will see in later sections.
\bigskip 

\textbf{Example 1.} To see that a component GPT system $F$ of a given composite GPT system $F\otimes G$ is always a GPT subsystem of $F\otimes G$, it suffices to note that the composition operation in any GPT, $\otimes$, is a bilinear map on states $\otimes:\Omega_F\times\Omega_G \to \Omega_{F\otimes G}$ and on effects $\otimes:\mathcal{E}_F\times\mathcal{E}_G\to \mathcal{E}_{F\otimes G}$, where $\otimes::(u_F,u_G)\mapsto u_F\otimes u_G = u_{F\otimes G}$, and where $(e_F\otimes e_G)[(s_F\otimes s_b)] = e_F[s_F]e_G[s_G]$. (In this and the other examples in this section, we view effects as linear functionals on states.) These  properties hold for the composition rule in arbitrary GPTs, including quantum theory, the minimal and maximal tensor products~\cite{barnum2011information}, composition rules between the minimal and maximal tensor products, and even for composition rules in tomographically nonlocal GPTs. So as we now prove, it follows that $F$ is always a GPT subsystem of $F\otimes G$.
One can define a GPT embedding of $F$ into $F\otimes G$ given by the linear map $\iota : \Omega_F \to \Omega_{FG} :: s \mapsto s\otimes s'$ for any fixed normalized $s'\in \bar{\Omega}_G$, and the linear map $\kappa:\mathcal{E}_F\to \mathcal{E}_{FG}:: e \mapsto e\otimes u_G$. These maps preserve the probabilistic predictions, since they satisfy ${e[s] = \kappa(e)[\iota(s)]}$ as  ${\kappa(e)[\iota(s)]= e\otimes u[s\otimes s']=e[s]u[s']=e[s]}$. They also act on the unit effect appropriately, since they satisfy $\kappa(u_F)=u_{FG}$, as $\kappa(u_F)=u_F\otimes u_G = u_{FG}$.
\bigskip

\textbf{Example 2.} Similarly, we now show that a direct sum factor $F$ of a given GPT system $F\oplus G$ is indeed a GPT subsystem of $F\oplus G$. 
The properties of the $\oplus$ operation that we need are merely that ${(e_F\oplus e_G)[(s_F\oplus s_G)]=e_F(s_F)+ e_G(s_G)}$
and $u_F\oplus u_G = u_{F\oplus G}$
for effects.
Hence, we can define a GPT embedding of $F$ into $F\oplus G$ given by the linear map $\iota : \Omega_F\to \Omega_{F\oplus G} :: s \mapsto s\oplus 0$ and the linear map $\kappa : \mathcal{E}_F \to \mathcal{E}_{F\oplus G} :: e \mapsto e\oplus u_G$. This is a valid embedding, since these  properties of the direct sum give $e[s] = \kappa(e)[\iota(s)]$ as  $\kappa(e)[\iota(s)]= e\oplus u[s\oplus 0]=e[s]+u[0]=e[s]$ and $\kappa[u_F]=u_F\oplus u_G = u_{F\oplus G}$.

That a tensor factor of some larger system is a subsystem is completely standard. That components of a direct sum system should be considered subsystems is also standard. The direct sum structure $F\oplus G$ entails that the system is either described by a state in $F$ {\em or} by a state in $G$. As $F$ forms a space of states, each of which is a possible description of the system $F\oplus G$, it makes sense to consider $F$ as a subsystem of $F\oplus G$. Consider for example a classical $d$-level system and a classical $D$-level system with $D> d$. A $d$-level system can be viewed as living inside the $D$-level system simply by ensuring that only $d$ of the possible states are ever accessed in one's experiment. Since classical systems can always be decomposed as direct sums, $\Delta_D=\Delta_d\oplus\Delta_{D-d}$, this is an instance of $F$ (the $d$-level system) being a subsystem of $F \oplus G$ (the $D$-level system). Note that the direct sum structure also arises in superselected quantum systems~\cite{bartlett2007reference}, where the components are not necessarily classical.
\bigskip

\textbf{Example 3.} That our notion of a GPT subsystem also deems virtual quantum systems~\cite{zanardi2001virtual} to be subsystems is quite immediate, as these are isomorphic to any other quantum systems of the same dimension, and so trivially admit a linear embedding in our sense.
So, our notion subsumes coarse-grained degrees of freedom encoded in some (potentially highly delocalized) larger set of fundamental systems. This includes quantum codes~\cite{Schumacher1995},  decoherence-free subsystems~\cite{Lidar2003}, and time-delocalized subsystems~\cite{Oreshkov2019timedelocalized}, as these are all virtual subsystems~\cite{zanardi2001virtual}. Consider for example the simple quantum code that encodes a logical qubit into three physical qubits via the linear map taking  $\ket{0} \rightarrow \ket{0}\otimes\ket{0}\otimes\ket{0}$ and $\ket{1} \rightarrow \ket{1}\otimes\ket{1}\otimes\ket{1}$. This defines a qubit on the Hilbert space defined by $\mathsf{Span}[\ket{0}\otimes\ket{0}\otimes\ket{0},\ket{1}\otimes\ket{1}\otimes\ket{1}]$. This qubit is a GPT subsystem of the three-qubit Hilbert space $\mathbb{C}^8$. Indeed, the state map for the GPT subsystem embedding can simply be taken to be the encoding map itself, and the effect map can be taken to be the decoding map (acting contravariantly on effects).
\bigskip

\textbf{Example 4.} Another example of a GPT subsystem arises when considering the stabilizer subtheory of quantum theory. Consider a stabilizer bit~\cite{gottesman1997stabilizer,gottesman1998heisenberg}, defined as the convex hull of the six eigenstates of Pauli operators on a qubit, together with the convex hull of the six effects that are projectors onto these (together with null and unit effects). According to our definition, the stabilizer qubit is a GPT subsystem of a standard quantum bit. Similarly, every stabilizer system of dimension $d$ is a GPT subsystem of every quantum system of dimension $d$.
\bigskip

\textbf{Example 5.} An example of a GPT subsystem arises when considering real quantum theory. Consider a real bit, or rebit~\cite{stueckelberg1960quantum,caves2001entanglement}, defined by the set of density matrices and POVM elements which can be written as linear combinations of the Pauli $I,X,$ and $Z$ operators (with no component of the Pauli $Y$ operator). 
The rebit is a GPT subsystem of a qubit. To see this, consider the natural representation of a given state (or effect) $O$ as a 3-dimensional real-valued vector 
$(1,\mathsf{Tr}[OX]),\mathsf{Tr}[OZ])$; then, one can embed the states and effects of the rebit back into the full qubit simply by associating to each the Bloch vector $(1,\mathsf{Tr}[OX]),0,\mathsf{Tr}[OZ])$. This embedding clearly satisfies all the necessary properties.
Similarly, any $d$-dimensional system in real quantum theory is a GPT subtheory of any $d$-dimensional system in quantum theory. 
\bigskip

\textbf{Example 6.} Another example of a GPT subsystem is that of a decohered system, which can naturally be taken to live within the fundamental (pre-decohered) system, and consequently viewed as a subsystem thereof. The notion of decoherence (relative to some fixed basis) in quantum theory is well established, and corresponds to destroying all off-diagonal terms (relative to that basis) in the matrix representation of a state or effect. Within the context of general GPTs, one can generalize this idea to define hyperdecoherence processes~\cite{zyczkowski2008quartic,Selby2017Entanglement,lee2018no,hefford2020hyper}, which correspond to idempotent maps taking GPT processes in the fundamental theory to GPT processes in the decohered theory. The fact that hyperdcoherence processes always leads to GPT subsystems of the fundamental theory follows from Ref.~\cite{muller2023testing}, which proved that there is always a linear embedding from a hyperdecohered system into the fundamental system that preserves the probabilities. We will expand on all of this in Section~\ref{beyondqt}, where we also draw out the important consequences of this fact for the study of noncontextuality.
\bigskip

\textbf{Example 7.} Finally, any fragment of a GPT $G$ which satisfies all the mathematical properties of a GPT constitutes a GPT subsystem of $G$, since the one can always take the identity embedding of a relatively tomographic GPT fragment into itself as a valid GPT embedding.

\subsection{Two GPT systems which are subsystems of each other are equivalent}

Any good notion of subsystem should satisfy the condition that if two systems are subsystems of each other, then they are the same system (at least up to isomorphism). Here we show that our notion of subsystem does indeed satisfy this property. 

 For the proof of the main proposition in this Section, we will need the following Lemma:
\begin{restatable}{lemma}{AutomorphismsAreIso}\label{lem:AutomorphismsAreIso}
Any embedding $(\iota,\kappa)$ of a GPT system $(\Omega,\mathcal{E},p)$ into itself is
    a GPT isomorphism. 
\end{restatable} 

The proof of this is in Appendix~\ref{sec:lemmaproof}   

We can now prove the main result of this section. 

\begin{prop}\label{equivalenceAssubsystems}
        Two GPT systems are equivalent if and only if they are GPT subsystems of each other.
    \end{prop}
  
\proof 
Suppose that the GPT systems ${F:=(\Omega_F,\mathcal{E}_F, p_F)}$ and ${G:=(\Omega_G,\mathcal{E}_G, p_G)}$ are equivalent. Then, by definition, we have a GPT isomorphism which embeds $F$ into $G$ and $G$ into $F$. Therefore $F$ is a GPT subsystem of $G$ and $G$ is a GPT subsystem of $F$.

Now, let us prove the converse. That $F$ is a GPT subsystem of $G$ means there is a (faithful) embedding $(\iota_1,\kappa_1): F\to G$, and that $G$ is a subsystem of $F$ means that there is a (faithful) embedding $(\iota_2,\kappa_2): G\to F$. We want to show that these facts give us an isomorphism between the two GPT systems.
    
We know that the $\iota_i$ and $\kappa_i$ are injective, since GPT embeddings are necessarily faithful. What remains to be shown is that $\iota_i$ and $\kappa_i$ are surjective state and effect maps, respectively. 

To see this, consider the composite embedding of $F$ into itself given by $(\iota_2\circ\iota_1,\kappa_2\circ\kappa_1):F\to F$. By Lemma~\ref{lem:AutomorphismsAreIso}, this self-embedding is in fact an isomorphism. In particular, this tells us that $\iota_2\circ\iota_1$ and $\kappa_2\circ\kappa_1$ are surjective, which implies that $\iota_2$ and $\kappa_2$ are surjective. Similarly, we can consider the composite embedding of $G$ into itself to show that $\iota_1$ and $\kappa_1$ are surjective. 

It immediately follows that $(\iota_1,\kappa_1)$ and $(\iota_2,\kappa_2)$ are (possibly distinct) isomorphisms between $F$ and $G$. Hence, $F$ and $G$ are equivalent GPT systems.
\endproof

\section{Shadows of GPT fragments}\label{se:shadow}

For any particular experiment on a given system, there is some true GPT describing the system, whether or not one knows what this GPT is. The laboratory operations one implements in the experiment will then be represented by a fragment of this true GPT. The states (effects) in this fragment will generally not be tomographically complete for the true GPT system.

If a given pair of distinct states (effects) in one's experiment cannot be discriminated by any of the effects (states) in the experiment, then those states (effects) will appear to be exactly identical within the context of the given experiment. In this case, the states (effects) will appear to exhibit operational identities that are not genuine---i.e., that do not hold in the true GPT.
If one blindly trusts these apparent operational equivalences, then the GPT description one would give to the experiment will be incorrect---it will in fact be given by a particular informationally lossy function of the true GPT fragment describing the scenario. We will call this informationally lossy map the {\em shadow map}. The image of this map on the fragment in one's experiment captures how the GPT states and effects {\em appear} relative to the effects and states (respectively) one has implemented in one's experiment---although, as just described, this apparent characterization will necessarily be incorrect if the states and/or effects are not tomographic for each other. 

If the states and effects are tomographic for each other, then this kind of mischaracterization cannot occur. We say that such states and effects are {\em relatively tomographic}, or simply that the fragment containing them is relatively tomographic.

Before formalizing these ideas fully, we present an illustrative example.

\subsection{Example 1: a failure of relative tomographic completeness} \label{example1}

To get some intuition for what happens when the states and effects in one's experiment are not relatively tomographic, consider a first simple example.

Imagine that one's experimental apparatus can prepare all quantum states, but that one's measurement devices can only measure the Pauli $Z$ observable. The set of states one can access in such an experiment is the full Bloch ball, while the set of effects contains only $\ket{0}\bra{0}$ and $\ket{1}\bra{1}$, as well as the effects in the convex hull of these two together with the null effect and unit effect (as these can be obtained by post-processing). This fragment of states and effects are represented in the Bloch sphere in Figure~\ref{fig:exSimpP}. 

\begin{figure}[htb!]
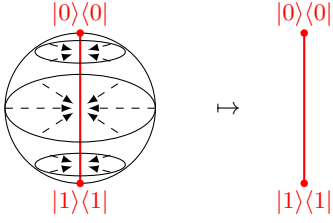

   \centering
   \begin{tabular}{ccc}
    \btp
    \begin{pgfonlayer}{nodelayer}
        \node[shape=circle, fill=red, inner sep =1pt] (0) at (0,2) {};
        \node[shape=circle, fill=red, inner sep =1pt] (1) at (0,-2) {};
        \node [style=none,font=\color{red}] (label0) at (0,2.5) {$\ket{0}\!\!\bra{0}$};
        \node [style=none,font=\color{red}] (label1) at (0,-2.5) {$\ket{1}\!\!\bra{1}$};
        \node [style=none] (aux1) at (2,0) {};
        \node [style=none] (aux2) at (-2,0) {};
        \node[style=none] (aux3) at (1,-0.65) {};
        \node[style=none] (aux4) at (1,0.65) {};
        \node [style=none] (aux5) at (-1,0.65) {};
        \node [style=none] (aux6) at (-1,-0.65) {};
        \node [style=circle,fill=none,inner sep=2pt] (centre) at (0,0) {};
        \node [style=circle,fill=none,inner sep=2pt] (centre2) at (0,1.5) {};
        \node [style=none] (aux3-2) at (0.8,1.25) {};
        \node [style=none] (aux4-2) at (-0.8,1.25) {};
        \node [style=none] (aux5-2) at (0.8,1.7) {};
        \node [style=none] (aux6-2) at (-0.8,1.7) {};
        \node [style=circle,fill=none,inner sep=2pt] (centre3) at (0,-1.5) {};
        \node [style=none] (aux3-3) at (0.8,-1.25) {};
        \node [style=none] (aux4-3) at (-0.8,-1.25) {};
        \node [style=none] (aux5-3) at (0.8,-1.7) {};
        \node [style=none] (aux6-3) at (-0.8,-1.7) {};
    \end{pgfonlayer}
    \begin{pgfonlayer}{edgelayer}
        \draw [fill=none](0,0) circle (2.0) node [black] {};
        \draw [thick,color=red] (0.center) to (1.center) {};
        \draw [fill=none] (0,0) ellipse (2cm and 0.9cm) {};
        \draw [fill=none] (0,1.5) ellipse (1.2cm and 0.3cm) {};
        \draw [fill=none] (0,-1.5) ellipse (1.2cm and 0.3cm) {};
        \draw [dashed,->, >=latex] (aux1.center) to (centre);
        \draw [dashed,->, >=latex] (aux2.center) to (centre);
        \draw [dashed,->,>=latex] (aux3.center) to (centre);
        \draw [dashed,->,>=latex] (aux4.center) to (centre);
        \draw [dashed,->,>=latex] (aux5.center) to (centre);
        \draw [dashed,->,>=latex] (aux6.center) to (centre);
        \draw [dashed,->,>=latex] (aux3-2.center) to (centre2);
        \draw [dashed,->,>=latex] (aux4-2.center) to (centre2);
        \draw [dashed,->,>=latex] (aux5-2.center) to (centre2);
        \draw [dashed,->,>=latex] (aux6-2.center) to (centre2);
        \draw [dashed,->,>=latex] (aux3-3.center) to (centre3);
        \draw [dashed,->,>=latex] (aux4-3.center) to (centre3);
        \draw [dashed,->,>=latex] (aux5-3.center) to (centre3);
        \draw [dashed,->,>=latex] (aux6-3.center) to (centre3);
    \end{pgfonlayer}
\etp & \quad\quad$\mapsto$\quad\quad &
\btp
\begin{pgfonlayer}{nodelayer}
    \node[shape=circle, fill=red, inner sep =1pt] (0) at (0,2) {};
        \node[shape=circle, fill=red, inner sep =1pt] (1) at (0,-2) {};
        \node [style=none,font=\color{red}] (label0) at (0,2.5) {$\ket{0}\!\!\bra{0}$};
        \node [style=none,font=\color{red}] (label1) at (0,-2.5) {$\ket{1}\!\!\bra{1}$};
\end{pgfonlayer}
\begin{pgfonlayer}{edgelayer}
    \draw [thick,color=red] (0.center) to (1.center) {};
\end{pgfonlayer}
\etp\\
   \end{tabular}
   \vspace{5pt}
\caption{If one restricts attention to only measurements of the $Z$ component of a quantum state, then any two quantum states with the same $Z$ component are informationally equivalent, regardless of their $X$ and $Y$ component. Consequently, the entire Bloch ball of states appears informationally equivalent to a single classical bit of states (whose extremal states correspond to eigenstates of the Pauli $Z$ operator).}
   \vspace{10pt}
   \label{fig:exSimpP}
\end{figure}

The statistics one can observe in any such experiment will necessarily be identical for any two states whose $Z$ component are the same, regardless of their $X$ and $Y$ components, since no measurements in the experiment are sensitive to the latter. Consequently, this $Z$ component of the state is the only feature of a given state that is relevant for the empirical data one can observe. One can then define the equivalence class of states which are identical relative to these measurements---one such class for every possible $Z$ value. Identifying all states with the same value of $Z$ results in a new `apparent' state space which is simply a line segment containing the convex hull of the states $\ket{0}\bra{0}$ and $\ket{1}\bra{1}$. Geometrically, these equivalence classes of states form disks at a fixed height, and one can choose a canonical representative of each equivalence class by orthogonally projecting the Bloch ball onto the line segment defined by the effects, as shown in Figure~\ref{fig:exSimpP}. This is an example of a shadow map, as defined formally in the next section.

In short, the shadow map applied to the {\em true} GPT states and effects in one's experiment describes how those states and effects will appear, provided that one does not have access to any information aside from the data generated by combinations of those particular states and effects. 
 
\subsection{From GPT fragments to their shadows}

In this subsection, we will define the notion of a shadow map and a GPT fragment shadow more formally. We begin by formalizing the notion of equivalence exemplified in the previous subsection and defining quotienting maps relative to it. This provides a first definition of GPT shadows. We then give a simpler definition  that highlights the key abstract features of a shadow, and then show that the two definitions are essentially equivalent.

As motivated in the previous subsection, the idea of a shadow is to  construct a proper GPT that reproduces the predictions of the fragment by taking the original fragment and identifying states (and effects) which cannot be distinguished by any of the effects (resp. states) in the theory. 
Formally, this means that we define an equivalence relation on the states (effects) and quotient with respect to this equivalence relation. That is, for a fragment $f=(\Omega_f,\mathcal{E}_f,p_f)$ we define an equivalence relation on states as
\begin{equation}\label{equivstates}
s_1   \backsim s_2 \iff p_f(s_1,e) = p_f(s_2,e)\ \ \forall e\in \mathcal{E}_f.
\end{equation}
Since the right domain of $p_f$ is $\mathsf{Span}[\mathcal{E}_f]$, this condition is equivalent to the condition that $p_f(s_1-s_2,\cdot )=0$, or in other words, the condition that $s_1-s_2 \in \mathsf{Ker}_L[p_f]$, where this denotes the left kernel of $p_f$. This mathematical reframing is useful, because it is now clear that what we are doing is in fact a linear quotient with respect to the subspace $\mathsf{Ker}_L[p_f]$. 
We denote the linear map defined by quotienting in this manner by $\backsim:\Omega_f \to \Omega_f\ns/\ns\backsim$. 
(A more explicit notation would be $\sim_{\mathsf{Ker}_L[p_f]}$, but as we only have a single notion of quotienting for states in this work, we leave the subscript implicit. Also, the use of the upside down symbol $\backsim$ is to distinguish from the usual notion of quotienting $\sim$ that takes unquotiented operational theories to GPTs~\cite{chiribella2010probabilistic,Schmid2024structuretheorem}.)

Similarly, we define an equivalence relation (and an equivalent linear algebraic reframing of it) for the effects in the fragment as
\begin{align}\label{equiveffects}
e_1 \backsim e_2 &\iff p_f(s,e_1)=p_f(s,e_2)\ \ \forall s\in \Omega_f \\
&\iff e_1-e_2 \in \mathsf{Ker}_R[p_f],
\end{align}
where we have introduced the obvious notation for the right Kernel of $p_f$.
Identifying effects which cannot be distinguished by states in the fragment is done by quotienting relative to this equivalence relation; this quotienting defines a linear map that we also denote $\backsim:\mathcal{E}_f \to \mathcal{E}_f\ns/\ns\backsim$. 
(A more explicit notation one could also use would be $\sim_{\mathsf{Ker}_R[p_f]}$.)

It is critical to note that this quotienting is distinct from the one that takes unquotiented operational theories to quotiented operational theories~\cite{chiribella2010probabilistic,Schmid2024structuretheorem} and that is also critically important for the study of noncontextuality~\cite{gencontext,SchmidGPT,schmid2020unscrambling}. Here, the equivalence relation one quotients relative to is defined by the specific experimental scenario, {\em not} by the set of all possible preparation and measurement procedures one could in principle have done on the system in question. Moreover, the quotienting we consider in this work maps one between two quotiented operational theories: from the true quotiented operational theory (the true GPT) to an {\em apparent} quotiented operational theory (the apparent GPT).

\sloppy The two quotients defined above map the states (effects) of the fragment to the states (effects) in its shadow, but it remains to specify the probability rule for the shadow. We define the probability rule simply by demanding that it assigns the same probability for the equivalence classes as the original rule assigned for the elements of those classes. That is, if we explicitly denote the equivalence classes of states and effects as $\btilde{s}$ and $\btilde{e}$, respectively, then the probability rule $\btilde{p_f}$ for the shadow is defined by ${\btilde{p_f}(\btilde{e},\btilde{s}):=p_f(e,s)}$ for all $\btilde{s} \in \Omega_f\ns/\ns\backsim$  and all $\btilde{e} \in \mathcal{E}_f\ns/\ns\backsim$. That $\btilde{p_f}(\btilde{e},\btilde{s})$ depends only on the equivalence classes follows from the fact that  $p_f(e,s)$ is the same for all $e\in\btilde{e}$ and is the same for all $s\in\btilde{s}$; that it actually constitutes a valid map---one which is linear and unique---follows from the universal property of quotients~\cite[Sec.~13]{gamelin2013introduction}.  That is, one can define $\btilde{p_f}$ as the unique map satisfying
\beq \label{quotientprobrule}
		\btp
		\begin{pgfonlayer}{nodelayer}
			\node [style=none] (1) at (-0, -2.8){};
			\node [style=none] (2) at (-0, 2.8){};
                \node [style=right label] (12) at (0,2.3) {$f$};
                \node [style=right label] (13) at (0,-2.2) {$f$};
			\node [style=empty circle, fill=gray!30] (8) at (0,0) {$\btilde{p_f}$};
			\node [style=right label] (8) at (0, -0.7) {$f\ns/\ns\backsim$};
			\node [style=right label] (9) at (0, 0.9) {$f\ns/\ns\backsim$};
                \node [style=coinc,fill=gray!30] (10) at (0,-1.5) {$\backsim$};
                \node [style=inc,fill=gray!30] (11) at (0,1.5) {$\backsim$};
		\end{pgfonlayer}
		\begin{pgfonlayer}{edgelayer}
			\draw [qWire] (2) to (11);
                \draw [qWire] (10) to (1);
                \draw [sWire] (10) to (11);
		\end{pgfonlayer}
		\etp
		\ =\ 
		\btp
		\begin{pgfonlayer}{nodelayer}
			\node [style=none] (1) at (-0, -1.4){};
			\node [style=none] (2) at (-0, 1.4){};
			\node [style=right label] (4) at (0, 1) {$f$};
			\node [style=empty circle, fill=white] (8) at (0,0) {$p_f$};
			\node [style=right label] (6) at (0, -0.8) {$f$};
		\end{pgfonlayer}
		\begin{pgfonlayer}{edgelayer}
			\draw [qWire] (2) to (1.center);
		\end{pgfonlayer}
		\etp.
		\eeq

\begin{definition}[GPT fragment shadow] \label{defn:method4}
	 Given a GPT fragment $(\Omega_f,\mathcal{E}_f,p_f)$,
  the GPT shadow $f\ns/\ns\backsim$ is the GPT system given by
\beq
f\ns/\ns\backsim:=\left(
\left\{\btp
\begin{pgfonlayer}{nodelayer}
	\node [style=Wsquareadj] (1) at (-0, -1) {$s$};
	\node [style=none] (2) at (-0, 1) {};
	\node [style=right label] (3) at (0, 1) {$f\ns/\ns\backsim$};
        \node [style=coinc,fill=gray!30] (4) at (0,0.2) {$\backsim$};
        \node [style=right label] (5) at (0,-0.4) {$f$};
\end{pgfonlayer}
\begin{pgfonlayer}{edgelayer}
	\draw [sWire] (2.center) to (4);
        \draw [qWire] (4) to (1);
\end{pgfonlayer}
\etp\right\}_{s\in\Omega_f},
\left\{\btp
\begin{pgfonlayer}{nodelayer}
	\node [style=none] (1) at (-0, -1) {};
	\node [style=Wsquare] (2) at (-0, 1) {$e$};
	\node [style=right label] (3) at (0, -1) {$f\ns/\ns\backsim$};
        \node [style=inc,fill=gray!30] (4) at (0,-0.2) {$\backsim$};
        \node [style=right label] (5) at (0,0.5) {$f$};
\end{pgfonlayer}
\begin{pgfonlayer}{edgelayer}
	\draw [sWire] (1.center) to (4);
        \draw [qWire] (4) to (2);
\end{pgfonlayer}
\etp\right\}_{e\in\mathcal{E}_f}
,\quad\btp
\begin{pgfonlayer}{nodelayer}
	\node [style=none] (1) at (-0, -1) {};
	\node [style=none] (2) at (-0, 1) {};
	\node [style=right label] (4) at (0, -1) {$f\ns/\ns\backsim$};
	\node [style=empty circle, fill=gray!30] (5) at (0,0) {$\btilde{p_f}$};
	\node [style=right label] (6) at (0, 1) {$f\ns/\ns\backsim$};
\end{pgfonlayer}
\begin{pgfonlayer}{edgelayer}
	\draw [sWire] (2) to (1.center);
\end{pgfonlayer}
\etp
\right),
\eeq
where \quad\btp
\begin{pgfonlayer}{nodelayer}
	\node [style=none] (1) at (-0, -0.9) {};
	\node [style=none] (2) at (-0, 0.9) {};
	\node [style=right label] (3) at (0, 0.6) {$f\ns/\ns\backsim$};
        \node [style=coinc,fill=gray!30] (4) at (0,0) {$\backsim$};
        \node [style=right label] (5) at (0,-0.5) {$f$};
\end{pgfonlayer}
\begin{pgfonlayer}{edgelayer}
	\draw [sWire] (2.center) to (4);
        \draw [qWire] (4) to (1);
\end{pgfonlayer}
\etp and \quad\btp
\begin{pgfonlayer}{nodelayer}
	\node [style=none] (1) at (-0, 0.9) {};
	\node [style=none] (2) at (-0, -0.9) {};
	\node [style=right label] (3) at (0,-0.5) {$f\ns/\ns\backsim$};
        \node [style=inc,fill=gray!30] (4) at (0,0) {$\backsim$};
        \node [style=right label] (5) at (0,0.6) {$f$};
\end{pgfonlayer}
\begin{pgfonlayer}{edgelayer}
	\draw [sWire] (2.center) to (4);
        \draw [qWire] (4) to (1);
\end{pgfonlayer}
\etp  are defined as quotienting relative to the equivalence relations we defined for states (Eq.~\eqref{equivstates}) and for effects (Eq.~\eqref{equiveffects}), respectively,
and the probability rule $\btilde{p_f}$ is defined as in Eq.~\eqref{quotientprobrule}. \footnote{Note that the shadow can also be defined (up to unique isomorphism) by the universal property of quotients in the category $\mathtt{GPT-Fragment}$.}
\end{definition}

Note that a GPT shadow of any fragment is always a valid GPT system (even if the fragment is not a valid GPT system---i.e., is not relatively tomographic).\footnote{ We can therefore think of the shadow as a map from GPT fragments to GPT systems, that is, some $\mathcal{S}:|\mathtt{GPT-Fragment}|\to|\mathtt{GPT-System}|$.}

Thus, quotienting defines a map from an arbitrary fragment to its shadow. This {\em shadow map} is a particular GPT fragment embedding of $f$ into $f\ns/\ns\backsim$. This follows from the fact that the shadow reproduces the empirical probabilities of the original fragment, since
\beq\label{probrulefrag}
		\btp
		\begin{pgfonlayer}{nodelayer}
			\node [style=Wsquareadj] (1) at (-0, -2.8) {$s$};
			\node [style=Wsquare] (2) at (-0, 2.8) {$e$};
                \node [style=right label] (12) at (0,2.3) {$f$};
                \node [style=right label] (13) at (0,-2.2) {$f$};
			\node [style=empty circle, fill=gray!30] (8) at (0,0) {$\btilde{p_f}$};
			\node [style=right label] (8) at (0, -0.85) {$f\ns/\ns\backsim$};
			\node [style=right label] (9) at (0, 0.9) {$f\ns/\ns\backsim$};
                \node [style=coinc,fill=gray!30] (10) at (0,-1.6) {$\backsim$};
                \node [style=inc,fill=gray!30] (11) at (0,1.5) {$\backsim$};
		\end{pgfonlayer}
		\begin{pgfonlayer}{edgelayer}
			\draw [qWire] (2) to (11);
                \draw [qWire] (10) to (1);
                \draw [sWire] (10) to (11);
		\end{pgfonlayer}
		\etp
		\ =\ 
		\btp
		\begin{pgfonlayer}{nodelayer}
			\node [style=Wsquareadj] (1) at (-0, -1.4) {$s$};
			\node [style=Wsquare] (2) at (-0, 1.4) {$e$};
			\node [style=right label] (4) at (0, 0.9) {$f$};
			\node [style=empty circle, fill=white] (8) at (0,0) {$p_f$};
			\node [style=right label] (6) at (0, -0.8) {$f$};
		\end{pgfonlayer}
		\begin{pgfonlayer}{edgelayer}
			\draw [qWire] (2) to (1.center);
		\end{pgfonlayer}
		\etp,\quad\forall s\in\Omega_f,e\in\mathcal{E}_f,
		\eeq
together with the facts that the quotienting map on states is a state map (from $\Omega_f$ to $\Omega_f\ns/\ns\backsim$) and that the quotienting map on effects is an effect map (from $\mathcal{E}_f$ to $\mathcal{E}_f\ns/\ns\backsim$).

At its core, the GPT shadow of a fragment is simply {\em a proper GPT system that reproduces the probabilities of the fragment}.\footnote{If one defines GPT fragments in a slightly more general way that allows for subnormalized states whose normalized counterparts are not in the fragment~\cite{selby2023accessible}, then the shadow of such a fragment is not strictly a GPT system, but rather is a relatively tomographic fragment.} 
That is, it is a {\em relatively tomographic} fragment that satisfies Eq.~\eqref{probrulefrag}. 

That the shadow is a relatively tomographic fragment is clear from the way that quotienting is defined. Explicitly: for any $\btilde{s}\in\Omega_f\ns/\ns\backsim$ one can define a linear functional $\btilde{p_f}(\btilde{s},\cdot)$ taking $\mathcal{E}^A/\backsim$ to $\mathds{R}$. By the definition of the equivalence relation, it follows that the set of all such linear functionals are separating for $\mathcal{E}_f\ns/\ns\backsim$. A similar argument holds for effects.

These two facts are indeed the {\em only} essential features of a GPT shadow. 

Consequently, we can give an alternative but equivalent definition of GPT shadows as follows.

\addtocounter{definition}{-1} % Decrease the definition counter by 1
\renewcommand{\thedefinition}{\arabic{definition}$'$} % Add a prime to the current definition number

	\begin{definition}[GPT fragment shadow]\label{defn:shadow}
  
Given a GPT fragment $(\Omega_f,\mathcal{E}_f,p_f)$,
  a GPT shadow $\mathcal{S}(f):=(\sigma(\Omega_f),\tau(\mathcal{E}_f),p_\mathcal{S})$ of $f$ is a GPT system defined by two maps $\sigma: \mathsf{Span}(\Omega_f)\to U$ and $\tau: \mathsf{Span}(\mathcal{E}_f) \to V$ and a \emph{tomographic} probability rule $p_\mathcal{S}:\mathsf{Span}[\sigma(\Omega_f)]\times\mathsf{Span}[\tau(\mathcal{E}_f)]\to \mathds{R}$ such that 
\beq \label{probruleshadow}
\btp
		\begin{pgfonlayer}{nodelayer}
			\node [style=none] (1) at (-0, -1.5) {};
			\node [style=none] (2) at (-0, 1.5) {};
			\node [style=right label] (4) at (0, 1) {$f$};
			\node [style=empty circle, fill=white] (8) at (0,0) {$p_f$};
			\node [style=right label] (6) at (0, -0.8) {$f$};
		\end{pgfonlayer}
		\begin{pgfonlayer}{edgelayer}
			\draw [qWire] (2) to (1.center);
		\end{pgfonlayer}
		\etp \ =\ 
		\btp
		\begin{pgfonlayer}{nodelayer}
			\node [style=none] (1) at (-0, -2.6) {};
			\node [style=none] (2) at (-0, 2.6) {};
			\node [style=right label] (4) at (0, 2.4) {$f$};
			\node [style=empty circle,fill=gray!30] (8) at (0,0) {$p_\mathcal{S}$};
			\node [style=right label] (6) at (0, -2.2) {$f$};
              \node [style=small box,inner sep=1pt,fill=gray!30] (10) at (0,1.65) {$\tau$};
              \node [style=right label] (11) at (0,1.0) {$\mathcal{S}(f)$};
              \node [style=right label] (12) at (0,-0.8) {$\mathcal{S}(f)$};
              \node [style=small box,inner sep=1pt,fill=gray!30] (13) at (0,-1.55) {$\sigma$};
		\end{pgfonlayer}
		\begin{pgfonlayer}{edgelayer}
			\draw [qWire] (13) to (1.center);
                \draw [qWire] (10) to (2.center);
                \draw [sWire] (10) to (13);
		\end{pgfonlayer}
		\etp.
\eeq
Any such map from $f$ to $\mathcal{S}(f)$ is a GPT fragment embedding that we term a {\em shadow map}.
\end{definition}
\addtocounter{definition}{-1} % Decrease the definition counter by 1
\renewcommand{\thedefinition}{\arabic{definition}} % Add a prime to the current definition number

In the strictest sense, this is more general than Definition~\ref{defn:method4}; for example, the states and effects of the shadow in this definition are not assumed to span the space they live in, while the states and effects in a quotiented fragment as in Definition~\ref{defn:method4} will always do so. However, as we have already established that GPT state (effect) spaces are independent (up to isomorphism) of the vector spaces they are taken to live in, this does not constitute a meaningful distinction between the two definitions. Indeed, the two definitions are equivalent up to this representational distinction, as is demonstrated by the following theorem. 

\begin{restatable}{theorem}{shadowequiv}\label{shadow-equiv}
    Any two shadows of a given GPT fragment are equivalent to each other (no matter how they are constructed).  Moreover,  any GPT system equivalent to a shadow is itself a shadow.
 \end{restatable}
The full proof is given in Appendix~\ref{secshadowequiv}, but we summarize the proof here.
\proof[Proof sketch] 
The only variability in the shadow is the choice of the $\sigma$ and $\tau$, which are necessarily linear maps (Lemma~\ref{lem:LinearityFromEmpAd} in Appendix~\ref{secshadowequiv}); once these are chosen,  there is a unique probability rule which reproduces the correct probabilities. Consequently, we sometimes refer to the shadow map as a tuple $(\sigma,\tau)$. Moreover, it is clear that the properties that $\sigma$ and $\tau$ must satisfy in order that there does exist a  probability rule which reproduces the probabilities is that the kernels of $\sigma$ and $\tau$ need to be contained in the left and right kernels of $p_f$ respectively. Moreover, as we demand that the probability rule be tomographic then this means that the kernels of $\sigma$ and $\tau$ are actually equal to the left and right kernels of $p_f$ respectively. The universal property of quotients then implies that $\sigma$ and $\tau$ factor through the quotienting maps as $\sigma = \btilde{\sigma}\circ \backsim$ and $\tau=\btilde{\tau}\circ \backsim$ for some injective $\btilde{\sigma}$ and $\btilde{\tau}$. It is then easy to see that $(\btilde{\sigma},\btilde{\tau})$ defines a GPT isomorphism between $f\ns/\ns\backsim$ and $\mathcal{S}(f)$, and we can then compose these isomorphisms to give an isomorphism between any two shadows $\mathcal{S}(f)$ and $\mathcal{S}'(f)$. 
\endproof

In other words, all methods for constructing a shadow lead to the same end result\footnote{ The results of Ref.~\cite{gitton2022solvable} prove that all methods for computing shadow maps are equivalent for assessing simplex-embeddability; our result is stronger as it refers to full equivalence as GPTs.}. Thus, we sometimes refer to `the' shadow of a fragment, since all shadows are equivalent.

Yet another equivalent characterization of GPT shadows is the following: a GPT shadow of a fragment $f$ is any GPT system $G$ such that there is a full GPT embedding of $f$ into $G$.
This is very useful, as it means that one can verify that a given embedding defines a shadow map simply by checking that 1) its codomain is a valid GPT system, and 2) it is a full embedding.

In Appendix~\ref{LinearityFromEmpAd}, we also prove the following:
\begin{restatable}{lemma}{fragsub}\label{fragsub}
The shadow map for a given fragment is faithful if and only if the fragment is relatively tomographic.
\end{restatable} 

For example, the shadow map in the example of Subsection~\ref{example1} is a (nontrivial) projection on the nontomographic fragment in question, and so fails to be faithful. 

The idea of GPT fragment shadows has appeared (at least indirectly) a number of times in the literature~\cite{Holevo1982,gitton2022solvable,muller2023testing}. Most notably, the notion of an ``alternative reduced space'' representation introduced in Ref.~\cite{gitton2022solvable} is an instance of Definition~\ref{defn:shadow} when one takes the probability rule to be an inner product (and specializes to quantum theory). Indeed, Ref.~\cite{gitton2022solvable} presents two different ways of computing such a shadow via projections. 
These two methods are illustrated in Example~2 in Subsection~\ref{twodisk} below.

\subsection{The shadow map acts trivially on GPT subsystems}

We now note a simple but critical fact: that shadow maps act essentially trivially on GPT fragments that are themselves GPT systems---that is, on GPT subsystems (fragments that are relatively tomographic).

\begin{theorem}\label{tomogshadows}
Every GPT subsystem is equivalent to its own shadow.  
\end{theorem}  
\begin{proof}
    This follows immediately from the definition of shadows in terms of quotienting (Definition~\ref{defn:method4}). As the states and effects in a GPT subsystem are  tomographic for each other, it follows that the equivalence relation is trivial, and the maps $\backsim$ (for both states and effects) are the identity.  
\end{proof}

In other words, a GPT fragment whose states and effects are relatively tomographic is isomorphic to its shadow, with the isomorphism given by the shadow map itself. 

We illustrate this by an example in the next section.

\subsection{Example 2: shadows of a GPT subsystem}\label{twodisk}

We can get some more intuition for the shadow map by considering a simple example within quantum theory. In this example, the fragment in question is relatively tomographic, and hence the shadow map acts trivially in the sense that its image and preimage are equivalent as GPT systems (by Theorem~\ref{tomogshadows}).
This example was suggested by Markus M\"uller.

Consider a fragment of a 2-dimensional quantum system comprised of an equatorial disk of states and an equatorial disk of effects, but where these disks are neither identical nor orthogonal, as shown in Figure~\ref{fig:twodisk}. 

\begin{figure}[htb!]
   \centering
\tikzset{every picture/.style={line width=0.75pt}} %set default line width to 0.75pt        
\begin{tabular}{ccc}
\begin{tikzpicture}[x=0.5pt,y=0.5pt,yscale=-1,xscale=1]
%uncomment if require: \path (0,300); %set diagram left start at 0, and has height of 300

%Shape: Ellipse [id:dp534904058587307] 
\draw   (160.33,100.07) .. controls (160.33,61.26) and (191.9,29.8) .. (230.83,29.8) .. controls (269.77,29.8) and (301.33,61.26) .. (301.33,100.07) .. controls (301.33,138.87) and (269.77,170.33) .. (230.83,170.33) .. controls (191.9,170.33) and (160.33,138.87) .. (160.33,100.07) -- cycle ;
%Shape: Ellipse [id:dp541626791863437] 
\draw  [fill={rgb, 255:red, 0; green, 134; blue, 0 }  ,fill opacity=0.3 ] (200.14,101.09) .. controls (200.14,90.18) and (213.92,81.33) .. (230.92,81.33) .. controls (247.92,81.33) and (261.7,90.18) .. (261.7,101.09) .. controls (261.7,112) and (247.92,120.85) .. (230.92,120.85) .. controls (213.92,120.85) and (200.14,112) .. (200.14,101.09) -- cycle ;
%Curve Lines [id:da3352965644056838] 
\draw  [dash pattern={on 0.84pt off 2.51pt}]  (160.33,100.07) .. controls (165.33,75.33) and (294.33,73.33) .. (301.33,100.07) ;
%Curve Lines [id:da7906783569855518] 
\draw    (160.33,100.07) .. controls (166.33,125.33) and (296.33,128.33) .. (301.33,100.07) ;
%Curve Lines [id:da9017235588342121] 
\draw  [dash pattern={on 0.84pt off 2.51pt}]  (261.53,36.6) .. controls (235.33,29.33) and (181.33,149.33) .. (200.14,163.53) ;
%Curve Lines [id:da05340805300809115] 
\draw    (200.14,163.53) .. controls (226.33,171.33) and (282.33,51.33) .. (261.53,36.6) ;
%Straight Lines [id:da16702492665559698] 
\draw [line width=0.75]  [dash pattern={on 4.5pt off 4.5pt}]  (261.53,36.6) -- (261.69,98.09) ;
\draw [shift={(261.7,101.09)}, rotate = 269.85] [fill={rgb, 255:red, 0; green, 0; blue, 0 }  ][line width=0.08]  [draw opacity=0] (10.72,-5.15) -- (0,0) -- (10.72,5.15) -- (7.12,0) -- cycle    ;
%Straight Lines [id:da3973263403675318] 
\draw [color={rgb, 255:red, 0; green, 0; blue, 0 }  ,draw opacity=1 ][line width=0.75]  [dash pattern={on 4.5pt off 4.5pt}]  (200.14,163.53) -- (200.14,104.09) ;
\draw [shift={(200.14,101.09)}, rotate = 90] [fill={rgb, 255:red, 0; green, 0; blue, 0 }  ,fill opacity=1 ][line width=0.08]  [draw opacity=0] (10.72,-5.15) -- (0,0) -- (10.72,5.15) -- (7.12,0) -- cycle    ;

% Text Node
\draw (210,91.4) node [anchor=north west][inner sep=0.75pt]    {$\tau (\mathcal{E})$};
% Text Node
\draw (264,18.4) node [anchor=north west][inner sep=0.75pt]    {$\mathcal{E}$};
% Text Node
\draw (307,94.4) node [anchor=north west][inner sep=0.75pt]    {$\Omega $};
\end{tikzpicture}
& \hspace{30pt}&
\begin{tikzpicture}[x=0.5pt,y=0.5pt,yscale=-1,xscale=1]
%uncomment if require: \path (0,300); %set diagram left start at 0, and has height of 300

%Shape: Ellipse [id:dp534904058587307] 
\draw   (160.33,100.07) .. controls (160.33,61.26) and (191.9,29.8) .. (230.83,29.8) .. controls (269.77,29.8) and (301.33,61.26) .. (301.33,100.07) .. controls (301.33,138.87) and (269.77,170.33) .. (230.83,170.33) .. controls (191.9,170.33) and (160.33,138.87) .. (160.33,100.07) -- cycle ;
%Shape: Ellipse [id:dp541626791863437] 
\draw  [fill={rgb, 255:red, 0; green, 134; blue, 0 }  ,fill opacity=0.3 ] (218.39,129.21) .. controls (208.43,124.77) and (205.95,108.58) .. (212.87,93.05) .. controls (219.79,77.52) and (233.47,68.53) .. (243.44,72.97) .. controls (253.41,77.41) and (255.88,93.6) .. (248.97,109.13) .. controls (242.05,124.66) and (228.36,133.65) .. (218.39,129.21) -- cycle ;
%Curve Lines [id:da3352965644056838] 
\draw  [dash pattern={on 0.84pt off 2.51pt}]  (160.33,100.07) .. controls (165.33,75.33) and (294.33,73.33) .. (301.33,100.07) ;
%Curve Lines [id:da7906783569855518] 
\draw    (160.33,100.07) .. controls (166.33,125.33) and (296.33,128.33) .. (301.33,100.07) ;
%Curve Lines [id:da9017235588342121] 
\draw  [dash pattern={on 0.84pt off 2.51pt}]  (261.53,36.6) .. controls (235.33,29.33) and (181.33,149.33) .. (200.14,163.53) ;
%Curve Lines [id:da05340805300809115] 
\draw    (200.14,163.53) .. controls (226.33,171.33) and (282.33,51.33) .. (261.53,36.6) ;
%Straight Lines [id:da16702492665559698] 
\draw [line width=0.75]  [dash pattern={on 4.5pt off 4.5pt}]  (301.33,100.07) -- (246.16,74.25) ;
\draw [shift={(243.44,72.97)}, rotate = 25.08] [fill={rgb, 255:red, 0; green, 0; blue, 0 }  ][line width=0.08]  [draw opacity=0] (10.72,-5.15) -- (0,0) -- (10.72,5.15) -- (7.12,0) -- cycle    ;
%Straight Lines [id:da3973263403675318] 
\draw [color={rgb, 255:red, 0; green, 0; blue, 0 }  ,draw opacity=1 ][line width=0.75]  [dash pattern={on 4.5pt off 4.5pt}]  (160.33,100.07) -- (215.71,127.86) ;
\draw [shift={(218.39,129.21)}, rotate = 206.65] [fill={rgb, 255:red, 0; green, 0; blue, 0 }  ,fill opacity=1 ][line width=0.08]  [draw opacity=0] (10.72,-5.15) -- (0,0) -- (10.72,5.15) -- (7.12,0) -- cycle    ;

% Text Node
\draw (210,91.4) node [anchor=north west][inner sep=0.75pt]    {$\sigma ( \Omega )$};
% Text Node
\draw (264,18.4) node [anchor=north west][inner sep=0.75pt]    {$\mathcal{E}$};
% Text Node
\draw (307,94.4) node [anchor=north west][inner sep=0.75pt]    {$\Omega $};
\end{tikzpicture}
\end{tabular}
   \vspace{10pt}
   \caption{Consider the fragment of quantum theory whose state space is the disk of states shown here as $\Omega$ and whose effect space is the disk of effects shown here as $\mathcal{E}$(strictly speaking, the effect space is the convex hull of this disc with the zero and unit effects). Two natural GPT shadows are constructed by projecting one of these onto the other, as shown in (a) and (b) respectively. These two shadows are equivalent as GPT systems, in accordance with Theorem~\ref{shadow-equiv}. they are also equivalent to the original GPT fragment, in accordance with Theorem~\ref{tomogshadows}. }
   \label{fig:twodisk}
\end{figure}
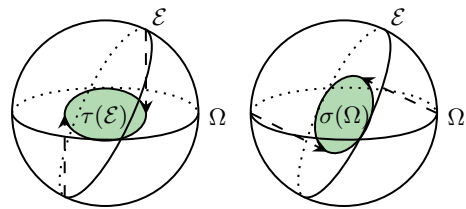

That is, the set of states in the fragment corresponds to the set of density matrices associated with points in the disk labeled by $\Omega$, and the set of effects in the fragment corresponds to the set of POVM elements associated with the points in the convex hull of the disk labeled by $\mathcal{E}$, the zero effect and the unit effect.
Denoting a generic density operator in $\Omega$ by $s$ and a generic POVM element in $\mathcal{E}$ by $E$, we assign probabilities to state-effect pairs following the Born rule, as $p(E,s)= \mathsf{Tr}[Es]$. 

We describe two different methods of constructing shadows of this fragment. These are depicted in Fig.~\ref{fig:twodisk}(a) and (b), respectively, and are essentially the same as the two methods for computing a `reduced space representation' introduced in Ref.~\cite{gitton2022solvable}.

The first shadow map we consider is simply the orthogonal projection of the effect space onto the state space. This is depicted in Fig.~\ref{fig:twodisk}(a), and results in a new effect space which constitutes a smaller disk lying in the same plane as the state space (which is unchanged). Note that the probabilities assigned by the Born rule to any given state-effect pair in the original fragment are unchanged by this projection, since the projection only changes the y-component of the effects and since the y-component of every state in the fragment is already zero.

This is an instance of the shadow map in Definition~\ref{defn:shadow}, where the map $\sigma$ is taken to be trivial and $\tau$ is taken to be the projection map just described; since these maps together preserve the probabilities while mapping the states and effects to a valid GPT system, they define a valid shadow map.

The second method is the orthogonal projection of the state space onto the effect space. This is depicted in Fig.~\ref{fig:twodisk}(b), and results in a new state space which constitutes a smaller disk lying in the same plane as the effect space (which is unchanged). (Here, it is $\tau$ that is trivial and $\sigma$ that is defined by a projection map.)

The two resulting GPT systems are equivalent GPT systems, in accordance with Theorem~\ref{shadow-equiv}. (This is evident visually by the fact that they are identical subsets of the Bloch sphere up to a rotation.) 
Moreover, they are also equivalent as GPT systems to the original fragment itself, as a consequence of the fact that the latter constitutes a GPT system together with Theorem~\ref{tomogshadows}. This is the sense in which the shadow map is said to act trivially (despite the fact that it does change the precise representation of the fragment in question).

This fragment of the qubit has states and effects which are tomographic for one another, but which do not span the same vector space. Not all papers in the literature would consider such an example to be a valid GPT, since many papers demand that the states and effects of a GPT span the same vector space. But such more general cases are no more or less expressive (in terms of the probabilities they generate as well as the operational identities they encode) as those where the states and effects span the same space, so we consider them to be valid GPTs as well. Indeed, as just shown, the original fragment (whose states and effects are tomographic for each other but not spanning the same space) is equivalent to the two shadows of it we constructed (whose states and effects are both tomographic {\em and} span the same space). This again demonstrates how our definitions are more flexible than standard definitions.

\subsection{Summary of concepts}\label{subse:sum}

We finish this section by presenting a pictorial depiction of the main concepts introduced from Sections \ref{se:gptp} to \ref{se:shadow}, and how they relate to each other (see Fig.~\ref{fig:sumfig}). In the remaining of the paper we discuss their impact on assessments of nonclassicality. 

\begin{figure*}
    \centering    
\tikzset{every picture/.style={line width=0.75pt}} %set default line width to 0.75pt     
\begin{tikzpicture}[x=0.75pt,y=0.75pt,yscale=-1,xscale=1]

%Straight Lines [id:da4276313131771834] 
\draw  [dash pattern={on 4.5pt off 4.5pt}]  (12.33,120) -- (644.33,120) ;
%Straight Lines [id:da7941276801899804] 
\draw    (160,74) -- (160,271) ;
\draw [shift={(160,273)}, rotate = 270] [color={rgb, 255:red, 0; green, 0; blue, 0 }  ][line width=0.75]    (10.93,-3.29) .. controls (6.95,-1.4) and (3.31,-0.3) .. (0,0) .. controls (3.31,0.3) and (6.95,1.4) .. (10.93,3.29)   ;
%Straight Lines [id:da48920466915212446] 
\draw    (197.33,70) -- (308.57,130.05) ;
\draw [shift={(310.33,131)}, rotate = 208.36] [color={rgb, 255:red, 0; green, 0; blue, 0 }  ][line width=0.75]    (10.93,-3.29) .. controls (6.95,-1.4) and (3.31,-0.3) .. (0,0) .. controls (3.31,0.3) and (6.95,1.4) .. (10.93,3.29)   ;
%Straight Lines [id:da8346293507267426] 
\draw    (346.77,277) -- (346.77,178.02) ;
\draw [shift={(346.77,176.02)}, rotate = 90] [color={rgb, 255:red, 0; green, 0; blue, 0 }  ][line width=0.75]    (10.93,-3.29) .. controls (6.95,-1.4) and (3.31,-0.3) .. (0,0) .. controls (3.31,0.3) and (6.95,1.4) .. (10.93,3.29)   ;
%Straight Lines [id:da5258960987942441] 
\draw    (355.33,176.36) -- (355.33,265.58) -- (355.33,275.76) ;
\draw [shift={(355.33,277.76)}, rotate = 270] [color={rgb, 255:red, 0; green, 0; blue, 0 }  ][line width=0.75]    (10.93,-3.29) .. controls (6.95,-1.4) and (3.31,-0.3) .. (0,0) .. controls (3.31,0.3) and (6.95,1.4) .. (10.93,3.29)   ;
%Straight Lines [id:da7809113525277013] 
\draw   [color={rgb, 255: red, 125; green, 0; blue, 0}] (188.33,271) -- (299.92,159.41) (192.58,261.1) -- (198.23,266.76)(199.65,254.03) -- (205.3,259.69)(206.72,246.96) -- (212.37,252.62)(213.79,239.89) -- (219.45,245.54)(220.86,232.82) -- (226.52,238.47)(227.93,225.75) -- (233.59,231.4)(235,218.67) -- (240.66,224.33)(242.07,211.6) -- (247.73,217.26)(249.14,204.53) -- (254.8,210.19)(256.22,197.46) -- (261.87,203.12)(263.29,190.39) -- (268.94,196.05)(270.36,183.32) -- (276.01,188.98)(277.43,176.25) -- (283.09,181.9)(284.5,169.18) -- (290.16,174.83);
\draw [shift={(301.33,158)}, rotate = 135] [color={rgb, 255:red, 125; green, 0; blue, 0 }  ][line width=0.75]    (10.93,-3.29) .. controls (6.95,-1.4) and (3.31,-0.3) .. (0,0) .. controls (3.31,0.3) and (6.95,1.4) .. (10.93,3.29)   ;
%Straight Lines [id:da9111887913244565] 
\draw    (513.33,301.52) -- (422.33,301.52) ;
\draw [shift={(420.33,301.52)}, rotate = 360] [color={rgb, 255:red, 0; green, 0; blue, 0 }  ][line width=0.75]    (10.93,-3.29) .. controls (6.95,-1.4) and (3.31,-0.3) .. (0,0) .. controls (3.31,0.3) and (6.95,1.4) .. (10.93,3.29)   ;
%Straight Lines [id:da5058575292229358] 
\draw    (576.77,277) -- (576.77,178.02) ;
\draw [shift={(576.77,176.02)}, rotate = 90] [color={rgb, 255:red, 0; green, 0; blue, 0 }  ][line width=0.75]    (10.93,-3.29) .. controls (6.95,-1.4) and (3.31,-0.3) .. (0,0) .. controls (3.31,0.3) and (6.95,1.4) .. (10.93,3.29)   ;
%Straight Lines [id:da9139680500057482] 
\draw    (585.33,176.36) -- (585.33,265.58) -- (585.33,275.76) ;
\draw [shift={(585.33,277.76)}, rotate = 270] [color={rgb, 255:red, 0; green, 0; blue, 0 }  ][line width=0.75]    (10.93,-3.29) .. controls (6.95,-1.4) and (3.31,-0.3) .. (0,0) .. controls (3.31,0.3) and (6.95,1.4) .. (10.93,3.29)   ;
%Straight Lines [id:da9223684064601653] 
\draw    (513.33,141.52) -- (422.33,141.52) ;
\draw [shift={(420.33,141.52)}, rotate = 360] [color={rgb, 255:red, 0; green, 0; blue, 0 }  ][line width=0.75]    (10.93,-3.29) .. controls (6.95,-1.4) and (3.31,-0.3) .. (0,0) .. controls (3.31,0.3) and (6.95,1.4) .. (10.93,3.29)   ;

% Text Node
\draw (21.33,125) node [anchor=north west][inner sep=0.75pt]   [align=left] {\begin{minipage}[lt]{55.34pt}\setlength\topsep{0pt}
\begin{center}
\textit{{\footnotesize (Relatively) Tomographic}}
\end{center}

\end{minipage}};
% Text Node
\draw (17,101) node [anchor=north west][inner sep=0.75pt]   [align=left] {\begin{minipage}[lt]{70pt}\setlength\topsep{0pt}
\begin{center}
{\footnotesize \textit{Non-tomographic}}
\end{center}

\end{minipage}};
% Text Node
\draw  [draw opacity=0][fill={quantumviolet}  ,fill opacity=0.3]  (115,22) -- (208,22) -- (208,67) -- (115,67) -- cycle  ;
\draw (118,26) node [anchor=north west][inner sep=0.75pt]   [align=left] {\begin{minipage}[lt]{61.18pt}\setlength\topsep{0pt}
\begin{center}
{\footnotesize GPT fragment $f$ (Def.~\ref{def:GPTfragment})}
\end{center}

\end{minipage}};
% Text Node
\draw  [draw opacity=0][fill={rgb, 255:red, 134; green, 0; blue, 0 }  ,fill opacity=0.3]  (119,275.62) -- (206,275.62) -- (206,335) -- (119,335) -- cycle  ;
\draw (122,279.62) node [anchor=north west][inner sep=0.75pt]   [align=left] {\begin{minipage}[lt]{59pt}\setlength\topsep{0pt}
\begin{center}
{\footnotesize GPT fragment  shadow $\mathcal{S}( f)$ (Def.~\ref{defn:method4}/\ref{defn:shadow})}
\end{center}

\end{minipage}};
% Text Node
\draw  [draw opacity=0][fill={quantumviolet}  ,fill opacity=0.3]  (316,126) -- (392,126) -- (392,171) -- (316,171) -- cycle  ;
\draw (319,130) node [anchor=north west][inner sep=0.75pt]   [align=left] {\begin{minipage}[lt]{50pt}\setlength\topsep{0pt}
\begin{center}
{\footnotesize GPT system $\mathcal{G}$ (Def.~\ref{def:GPTsystem})}
\end{center}

\end{minipage}};
% Text Node
\draw  [draw opacity=0][fill={quantumviolet}  ,fill opacity=0.3]  (291,288.87) -- (412,288.87) -- (412,333.87) -- (291,333.87) -- cycle  ;
\draw (294,292.87) node [anchor=north west][inner sep=0.75pt]   [align=left] {\begin{minipage}[lt]{80.12pt}\setlength\topsep{0pt}
\begin{center}
{\footnotesize GPT shadow $\mathcal{S}(\mathcal{G})$ (Def.~\ref{defn:method4}/\ref{defn:shadow})}
\end{center}

\end{minipage}};
% Text Node
\draw  [draw opacity=0][fill={quantumviolet}  ,fill opacity=0.3]  (535,127) -- (631,127) -- (631,172) -- (535,172) -- cycle  ;
\draw (538,131) node [anchor=north west][inner sep=0.75pt]   [align=left] {\begin{minipage}[lt]{65pt}\setlength\topsep{0pt}
\begin{center}
{\footnotesize GPT subsystem   $\mathcal{F}$ (Def.~\ref{def:subsystem})}
\end{center}

\end{minipage}};
% Text Node
\draw  [draw opacity=0][fill={quantumviolet}  ,fill opacity=0.3]  (525,288.49) -- (639,288.49) -- (639,333.49) -- (525,333.49) -- cycle  ;
\draw (528,292.49) node [anchor=north west][inner sep=0.75pt]   [align=left] {\begin{minipage}[lt]{75.13pt}\setlength\topsep{0pt}
\begin{center}
{\footnotesize GPT shadow $\mathcal{S}(\mathcal{F})$ (Def.~\ref{defn:method4}/\ref{defn:shadow})}
\end{center}

\end{minipage}};
% Text Node
\draw (210,50) node [anchor=north west][inner sep=0.75pt]  [rotate=-30.16] [align=left] {\begin{minipage}[lt]{90pt}\setlength\topsep{0pt}
\begin{center}
{\footnotesize Inclusion (Def.~\ref{def:GPTfragment})}
\end{center}

\end{minipage}};
% Text Node
\draw (120,252.13) node [anchor=north west][inner sep=0.75pt]  [rotate=-270] [align=left] {\begin{minipage}[lt]{95pt}\setlength\topsep{0pt}
\begin{center}
{\footnotesize Shadow maps (Def.~\ref{defn:method4}/\ref{defn:shadow})}
\end{center}

\end{minipage}};
% Text Node
\draw (365.32,282.26) node [anchor=north west][inner sep=0.75pt]  [rotate=-269.61] [align=left] {\begin{minipage}[lt]{75.74pt}\setlength\topsep{0pt}
\begin{center}
{\footnotesize Equivalence (Thm.~\ref{tomogshadows})}
\end{center}

\end{minipage}};
% Text Node
\draw (412,306) node [anchor=north west][inner sep=0.75pt]   [align=left] {{\footnotesize Embedding (Def.~\ref{def:subsystem})}};
% Text Node
\draw (174,226) node [anchor=north west][inner sep=0.75pt]  [rotate=-315.89] [align=left] {\begin{minipage}[lt]{92.08pt}\setlength\topsep{0pt}
\begin{center}
{\footnotesize Not always embeddable (e.g. Thm.~\ref{thm:Holevo})}
\end{center}
\end{minipage}};
% Text Node
\draw (595.32,282.26) node [anchor=north west][inner sep=0.75pt]  [rotate=-269.61] [align=left] {\begin{minipage}[lt]{75.74pt}\setlength\topsep{0pt}
\begin{center}
{\footnotesize Equivalence (Thm.~\ref{tomogshadows})}
\end{center}

\end{minipage}};
% Text Node
\draw (415,146) node [anchor=north west][inner sep=0.75pt]   [align=left] {{\footnotesize Embedding (Def.~\ref{def:subsystem})}};
\end{tikzpicture}
\bigskip
    \caption{Summary of concepts introduced from Sections \ref{se:gptp} to \ref{se:shadow}, and how they relate to each other. \textbf{Concepts:} GPT system $\mathcal{G}$, GPT subsystem $\mathcal{F}$, fragment $f$ of a GPT system $\mathcal{G}$, shadow $\mathcal{S}(f)$ of a GPT fragment $f$, shadow $\mathcal{S}(\mathcal{G})$ of a GPT system $\mathcal{G}$, and shadow $\mathcal{S}(\mathcal{F})$ of a GPT subsystem $\mathcal{F}$. \textbf{Relations:} $\mathcal{G}$ and $\mathcal{S}(\mathcal{G})$ are equivalent, just as $\mathcal{F}$ and $\mathcal{S}(\mathcal{F})$ are also equivalent. Subsystem $\mathcal{F}$ of $\mathcal{G}$ can be embedded into $\mathcal{G}$, just as $\mathcal{S}(\mathcal{F})$ can be embedded into $\mathcal{S}(\mathcal{G})$. The GPT fragment $f$ can be included into $\mathcal{G}$ via an inclusion map, and mapped into the shadow $S(f)$ via the shadow maps. The shadow $S(f)$ cannot always be embedded into $\mathcal{G}$. As we demonstrate in the coming sections, only GPT shadows that do not constitute subsystems (red box) impose challenges to assessments of nonclassicality, while all other objects (purple boxes) yield trustworthy assessments.}
    \label{fig:sumfig}
\end{figure*}

\section{The effect of shadow maps on simplex embeddability} 

In this section, we consider the effect of shadow maps on simplex-embeddability of a given fragment.

\subsection{Shadow maps can break simplex embeddability}\label{sec:ProbFailureTomog}
	
We now prove that the shadow of a GPT fragment may not be simplex-embeddable, even if the fragment itself is simplex-embeddable. This follows from a general fact: that {\em every possible} polytopic GPT system can be recovered as the shadow of some fragment of a simplicial GPT system.

\begin{restatable}{theorem}{holevo}\label{thm:Holevo}
Every polytopic GPT system is the shadow of some fragment of a strictly classical (i.e., simplicial) GPT system of some dimension.
\end{restatable}

We prove this in Appendix~\ref{sec:Holevogeneral}.  This theorem dates back to Holevo~\cite{Holevo1982} and Beltrametti-Bugajski~\cite{Beltrametti_1995}, and an analogous construction works for arbitrarily close approximations of nonpolytopic GPTs~\cite{muller2023testing}. Our proof aims to give as clear and explicit a construction of the fragment as possible. 

The idea of the construction is to simply imagine that all extremal preparations in one's experiment are perfectly discriminable---the vertices of a simplicial theory, and that one merely has not implemented sufficient measurements to see this fact. Under this assumption, it is easy to construct a simplicial fragment that reproduces one's data table. All that remains is to figure out which subset of the logically possible measurements are the ones that reproduce the observed data. This is done simply by associating to each effect $e$ the functional on vertices of the simplex defined by $\xi_e(\cdot):=p(e,\cdot)$.

We illustrate this theorem by a simple example: by demonstrating how the gbit is the shadow of a fragment of a simplicial theory. This example has also been discussed previously in the literature~\cite{Holevo1982,muller2023testing,selby2024linear}. (Here, we give some extra details relative to those earlier presentations.)

\subsection{ Example 3: the gbit as the shadow of a fragment of a simplicial theory }\label{HolevoSec}

We now repeat a well-known example that was originally due to Holevo~\cite[Sec.~1.5, p.~18]{Holevo1982}. 

The gbit $\mathcal{G}$ has four pure states which form a regular square in a plane of $\mathds{R}^3$ that does not contain the zero vector. We can write these as
\begin{align}\label{eq:statesG}
   s_1= \begin{bmatrix}
        1\\
        1\\
        0
    \end{bmatrix}\ns,\ 
    s_2 = \begin{bmatrix}
        1\\
        -1\\
        0
    \end{bmatrix}\ns,\ 
    s_3 = \begin{bmatrix}
        1\\
        0\\
        1
    \end{bmatrix}\ns,\ 
    s_4 = \begin{bmatrix}
        1\\
        0\\
        -1
    \end{bmatrix}\ns.
\end{align}
The normalized state space $\bar{\Omega}_\mathcal{G}$ of the gbit is the convex hull of those extremal states. The effects in the gbit include all of the logically possible effects (so that the GPT is said to satisfy the no-restriction hypothesis~\cite{chiribella2010probabilistic}). These effects constitute the convex hull $\mathcal{E}_\mathcal{G}=\mathsf{Conv}\left(\{0,u,e_{13},e_{14},e_{23},e_{24}\}\right)$, where $e_{ij}$ are the effects that assign probability one to all states on the edge of the square that contains states $s_i$ and $s_j$. By defining $e(s) = e\circ s$, where $\circ$ is the usual matrix product\footnote{Note that the effects are naturally described as row vectors, as they are elements of the dual. Here we write their transposes (as column vectors) to spare space.}, these vectors take the explicit form
\begin{align}\label{eq:effectsG}
    e^T_{13}=\frac{1}{2}\begin{bmatrix}
        1\\
        1\\
        1
    \end{bmatrix}\ns,\ 
    e^T_{14}=\frac{1}{2}\begin{bmatrix}
        1\\
        1\\
        -1
    \end{bmatrix}\ns,\ 
    e^T_{23} = \frac{1}{2}\begin{bmatrix}
        1\\
        -1\\
        1,
    \end{bmatrix}\ns,\ 
    e^T_{24} = \frac{1}{2}\begin{bmatrix}
        1\\
        -1\\
        -1
    \end{bmatrix}\ns.
\end{align}
The states and effects in Eq.~\eqref{eq:statesG} and Eq.~\eqref{eq:effectsG} are plotted in purple in Figure~\ref{fig:H1}. 
Defining the unit effect as
\begin{align}
    u^T = \begin{bmatrix}
        1\\
        0\\
        0
    \end{bmatrix}\ns,
\end{align}
 \sloppy one can check that $u(\omega)=1$ for all $\omega\in\bar{\Omega}_\mathcal{G}$ and that ${e_{13}+e_{24}=e_{14}+e_{23}=u}$. Thus, the gbit is taken to have two binary outcome measurements, defined by $\{e_{13},e_{24}\}$ and $\{e_{14},e_{23}\}$.
  
Now we will construct a fragment of a simplicial GPT system whose shadow reproduces the gbit, which we hereby denote as the \emph{Holevo fragment} $\mathscr{h}$. This construction is an example of the general construction in Appendix~\ref{sec:Holevogeneral} that is used to prove Theorem~\ref{thm:Holevo}. 
 
There are $n=4$ extremal states of the gbit, so the classical system in question has a 4-vertex simplex as its normalized state space---i.e.,  a tetrahedron $\Delta_3$ in three dimensions. We can take each of the extremal states in $\mathbbm{R}^4$ to be unit basis vectors
\begin{align}
   \mu_1= \begin{bmatrix}
        1\\
        0\\
        0\\
        0
    \end{bmatrix}\ns,\ 
    \mu_2 = \begin{bmatrix}
        0\\
        1\\
        0\\
        0
    \end{bmatrix}\ns,\ 
    \mu_3 = \begin{bmatrix}
        0\\
        0\\
        1\\
        0
    \end{bmatrix}\ns,\
    \mu_4 = \begin{bmatrix}
        0\\
        0\\
        0\\
        1
    \end{bmatrix}\ns, 
\end{align}
and we take the set of normalized states in the fragment to be $\bar{\Omega}_{\mathscr{h}}:=\mathsf{Conv}(\{\mu_1,\mu_2,\mu_3,\mu_4\})$.

To construct the effects in the Holevo fragment, we first define the matrix $L$ in which each column is a vector $s_i$ of the gbit, namely 
 \begin{align}
     L = \begin{bmatrix}
         1 & 1 &1 &1\\
         1& -1 &0& 0\\
         0& 0& 1& -1 
     \end{bmatrix}.
 \end{align}
We then act this matrix on the right of each gbit effect $e\in\mathcal{E}_\mathcal{G}$ to get the effects of the Holevo fragment, defining ${\mathcal{E}_\mathscr{h}:=\mathsf{Conv}\left(\{\xi_0,\xi_{13},\xi_{14},\xi_{23},\xi_{24},\xi_u\}\right)}$ where $\xi_e= e\circ L$. Explicitly, these are the effects $\xi_0 = \vec{0}$, $\xi_u=\begin{bmatrix}
	1 & 1 & 1 & 1
\end{bmatrix}$, and
\begin{align}
    \xi^T_{e_{13}}=\begin{bmatrix}
        1\\
        0\\
        1\\
        0
    \end{bmatrix}\ns,\ 
    \xi^T_{e_{14}}=\begin{bmatrix}
        1\\
        0\\
        0\\
        1
    \end{bmatrix}\ns,\ 
    \xi^T_{e_{23}}=\begin{bmatrix}
        0\\
        1\\
        1\\
        0
    \end{bmatrix}\ns,\ 
    \xi^T_{e_{24}}=\begin{bmatrix}
        0\\
        1\\
        0\\
        1
    \end{bmatrix}\ns.
\end{align}

In summary, the Holevo fragment $\mathscr{h}$ is defined by $\Omega_\mathscr{h}:=\mathsf{Conv}(0\cup\bar{\Omega}_\mathscr{h})$, $\mathcal{E}_\mathscr{h}$, and the probability rule given by the evaluation map (i.e. by simply composing the effect with the state).

The states in this fragment span $\mathds{R}^{4}$. However, the effects do not span the whole dual space $(\mathds{R}^{4})^*$, since the linear dependence $\xi_{13}+\xi_{24}=\xi_{14}+\xi_{23}$ between them implies that only three of them are linearly independent. Consequently, $\mathsf{dim}\left[\mathsf{Span}\left(\mathcal{E}_\mathscr{h}\right)\right]<\mathsf{dim}\left[\mathsf{Span}\left(\Omega_\mathscr{h}\right)\right]$, and so this Holevo fragment is not relatively tomographic. More specifically, the effects cannot separate the states.

Finally, let us show that the shadow $\mathcal{S}(\mathscr{h})$ of this Holevo fragment is the the gbit $\mathcal{G}$. This fragment and its shadow are depicted in Figure~\ref{fig:H1}.

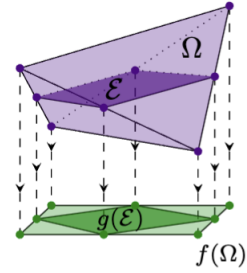
\begin{figure}[t]
   \centering
\tikzset{every picture/.style={line width=0.5pt}} %set default line width to 0.75pt        

\begin{tikzpicture}[x=0.6pt,y=0.6pt,yscale=-1,xscale=1]
	%uncomment if require: \path (0,300); %set diagram left start at 0, and has height of 300
	
	%Shape: Polygon [id:ds13161920711702235] 
	\draw  [fill={quantumviolet}  ,fill opacity=0.3 ] (308.52,42.59) -- (284.48,158.02) -- (167.85,138.78) -- (145.01,93.09) -- cycle ;
	%Shape: Polygon [id:ds0324025224285468] 
	\draw  [draw opacity=0][fill={quantumviolet}  ,fill opacity=0.5 ] (236.58,93.09) -- (297.1,99.1) -- (214.74,125.55) -- (156.43,115.93) -- cycle ;
	%Shape: Polygon [id:ds4118660503580497] 
	\draw  [fill={rgb, 255:red, 0; green, 134; blue, 0}  ,fill opacity=0.3 ] (167.85,197.69) -- (308.52,197.69) -- (284.48,221.74) -- (145.01,220.54) -- cycle ;
	%Shape: Polygon [id:ds80744206474222] 
	\draw  [fill={rgb, 255:red, 0; green, 134; blue, 0}  ,fill opacity=0.5 ] (233.98,198.29) -- (295.3,209.11) -- (214.74,221.14) -- (157.03,209.11) -- cycle ;
	%Straight Lines [id:da27737977231870703] 
	\draw    (145.01,93.09) -- (284.48,158.02) ;
	%Straight Lines [id:da995141324177293] 
	\draw  [dash pattern={on 0.84pt off 2.51pt}]  (308.52,42.59) -- (173.86,133.97) ;
	%Straight Lines [id:da1588826299657745] 
	\draw    (156.43,115.93) -- (214.74,125.55) ;
	%Straight Lines [id:da054162386510361316] 
	\draw    (214.74,125.55) -- (296.5,100.3) ;
	%Straight Lines [id:da20945772577328303] 
	\draw  [dash pattern={on 0.84pt off 2.51pt}]  (236.58,93.09) -- (296.5,100.3) ;
	%Straight Lines [id:da8655504050443374] 
	\draw  [dash pattern={on 0.84pt off 2.51pt}]  (156.43,115.93) -- (236.58,93.09) ;
	%Straight Lines [id:da7509171119842623] 
	\draw  [dash pattern={on 4.5pt off 4.5pt}]  (145.01,93.09) -- (145.01,220.54) ;
	%Straight Lines [id:da5859423008566036] 
	\draw  [dash pattern={on 4.5pt off 4.5pt}]  (167.85,138.78) -- (167.85,197.69) ;
	%Straight Lines [id:da7967899176260733] 
	\draw  [dash pattern={on 4.5pt off 4.5pt}]  (284.48,158.02) -- (284.48,221.74) ;
	%Straight Lines [id:da06946544180443814] 
	\draw  [dash pattern={on 4.5pt off 4.5pt}]  (308.52,42.59) -- (308.52,197.69) ;
	%Straight Lines [id:da19550280942224496] 
	\draw  [dash pattern={on 4.5pt off 4.5pt}]  (214.74,125.55) -- (214.74,221.14) ;
	%Straight Lines [id:da6674296702452525] 
	\draw  [dash pattern={on 4.5pt off 4.5pt}]  (236.58,93.09) -- (236.58,198.29) ;
	%Straight Lines [id:da40753660412955905] 
	\draw  [dash pattern={on 4.5pt off 4.5pt}]  (156.43,115.93) -- (156.43,209.11) ;
	%Straight Lines [id:da5200959500289328] 
	\draw  [dash pattern={on 4.5pt off 4.5pt}]  (296.5,100.3) -- (296.5,209.72) ;
	\draw  [fill={rgb, 255:red, 0; green, 0; blue, 0 }  ,fill opacity=1 ] (149.82,180.86) -- (145.41,189.68) -- (141,180.86) -- (145.41,185.27) -- cycle ;
	\draw  [fill={rgb, 255:red, 0; green, 0; blue, 0 }  ,fill opacity=1 ] (160.64,168.84) -- (156.23,177.65) -- (151.82,168.84) -- (156.23,173.24) -- cycle ;
	\draw  [fill={rgb, 255:red, 0; green, 0; blue, 0 }  ,fill opacity=1 ] (171.46,156.81) -- (167.05,165.63) -- (162.64,156.81) -- (167.05,161.22) -- cycle ;
	\draw  [fill={rgb, 255:red, 0; green, 0; blue, 0 }  ,fill opacity=1 ] (219.55,180.86) -- (215.14,189.68) -- (210.73,180.86) -- (215.14,185.27) -- cycle ;
	\draw  [fill={rgb, 255:red, 0; green, 0; blue, 0 }  ,fill opacity=1 ] (289.29,180.86) -- (284.88,189.68) -- (280.47,180.86) -- (284.88,185.27) -- cycle ;
	\draw  [fill={rgb, 255:red, 0; green, 0; blue, 0 }  ,fill opacity=1 ] (301.31,168.84) -- (296.9,177.65) -- (292.49,168.84) -- (296.9,173.24) -- cycle ;
	\draw  [fill={rgb, 255:red, 0; green, 0; blue, 0 }  ,fill opacity=1 ] (313.33,156.81) -- (308.92,165.63) -- (304.52,156.81) -- (308.92,161.22) -- cycle ;
	\draw  [fill={rgb, 255:red, 0; green, 0; blue, 0 }  ,fill opacity=1 ] (241.19,156.81) -- (236.79,165.63) -- (232.38,156.81) -- (236.79,161.22) -- cycle ;
	%Shape: Ellipse [id:dp12534409482539088] 
	\draw  [draw opacity=0][fill={quantumviolet}  ,fill opacity=1 ] (141.49,93.09) .. controls (141.49,91.15) and (143.07,89.57) .. (145.01,89.57) .. controls (146.95,89.57) and (148.52,91.15) .. (148.52,93.09) .. controls (148.52,95.03) and (146.95,96.61) .. (145.01,96.61) .. controls (143.07,96.61) and (141.49,95.03) .. (141.49,93.09) -- cycle ;
	%Shape: Ellipse [id:dp20478671124531078] 
	\draw  [draw opacity=0][fill={quantumviolet}  ,fill opacity=1 ] (164.34,138.78) .. controls (164.34,136.84) and (165.91,135.26) .. (167.85,135.26) .. controls (169.79,135.26) and (171.37,136.84) .. (171.37,138.78) .. controls (171.37,140.72) and (169.79,142.29) .. (167.85,142.29) .. controls (165.91,142.29) and (164.34,140.72) .. (164.34,138.78) -- cycle ;
	%Shape: Ellipse [id:dp4893991229073703] 
	\draw  [draw opacity=0][fill={quantumviolet}  ,fill opacity=1 ] (280.96,158.02) .. controls (280.96,156.07) and (282.54,154.5) .. (284.48,154.5) .. controls (286.42,154.5) and (287.99,156.07) .. (287.99,158.02) .. controls (287.99,159.96) and (286.42,161.53) .. (284.48,161.53) .. controls (282.54,161.53) and (280.96,159.96) .. (280.96,158.02) -- cycle ;
	%Shape: Ellipse [id:dp16059827384283132] 
	\draw  [draw opacity=0][fill={quantumviolet}  ,fill opacity=1 ] (305.01,42.59) .. controls (305.01,40.65) and (306.58,39.07) .. (308.52,39.07) .. controls (310.47,39.07) and (312.04,40.65) .. (312.04,42.59) .. controls (312.04,44.53) and (310.47,46.11) .. (308.52,46.11) .. controls (306.58,46.11) and (305.01,44.53) .. (305.01,42.59) -- cycle ;
	%Shape: Ellipse [id:dp867808786136118] 
	\draw  [draw opacity=0][fill={quantumviolet}  ,fill opacity=1 ] (211.23,125.55) .. controls (211.23,123.61) and (212.8,122.04) .. (214.74,122.04) .. controls (216.68,122.04) and (218.26,123.61) .. (218.26,125.55) .. controls (218.26,127.49) and (216.68,129.07) .. (214.74,129.07) .. controls (212.8,129.07) and (211.23,127.49) .. (211.23,125.55) -- cycle ;
	%Shape: Ellipse [id:dp44059795188298023] 
	\draw  [draw opacity=0][fill={quantumviolet}  ,fill opacity=1 ] (152.91,115.93) .. controls (152.91,113.99) and (154.49,112.42) .. (156.43,112.42) .. controls (158.37,112.42) and (159.95,113.99) .. (159.95,115.93) .. controls (159.95,117.88) and (158.37,119.45) .. (156.43,119.45) .. controls (154.49,119.45) and (152.91,117.88) .. (152.91,115.93) -- cycle ;
	%Shape: Ellipse [id:dp7056669802068666] 
	\draw  [draw opacity=0][fill={quantumviolet}  ,fill opacity=1 ] (292.98,100.3) .. controls (292.98,98.36) and (294.56,96.79) .. (296.5,96.79) .. controls (298.44,96.79) and (300.02,98.36) .. (300.02,100.3) .. controls (300.02,102.25) and (298.44,103.82) .. (296.5,103.82) .. controls (294.56,103.82) and (292.98,102.25) .. (292.98,100.3) -- cycle ;
	%Shape: Ellipse [id:dp3103568644927319] 
	\draw  [draw opacity=0][fill={quantumviolet}  ,fill opacity=1 ] (233.07,93.09) .. controls (233.07,91.15) and (234.64,89.57) .. (236.58,89.57) .. controls (238.53,89.57) and (240.1,91.15) .. (240.1,93.09) .. controls (240.1,95.03) and (238.53,96.61) .. (236.58,96.61) .. controls (234.64,96.61) and (233.07,95.03) .. (233.07,93.09) -- cycle ;
	%Shape: Ellipse [id:dp06282522874034746] 
	\draw  [draw opacity=0][fill={rgb, 255:red, 0; green, 134; blue, 0}  ,fill opacity=1 ] (153.51,209.11) .. controls (153.51,207.17) and (155.09,205.6) .. (157.03,205.6) .. controls (158.97,205.6) and (160.55,207.17) .. (160.55,209.11) .. controls (160.55,211.06) and (158.97,212.63) .. (157.03,212.63) .. controls (155.09,212.63) and (153.51,211.06) .. (153.51,209.11) -- cycle ;
	%Shape: Ellipse [id:dp9863628897494656] 
	\draw  [draw opacity=0][fill={rgb, 255:red, 0; green, 134; blue, 0}  ,fill opacity=1 ] (141.49,220.54) .. controls (141.49,218.59) and (143.07,217.02) .. (145.01,217.02) .. controls (146.95,217.02) and (148.52,218.59) .. (148.52,220.54) .. controls (148.52,222.48) and (146.95,224.05) .. (145.01,224.05) .. controls (143.07,224.05) and (141.49,222.48) .. (141.49,220.54) -- cycle ;
	%Shape: Ellipse [id:dp7493602185825027] 
	\draw  [draw opacity=0][fill={rgb, 255:red, 0; green, 134; blue, 0}  ,fill opacity=1 ] (164.34,197.69) .. controls (164.34,195.75) and (165.91,194.18) .. (167.85,194.18) .. controls (169.79,194.18) and (171.37,195.75) .. (171.37,197.69) .. controls (171.37,199.63) and (169.79,201.21) .. (167.85,201.21) .. controls (165.91,201.21) and (164.34,199.63) .. (164.34,197.69) -- cycle ;
	%Shape: Ellipse [id:dp952306491449941] 
	\draw  [draw opacity=0][fill={rgb, 255:red, 0; green, 134; blue, 0}  ,fill opacity=1 ] (211.23,221.14) .. controls (211.23,219.19) and (212.8,217.62) .. (214.74,217.62) .. controls (216.68,217.62) and (218.26,219.19) .. (218.26,221.14) .. controls (218.26,223.08) and (216.68,224.65) .. (214.74,224.65) .. controls (212.8,224.65) and (211.23,223.08) .. (211.23,221.14) -- cycle ;
	%Shape: Ellipse [id:dp9218063427852546] 
	\draw  [draw opacity=0][fill={rgb, 255:red, 0; green, 134; blue, 0}  ,fill opacity=1 ] (280.96,221.74) .. controls (280.96,219.8) and (282.54,218.22) .. (284.48,218.22) .. controls (286.42,218.22) and (287.99,219.8) .. (287.99,221.74) .. controls (287.99,223.68) and (286.42,225.26) .. (284.48,225.26) .. controls (282.54,225.26) and (280.96,223.68) .. (280.96,221.74) -- cycle ;
	%Shape: Ellipse [id:dp41405547905371887] 
	\draw  [draw opacity=0][fill={rgb, 255:red, 0; green, 134; blue, 0}  ,fill opacity=1 ] (233.07,198.29) .. controls (233.07,196.35) and (234.64,194.78) .. (236.58,194.78) .. controls (238.53,194.78) and (240.1,196.35) .. (240.1,198.29) .. controls (240.1,200.24) and (238.53,201.81) .. (236.58,201.81) .. controls (234.64,201.81) and (233.07,200.24) .. (233.07,198.29) -- cycle ;
	%Shape: Ellipse [id:dp3263634248079793] 
	\draw  [draw opacity=0][fill={rgb, 255:red, 0; green, 134; blue, 0}  ,fill opacity=1 ] (305.01,197.69) .. controls (305.01,195.75) and (306.58,194.18) .. (308.52,194.18) .. controls (310.47,194.18) and (312.04,195.75) .. (312.04,197.69) .. controls (312.04,199.63) and (310.47,201.21) .. (308.52,201.21) .. controls (306.58,201.21) and (305.01,199.63) .. (305.01,197.69) -- cycle ;
	%Shape: Ellipse [id:dp010461505491899503] 
	\draw  [draw opacity=0][fill={rgb, 255:red, 0; green, 134; blue, 0}  ,fill opacity=1 ] (291.78,209.11) .. controls (291.78,207.17) and (293.36,205.6) .. (295.3,205.6) .. controls (297.24,205.6) and (298.82,207.17) .. (298.82,209.11) .. controls (298.82,211.06) and (297.24,212.63) .. (295.3,212.63) .. controls (293.36,212.63) and (291.78,211.06) .. (291.78,209.11) -- cycle ;
	
	% Text Node
	\draw (215,200) node [anchor=north west][inner sep=0.75pt]   [align=left] {\small $\displaystyle {\tau(\mathcal{E}_{\mathscr{h}})}$};
	% Text Node
	\draw (303,206) node [anchor=north west][inner sep=0.75pt]   [align=left] {\small $\displaystyle{\sigma(\Omega _{\mathscr{h}}})$};
	% Text Node
	\draw (269,66) node [anchor=north west][inner sep=0.75pt]   [align=left] {\small $\displaystyle {\Omega }_{\mathscr{h}}$};
	% Text Node
	\draw (207,101) node [anchor=north west][inner sep=0.75pt]   [align=left] {\small $\displaystyle {\mathcal{E}}_{\mathscr{h}}$};

\end{tikzpicture}
\bigskip
   \caption{Here we depict the Holevo fragment: states and effects from a simplicial theory whose shadow map (depicted as downward arrows) gives the gbit state and effect space, $\sigma(\Omega_\mathscr{h})=\Omega_\mathcal{G}$ and $\tau(\mathcal{E}_\mathscr{h})=\mathcal{E}_\mathcal{G}$, respectively. Note that we have only shown the (convex hull of the) four extremal effects that are not the unit and null effect, since depicting the full effect space would require four dimensions. }
   \label{fig:H1}
\end{figure}

As we prove in Appendix~\ref{sec:Holevogeneral}, the shadow of a fragment of a simplicial system constructed in this manner is given by taking the shadow maps to be $\tau=L$ and $\sigma=L_r^{-1}$, where $L_r^{-1}$ is a right inverse of $L$. 
 That this defines a shadow map will be justified below (by verifying that it leads to valid GPT system and that $\tau=L$ and $\sigma=L_r^{-1}$ are full).

The matrix representation of $L_r^{-1}$ is
\begin{align}
    L^{-1}_r = \frac{1}{4}\begin{bmatrix}
        1& 2& 0\\
        1& -2& 0\\
        1& 0& 2\\
        1& 0& -2\\
    \end{bmatrix},
\end{align}
since $L\circ L^{-1}_r=\mathds{1}_3$, as one can check. Now, we get the shadow states by applying $L$ to the states of $\Omega_\mathscr{h}$, i.e., $\sigma(\mu_i) = L\circ \mu_i\in\Omega_{\mathcal{S}(\mathscr{h})}$, $i=1,...,4$ (recall that, given our choice of basis, each $\mu_i$ is a member of the standard basis of $\mathds{R}^4$).
Finally,  we apply $L^{-1}_r$ (via pre-composition) to the effects in the Holevo fragment to get the effects $\tau(\xi_j)\in\mathcal{E}_{\mathcal{S}(\mathscr{h})}$, $j\in\{0,e_{13},e_{14},e_{23},e_{24},u\}$. Working this multiplication out for the pure objects gives 
\begin{align}\label{stateimage}
    \sigma(\mu_1) = \begin{bmatrix}
        1\\
        1\\
        0
    \end{bmatrix}\ns,\
    \sigma(\mu_2) = \begin{bmatrix}
        1\\
        -1\\
        0
    \end{bmatrix}\ns,\\
    \sigma(\mu_3)=\begin{bmatrix}
        1\\
        0\\
        1
    \end{bmatrix}\ns,\ 
    \sigma(\mu_4) = \begin{bmatrix}
        1\\
        0\\
        -1
    \end{bmatrix}\ns,\ 
\end{align}
for states, and gives
 \begin{align}\label{effectimage}
    [\tau(\xi_{e_{13}})]^T=\frac{1}{2}\begin{bmatrix}
        1\\
        1\\
        1
    \end{bmatrix}\ns,\ 
   [\tau(\xi_{e_{14}})]^T=\frac{1}{2}\begin{bmatrix}
        1\\
        1\\
        -1
    \end{bmatrix}\ns,\nonumber \\
    [\tau(\xi_{e_{23}})]^T = \frac{1}{2}\begin{bmatrix}
        1\\
        -1\\
        1,
    \end{bmatrix}\ns,\ 
    [\tau(\xi_{e_{24}})]^T = \frac{1}{2}\begin{bmatrix}
        1\\
        -1\\
        -1
    \end{bmatrix},
\end{align}
for effects, as well as $[\tau(\xi_0)]^T =\vec{0}$ and $[\tau(u)]^T=u$. Comparing these with Eq.~\eqref{eq:statesG} and Eq.~\eqref{eq:effectsG}, respectively, one sees that indeed $\Omega_{\mathcal{S}(\mathscr{h})}=\Omega_\mathcal{G}$ and $\mathcal{E}_{\mathcal{S}(\mathscr{h})}=\mathcal{E}_{\mathcal{G}}$, and the same probability rule is being used (namely the evaluation map, given by standard matrix multiplication), so the shadow $\mathcal{S}(\mathscr{h})$ is indeed $\mathcal{G}$.

One can now easily verify that taking $\tau=L$ and $\sigma=L_r^{-1}$ does indeed define a shadow map. It suffices to check that the image of these maps is a valid GPT system (which it is, as we just showed the image is $\mathcal{G}$), and that $\tau=L$ and $\sigma=L_r^{-1}$ define a full embedding. For $\sigma$ to be full means that $\sigma(\Omega_{\mathscr{h}})=\Omega_\mathcal{G}$, which is the case, as Eq.~\eqref{stateimage} shows that the image of the extremal states of $\Omega_{\mathscr{h}}$ coincides with the extremal points of $\Omega_\mathcal{G}$. Similarly, one can check that $\tau=L$ is full via Eq.~\eqref{effectimage}.

Because the Holevo fragment is not relatively tomographic, its shadow is not faithful and so introduces new operational identities. We can see this explicitly from the fact that a generic state in the tetrahedron $\Delta_3$ gets mapped under the shadow map as 
\begin{equation}
    \begin{bmatrix}
        p_1\\
        p_2\\
        p_3\\
        p_4
    \end{bmatrix} 
    \ \ \rightarrow \ \
    \begin{bmatrix}
        1\\
        p_1-p_2\\
        p_3-p_4
    \end{bmatrix}.
\end{equation}
So, two different states $p$ and $q$ in $\Delta_3$ that happen to have components obeying $p_1-p_2=q_1-q_2$ and $p_3-p_4=q_3-q_4$ will get mapped to the same state in $\Omega_{\mathcal{S}(\mathscr{h})}$. This identification of different states occurs because the effects in $\mathcal{E}_\mathcal{G}$ that are able to distinguish these distributions are absent from $\mathcal{E}_\mathscr{h}$. 
Thus, experiments that fail to have access to these distinguishing effects lead one to the mistaken identification of fundamentally distinct states, which in turn leads one to believe that the GPT system in question is a gbit, even though it is by assumption a four-level classical system.
This is an extreme example, in that gbits are known to be maximally contextual and compose to reach maximal violations of Bell-CHSH inequalities.

\subsection{Shadow maps cannot introduce simplex-embeddability}

One might also like to see an example of a fragment that is not simplex-embeddable, but whose shadow is simplex-embeddable. However, this is impossible.

\begin{prop}\label{classicalityprop}
    Consider {\em any} fragment $f$ of a GPT system $G$ (even one whose states and effects are not tomographic for each other). If the shadow of $f$ is simplex-embeddable, then $f$ is simplex-embeddable.
\end{prop}

The proof follows immediately from transitivity of GPT embeddings:
since every fragment embeds into its shadow, and since the shadow by assumption embeds into a simplicial GPT system, the fragment itself embeds into a simplicial theory. 

We give an example of a shadow map that dramatically distorts the fragment in question while nonetheless preserving simplex-embeddability in the next section.

Theorem~\ref{thm:Holevo} and Proposition~\ref{classicalityprop} (respectively) can be summarized as follows: {\em Shadow maps can break simplex-embeddability, but cannot introduce it.}

\subsection{Example 1 redux: a nontrivial shadow map that preserves simplex embeddability}\label{sec:redux}

Finally, reconsider the example from Section~\ref{example1}, where the shadow map acts in a highly nontrivial manner, so that (like in the Holevo example) the shadow is very different from the original GPT fragment. But despite this (and in contrast to the Holevo example), simplex-embeddability is nonetheless preserved by the shadow map in this example.  

Consider again the fragment of a qubit containing all possible states in the Bloch ball and containing only those effects in the convex hull of the projectors $\ket{0}\bra{0}$ and $\ket{1}\bra{1}$ (besides the zero and unit effects). This fragment is depicted in Figure~\ref{fig:exSimpP2} (which is simply a repeat of Figure~\ref{fig:exSimpP}). 

\begin{figure}[htb!]
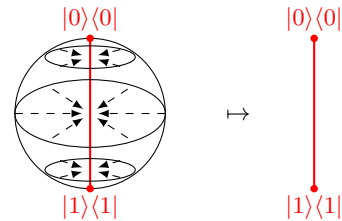

   \centering
   \begin{tabular}{ccc}
    \btp
    \begin{pgfonlayer}{nodelayer}
        \node[shape=circle, fill=red, inner sep =1pt] (0) at (0,2) {};
        \node[shape=circle, fill=red, inner sep =1pt] (1) at (0,-2) {};
        \node [style=none,font=\color{red}] (label0) at (0,2.5) {$\ket{0}\!\!\bra{0}$};
        \node [style=none,font=\color{red}] (label1) at (0,-2.5) {$\ket{1}\!\!\bra{1}$};
        \node [style=none] (aux1) at (2,0) {};
        \node [style=none] (aux2) at (-2,0) {};
        \node[style=none] (aux3) at (1,-0.65) {};
        \node[style=none] (aux4) at (1,0.65) {};
        \node [style=none] (aux5) at (-1,0.65) {};
        \node [style=none] (aux6) at (-1,-0.65) {};
        \node [style=circle,fill=none,inner sep=2pt] (centre) at (0,0) {};
        \node [style=circle,fill=none,inner sep=2pt] (centre2) at (0,1.5) {};
        \node [style=none] (aux3-2) at (0.8,1.25) {};
        \node [style=none] (aux4-2) at (-0.8,1.25) {};
        \node [style=none] (aux5-2) at (0.8,1.7) {};
        \node [style=none] (aux6-2) at (-0.8,1.7) {};
        \node [style=circle,fill=none,inner sep=2pt] (centre3) at (0,-1.5) {};
        \node [style=none] (aux3-3) at (0.8,-1.25) {};
        \node [style=none] (aux4-3) at (-0.8,-1.25) {};
        \node [style=none] (aux5-3) at (0.8,-1.7) {};
        \node [style=none] (aux6-3) at (-0.8,-1.7) {};
    \end{pgfonlayer}
    \begin{pgfonlayer}{edgelayer}
        \draw [fill=none](0,0) circle (2.0) node [black] {};
        \draw [thick,color=red] (0.center) to (1.center) {};
        \draw [fill=none] (0,0) ellipse (2cm and 0.9cm) {};
        \draw [fill=none] (0,1.5) ellipse (1.2cm and 0.3cm) {};
        \draw [fill=none] (0,-1.5) ellipse (1.2cm and 0.3cm) {};
        \draw [dashed,->, >=latex] (aux1.center) to (centre);
        \draw [dashed,->, >=latex] (aux2.center) to (centre);
        \draw [dashed,->,>=latex] (aux3.center) to (centre);
        \draw [dashed,->,>=latex] (aux4.center) to (centre);
        \draw [dashed,->,>=latex] (aux5.center) to (centre);
        \draw [dashed,->,>=latex] (aux6.center) to (centre);
        \draw [dashed,->,>=latex] (aux3-2.center) to (centre2);
        \draw [dashed,->,>=latex] (aux4-2.center) to (centre2);
        \draw [dashed,->,>=latex] (aux5-2.center) to (centre2);
        \draw [dashed,->,>=latex] (aux6-2.center) to (centre2);
        \draw [dashed,->,>=latex] (aux3-3.center) to (centre3);
        \draw [dashed,->,>=latex] (aux4-3.center) to (centre3);
        \draw [dashed,->,>=latex] (aux5-3.center) to (centre3);
        \draw [dashed,->,>=latex] (aux6-3.center) to (centre3);
    \end{pgfonlayer}
\etp & \quad\quad$\mapsto$\quad\quad &
\btp
\begin{pgfonlayer}{nodelayer}
    \node[shape=circle, fill=red, inner sep =1pt] (0) at (0,2) {};
        \node[shape=circle, fill=red, inner sep =1pt] (1) at (0,-2) {};
        \node [style=none,font=\color{red}] (label0) at (0,2.5) {$\ket{0}\!\!\bra{0}$};
        \node [style=none,font=\color{red}] (label1) at (0,-2.5) {$\ket{1}\!\!\bra{1}$};
\end{pgfonlayer}
\begin{pgfonlayer}{edgelayer}
    \draw [thick,color=red] (0.center) to (1.center) {};
\end{pgfonlayer}
\etp\\
   \end{tabular}
   \vspace{10pt}
   \caption{Consider the fragment of quantum theory whose state space contains all states in the Bloch ball and whose effect space is the set of all effects in the line segment in red, besides the zero and unit effects. A GPT shadow of this fragment can be constructed by orthogonally projecting the set of states onto the line segment. This results in a shadow which is simply the GPT representing a classical bit.}
   \label{fig:exSimpP2}
\end{figure}

One can compute a shadow of this fragment by orthogonally projecting the state space onto the line segment defined by the effects, as shown in the figure. Like in the projections in the previous section, this linearly maps the states and effects into a proper GPT system (here, the classical bit), while preserving the probabilities assigned (by the Born rule) to any state-effect pair. Consequently, it constitutes a valid shadow map. 

The shadow constitutes all and only the states and effects of a classical bit. Therefore it is simplex-embeddable (and indeed is itself simplicial). 

As one can check explicitly (e.g., via a linear program~\cite{selby2024linear}), the original GPT fragment in this example is simplex-embeddable.  So this constitutes an example where the shadow map takes a simplex-embeddable fragment to a simplex embeddable GPT system, despite dramatically altering the fragment.

\section{Implications for noncontextuality}\label{sec:NonProblematicFailureTomography}	

At this stage, readers familiar with the notion of simplex-embeddability may already see how the above results are relevant to the study of nonclassicality. However, we will emphasize and make more explicit these connections in this section.

In the literature, it is often claimed that assessments of noncontextuality rest on the assumption of tomographic completeness. However, our work implies that they rest on a weaker assumption, namely that the states and effects in the experimentally realized fragment are {\em relatively tomographic},  regardless of whether they are tomographic for the full set of effects and states of the true GPT system. We term this the {\em relative tomographic completeness} assumption.

Recall from Section~\ref{example1} that when the states and effects in one's experiment are not relatively tomographic, the statistics one observes in any experiment with them will be insufficient to correctly characterize them. More specifically, some states and/or effects that are distinct in the true GPT governing the experiment will appear to be equivalent for all the data obtains in that experiment. 
Thus, in any experiment that makes use only of this data (as opposed to, for example, appealing to some theory-dependent arguments concerning the inner workings of one's laboratory apparatuses), any attempt to characterize the GPT state and effect vectors describing one's preparations and measurements in the experiment will in fact characterize a shadow of the fragment, rather than the fragment itself. We expand on this idea in the next subsection.

If the fragment in one's experiment is relatively tomographic, then this cannot lead to any problems, since characterizing such a fragment's shadow is equivalent to characterizing the true fragment. However, when relative tomography does fail, the mischaracterizations of one's GPT states and effects (and equivalently, of the operational identities that they obey) described by the shadow map can lead to incorrect assessments of nonclassicality (simplex-embeddability).
However, it does not {\em necessarily} imply that one will reach the wrong conclusion about nonclassicality; the specific details of the mischaracterization---the details of how the shadow map distorts the true fragment---will determine whether or not this happens. Later subsections will expand on the different possibilities, which are summarized in Table~\ref{EmbedTable}.

\begin{table}[h!]
\centering

\tikzset{every picture/.style={line width=0.75pt}} %set default line width to 0.75pt        

\begin{tikzpicture}[x=0.75pt,y=0.75pt,yscale=-1,xscale=1]
	%uncomment if require: \path (0,300); %set diagram left start at 0, and has height of 300
	
	%Shape: Rectangle [id:dp6635465653642423] 
	\draw  [fill=%{rgb, 255:red, 155; green, 155; blue, 155}
	{quantumgray}  ,fill opacity=0.5] (126.21,124.41) -- (196.09,124.41) -- (196.09,202) -- (126.21,202) -- cycle ;
	%Shape: Rectangle [id:dp20611827882093403] 
	\draw  [fill={quantumgray}  ,fill opacity=0.5 ] (126.21,70.3) -- (196.09,70.3) -- (196.09,124.41) -- (126.21,124.41) -- cycle ;
	%Shape: Rectangle [id:dp6827060097225476] 
	\draw  [fill={quantumgray}  ,fill opacity=0.3] (196.09,70.3) -- (318.45,70.3) -- (318.45,124.41) -- (196.09,124.41) -- cycle ;
	%Shape: Rectangle [id:dp7926206266486929] 
	\draw  [fill={quantumgray}  ,fill opacity=0.3] (196.09,124.41) -- (318.45,124.41) -- (318.45,202) -- (196.09,202) -- cycle ;
	%Shape: Rectangle [id:dp11349820126421617] 
	\draw  [fill={quantumgray}  ,fill opacity=0.3] (318.45,124.41) -- (427.33,124.41) -- (427.33,202) -- (318.45,202) -- cycle ;
	%Shape: Rectangle [id:dp7708579858040333] 
	\draw  [fill={black}  ,fill opacity=1] (318.45,70) -- (427.33,70) -- (427.33,124.41) -- (318.45,124.41) -- cycle ;
	%Shape: Rectangle [id:dp06895088003719985] 
	\draw  [fill={quantumgray}  ,fill opacity=0.5] (318.45,44) -- (427.33,44) -- (427.33,70) -- (318.45,70) -- cycle ;
	%Shape: Rectangle [id:dp9176681147247345] 
	\draw  [fill={quantumgray}  ,fill opacity=0.5] (196.09,44) -- (318.45,44) -- (318.45,70) -- (196.09,70) -- cycle ;
	
	% Text Node
	\draw (203,46) node [anchor=north west][inner sep=0.75pt]  [font=\scriptsize] [align=left] {\begin{minipage}[lt]{80pt}\setlength\topsep{0pt}
			\begin{center}
			\textbf{Shadow is embeddable}
		\end{center}
	\end{minipage}};
	% Text Node
	\draw (320,45) node [anchor=north west][inner sep=0.75pt]  [font=\scriptsize] [align=left] {\begin{minipage}[lt]{80pt}\setlength\topsep{0pt}
			\begin{center}
				\textbf{Shadow is non-embeddable}
			\end{center}
	\end{minipage}};
	% Text Node
	\draw (128,83) node [anchor=north west][inner sep=0.75pt]  [font=\scriptsize] [align=left] {\begin{minipage}[lt]{42pt}\setlength\topsep{0pt}
			\begin{center}
				\textbf{Fragment is embeddable}
			\end{center}
	\end{minipage}};
	% Text Node
	\draw (127,139.87) node [anchor=north west][inner sep=0.75pt]  [font=\scriptsize] [align=left] {\begin{minipage}[ct]{42pt}\setlength\topsep{0pt}
			\begin{center}
				\textbf{Fragment is non-\\embeddable}
			\end{center}
	\end{minipage}};
	% Text Node
	\draw (200.02,74) node [anchor=north west][inner sep=0.75pt]  [font=\scriptsize] [align=left] {\begin{minipage}[lt]{84pt}\setlength\topsep{0pt}
	- Any tomographic fragment of a simplicial theory;\\ 
	- Our example in Fig.~\ref{fig:exSimpP2}.
	\end{minipage}};
	% Text Node
	\draw (217.99,142.57) node [anchor=north west][inner sep=0.75pt]  [font=\scriptsize] [align=left] {\begin{minipage}[lt]{48.35pt}\setlength\topsep{0pt}
			\begin{center}
				IMPOSSIBLE\\(Proposition~\ref{classicalityprop})
			\end{center}
			
	\end{minipage}};
	% Text Node
	\draw (320,125) node [anchor=north west][inner sep=0.75pt]  [font=\scriptsize] [align=left] {\begin{minipage}[lt]{75pt}\setlength\topsep{0pt}
		- A qubit;\\
		- Any tomographic non-embeddable fragment;\\
		- Our example in Fig~\ref{fig:twodisk}.
	\end{minipage}};
	% Text Node
	\draw (323.09,82) node [anchor=north west][inner sep=0.75pt]  [font=\scriptsize] [align=left] {\textcolor[rgb]{1,1,1}{- Holevo's example }\\\textcolor[rgb]{1,1,1}{\changelinkcolor{SkyBlue} (Sec.~\ref{HolevoSec}).\changelinkcolor{carmine}}};

\end{tikzpicture}

\vspace{1mm}
\caption{The different possible cases for how a shadow map affects the simplex-embeddability properties of a given GPT fragment, and examples of each. Mistaken assessments of nonclassicality occur {\em only} in the top right case (the cell with a black background).}
\label{EmbedTable}
\end{table}

\subsection{Theory-agnostic tomography}\label{sec:TheoryAgnosticTomography}

We can exemplify and formalize the fact that theory-independent experiments characterize the shadow of one's true GPT fragment (rather than the true fragment itself) by considering theory-agnostic tomography~\cite{mazurek2021experimentally,grabowecky2021experimentally}. This is the state-of-the-art method for determining the characterization of GPT states and effects in a prepare-measure experiment. 

We formalize theory-agnostic tomography using a novel 
diagrammatic representation. Note that in this section (and the associated Appendix~\ref{sec:TheoryAgnosticTomographyProof}), we consider linear maps whose inputs (outputs) are not associated with a GPT system, and so are labeled by the domain (codomain) of the map (whereas in the rest of the paper, systems were labeled by the associated GPT system).

Imagine one has a fragment $f=(\Omega_f,\mathcal{E}_f,p_f)$ of some GPT system $G$ representing some states and effects which are performed in an experiment. As discussed in Section~\ref{sec:frags}, no real experiment will actually implement every state and effect in the fragment, as this set is infinite; however, an experiment involving the convexly extremal states and effects of a fragment is sufficient to uniquely infer the existence of---and probabilities generated by---any states and effects in the convex hull of these. (Moreover, in a real, lossy experiment one will not have exactly normalized states, so a noise-robust analysis must be done to contend with this.) 

Given a fragment $f$, we can construct a data table $D_f$ in which we index columns by convexly extremal states and rows by convexly extremal effects, so that the entries in the data table are the probability of obtaining the given effect on the given state. We denote the entries of the data table
\begin{equation} \label{datatabfrag}
    [D_f]_{i,j}:=p_f(s_i,e_j),
\end{equation}
where the $s_i$ are the $n$ convexly-extremal GPT states in the fragment, the $e_j$ are the $m$ convexly-extremal GPT effects in the fragment, and $[D_f]_{i,j}$ is a matrix element of $D_f\in M_{m\times n}(\mathds{R})$, where $M_{m\times n}(\mathds{R})$ is the set of real $m \times n$ matrices.
We refer to this as the data table generated by the fragment. \footnote{This can be viewed as a map from fragments to data tables, which we denote by $\mathcal{D}:|\mathtt{GPT-Fragment}|\to |\mathtt{DataTable}|::f\mapsto D_f$. }

\begin{definition}[Theory-agnostic tomography] 
\label{def:TheoryAgnosticTomography}
Given a data table $D\in M_{m\times n}(\mathds{R})$, consider the following procedure:
 \begin{enumerate}
     \item Find the smallest $k$ such that\footnote{For data tables generated by real (noisy) experiments, this matrix factorization will not be exact, and one must take a more sophisticated approach to find the best-fit factorization~\cite{mazurek2021experimentally,grabowecky2021experimentally}.}
     	\beq
\begin{tikzpicture}
	\begin{pgfonlayer}{nodelayer}
		\node [style=small map] (0) at (-0, 0) {$D$};
		\node [style=none] (1) at (-0, 1) {};
		\node [style=none] (2) at (-0, -1) {};
		\node [style=right label] (0) at (-0, 1) {$\mathbb{R}^m$};
		\node [style=right label] (4) at (-0, -1) {$\mathbb{R}^n$};
	\end{pgfonlayer}
	\begin{pgfonlayer}{edgelayer}
		\draw (1.center) to (2.center);
	\end{pgfonlayer}
\end{tikzpicture}}  \quad = \quad %
\begin{tikzpicture}
	\begin{pgfonlayer}{nodelayer}
		\node [style=proj,fill=gray!30] (0) at (0, -0.75) {$S$};
		\node [style=none] (1) at (0, 1.5) {};
		\node [style=none] (2) at (0, -1.5) {};
		\node [style=right label] (3) at (0, 1.5) {$\mathbb{R}^m$};
		\node [style=right label] (4) at (0, -1.5) {$\mathbb{R}^n$};
		\node [style=inc,fill=gray!30] (5) at (0,0.75) {$E$};
		\node [style=right label] (6) at (0,0) {$\mathbb{R}^k$};
	\end{pgfonlayer}
	\begin{pgfonlayer}{edgelayer}
		\draw (1.center) to (5);
		\draw (5) to (0);
		\draw (0) to (2.center);
	\end{pgfonlayer}
\end{tikzpicture}},
		\eeq
      where $E\in M_{m\times k}(\mathds{R})$ and $S\in M_{k\times n}(\mathds{R})$;
     
     \item Consider the basis of unit vectors $\left\{%
\begin{tikzpicture}
	\begin{pgfonlayer}{nodelayer}
		\node [style=none] (1) at (-0, 0.5) {};
		\node [style=Wsquareadj] (2) at (-0, -0.5) {$i$};
		\node [style=right label] (3) at (-0, 0.5) {$\mathbb{R}^n$};
	\end{pgfonlayer}
	\begin{pgfonlayer}{edgelayer}
		\draw (1.center) to (2.center);
	\end{pgfonlayer}
\end{tikzpicture}}\right\}_i$ 
     of $\mathds{R}^n$. Then, for each $i\in \{1,\ldots,n\}$, denote by $\bar{s}_i$ the $i$th column of matrix S, namely

     \beq
     \btp
        \begin{pgfonlayer}{nodelayer}
        \node [style=Wsquareadj] (0) at (0,-0.6) {$\bar{s}_i$};
        \node [style=none] (1) at (0,0.6) {};
        \node [style=right label] (2) at (0,0.4) {$\mathbbm{R}^k$};
    \end{pgfonlayer}
    \begin{pgfonlayer}{edgelayer}
        \draw (0) to (1.center);
    \end{pgfonlayer}
\etp
\quad:=\quad
\btp
    \begin{pgfonlayer}{nodelayer}
        \node [style=none] (0) at (0,1.2) {};
        \node [style=proj,fill=gray!30] (1) at (0,0.3) {$S$};
        \node [style=Wsquareadj] (2) at (0,-1.2) {$i$};
        \node [style=right label] (3) at (0,1.2) {$\mathbbm{R}^k$};
        \node [style=right label] (4) at (0,-0.6) {$\mathds{R}^n$};
    \end{pgfonlayer}
    \begin{pgfonlayer}{edgelayer}
        \draw (0.center) to (1);
        \draw (1) to (2);
    \end{pgfonlayer}
\etp.
\eeq

     Similarly, consider the basis of unit covectors  $\left\{\btp
	\begin{pgfonlayer}{nodelayer}
		\node [style=none] (1) at (-0, -0.5) {};
		\node [style=Wsquare] (2) at (-0, 0.5) {$j$};
		\node [style=right label] (3) at (-0, -0.5) {$\mathbbm{R}^m$};
	\end{pgfonlayer}
	\begin{pgfonlayer}{edgelayer}
		\draw (1.center) to (2.center);
	\end{pgfonlayer}
\etp\right\}_j$
for $\mathbbm{R}^m$. For each $j\in\{1,\ldots,m\}$, denote by $\bar{e}_j$ the $j$th row of matrix E, i.e.,
\beq
\btp
        \begin{pgfonlayer}{nodelayer}
            \node [style=none] (0) at (0,-0.6) {};
            \node [style=Wsquare] (1) at (0,0.6) {$\bar{e}_j$};
            \node [style=right label] (2) at (0,-0.4) {$\mathds{R}^k$};
        \end{pgfonlayer}
        \begin{pgfonlayer}{edgelayer}
            \draw (0.center) to (1);
        \end{pgfonlayer}
    \etp
    \quad:=\quad
    \btp
    \begin{pgfonlayer}{nodelayer}
        \node [style=Wsquare] (0) at (0,1.2) {$j$};
        \node [style=inc, fill=gray!30] (1) at (0,-0.1) {$E$};
        \node [style=none] (2) at (0,-1.2) {};
        \node [style=right label] (3) at (0,0.6) {$\mathbbm{R}^m$};
        \node [style=right label] (4) at (0,-0.9) {$\mathds{R}^k$};
    \end{pgfonlayer}
    \begin{pgfonlayer}{edgelayer}
        \draw (0.center) to (1);
        \draw (1) to (2);
    \end{pgfonlayer}
    \etp;
\eeq
     \item Define a GPT ${G_D:=(\Omega_{G_D},\mathcal{E}_{G_D},p_{G_D})}$ by 
\begin{align}
		\bar{\Omega}_{G_D}:=\mathsf{Conv}\left[\left\{\btp
        \begin{pgfonlayer}{nodelayer}
        \node [style=Wsquareadj] (0) at (0,-0.6) {$\bar{s}_i$};
        \node [style=none] (1) at (0,0.6) {};
        \node [style=right label] (2) at (0,0.4) {$\mathbbm{R}^k$};
    \end{pgfonlayer}
    \begin{pgfonlayer}{edgelayer}
        \draw (0) to (1.center);
    \end{pgfonlayer}
\etp\right\}_{i=1}^n\right];\ \\
		\mathcal{E}_{G_D}:=\mathsf{Conv}\left[\left\{\btp
        \begin{pgfonlayer}{nodelayer}
            \node [style=none] (0) at (0,-0.6) {};
            \node [style=Wsquare] (1) at (0,0.6) {$\bar{e}_j$};
            \node [style=right label] (2) at (0,-0.4) {$\mathds{R}^k$};
        \end{pgfonlayer}
        \begin{pgfonlayer}{edgelayer}
            \draw (0.center) to (1);
        \end{pgfonlayer}
    \etp\right\}_{j=1}^m\right],
  \end{align}
 \end{enumerate}
where the probability rule $p_{G_D}$ is the  evaluation map. 

This procedure of theory-agnostic tomography takes as input a data table $D$ and constructs the GPT representations of the states and effects in the GPT system $G_D$ of minimal dimension compatible with $D$.\footnote{This defines a map from data tables to GPT systems which we denote by $\mathcal{T}:|\mathtt{DataTable}|\to|\mathtt{GPT-System}|::D\mapsto G_D$.} 
 \end{definition} 

Note that $\Omega_{G_D}$ will necessarily be a valid state space, since (by construction) the data table was generated from a collection of normalised states. Similarly $\mathcal{E}_{G_D}$ will necessarily be a valid effect space, since (by construction) the data table was generated from a collection of effects containing the complement to every effects in the collection, the null effect, and the unit effect.\footnote{If these latter conditions were not satisfied, then in the last stage of theory-agnostic tomography, one would need to add the unit effect and null effect to $\mathcal{E}_G$ and then close under coarse-graining in order to obtain a valid effect space.}

We now prove that theory-agnostic tomography does not directly characterize the GPT fragment $f$ that generated one's data table $D_f$, but rather the {\em shadow} of that fragment.  That is, the output $\mathcal{S}^{\text{TA}}(f):=G_{D_f}$ of theory-agnostic tomography, applied to a data table generated by fragment $f$, is a shadow for every fragment $f$. \footnote{That is, the composite map $\mathcal{S}^{\text{TA}}:=\mathcal{T}\circ\mathcal{D}:|\mathtt{GPT-Fragment}|\to|\mathtt{GPT-System}|$ acts as a shadow map for every fragment.} 
This was not explicitly recognized previously, but is critical to understanding when theory-agnostic tomography does and does not work as expected.  Prior works did recognize that theory-agnostic tomography is only guaranteed to be accurate under an assumption of tomographic completeness~\cite{mazurek2016experimental,mazurek2021experimentally,schmid2024addressing}; however, our work demonstrates that what is really required is {\em relative} tomographic completeness, and moreover formalizes what precisely goes wrong when this assumption fails. (Incidentally, the following theorem also constitutes a formal proof that the output of GPT tomography is always a valid GPT system.)

\begin{restatable}{theorem}{theoryagnostictomography} \label{thm:TheoryAgnosticTomography}
Consider the data table generated by some GPT fragment. The output of theory-agnostic tomography (as in Definition~\ref{def:TheoryAgnosticTomography}) applied to this data table is the GPT shadow of the GPT fragment.  That is, for all fragments $f$ we have that $\mathcal{S}^{\textrm{TA}}(f)=G_{D_f}$ is a valid shadow of $f$. 
\end{restatable}

We prove this theorem in Appendix~\ref{sec:TheoryAgnosticTomographyProof}.

Consequently, correct assessments of noncontextuality occur if and only if the simplex-embeddability property of the GPT fragment describing one's experiment is identical to the simplex-embeddability property of it shadow.

\subsection{Assessments of (non)contextuality for GPT subsystems of the true GPT are always accurate}\label{sec:mainresult}

A critical result is the following.
\begin{theorem}\label{keyresult}
Consider a fragment $f$ of the true GPT system $G$. If $f$ is a GPT subsystem of $G$, then the shadow of $f$ is nonclassical if and only if $f$ is nonclassical.
\end{theorem}

The proof is simple. Since $f$ is a GPT subsystem, its states and effects must be relatively tomographic; consequently, $f$ is equivalent to its shadow (by Lemma~\ref{tomogshadows}), and so either both are classical (simplex-embeddable), or neither is.

Together with the transitivity of GPT subsystems (Proposition~\ref{prop: transitivity}), this implies that proving nonclassicality for a GPT subsystem is sufficient to establish nonclassicality of the full GPT system containing it. That such a result should hold for subsystems in the sense of the parallel composition operation $\otimes$ is a basic consistency constraint that any sensible notion of classicality must satisfy, as it means that GPT tensor product subsystems of classically explainable (simplex-embeddable) GPTs are classically explainable (simplex-embeddable).\footnote{This also helps one understand the lab notebook argument of Ref.~\cite{schmid2024addressing}: there, establishing nonclassicality for a physical system is sufficient for establishing nonclassicality for that system taken together with any other physical systems, including those keeping classical records of what experimental procedures were implemented in the experiment.}  However, these results also apply to the wide class of other GPT subsystems defined and discussed in Section~\ref{ncofcomp}, including superselection sectors of a composite, virtual subsystems, stabilizer subsystems, and so on. 

In summary: in noncontextuality experiments, what matters is not tomographic completeness of states and effects in some absolute sense, but tomographic completeness of states and effects {\em relative to each other}. They need not be tomographically complete for the fundamental system that causally mediates correlations between the preparations and measurements.

Thus, for example, one can genuinely prove nonclassicality using states and effects of a qubit even without ever accessing the Pauli-$Y$ axis of the Bloch ball. Problems could only arise if either the states or the effects---{\em but not both}---had nontrivial expectation value on $Y$. One could similarly prove nonclassicality using some subset of the modes of an oscillator. 
In more extreme examples, one might even be probing some collective (coarse-grained) degree of freedom of a massive collection of more fundamental systems. This too poses no problem whatsoever to the goal of accurately assessing nonclassicality, as long as one has the ability to single out some set of states and effects on the {\em same} GPT subsystem---i.e., that are relatively tomographic.

Most importantly, this explains why we expect assessments of noncontextuality to hold up even if our current theories (e.g., quantum theory) are emergent rather than fundamental. We return to this point in Section~\ref{beyondqt}.

\subsection{Proofs of classical-explainability are robust to arbitrary failures of tomographic completeness}

Recall from Proposition~\ref{classicalityprop} that shadow maps cannot take simplex-embeddable fragments to non-simplex-embeddable fragments.

Thus, for {\em any} set of states and effects, a proof of classical-explainability is robust. So if one does theory-agnostic tomography and finds a classically-explainable GPT, one can be certain that the actual experiment is indeed classically-explainable. This is typically much less useful than, e.g., Theorem~\ref{keyresult}, since one typically wishes to certify {\em non}classicality rather than classicality. And it may still be the case that the {\em full} GPT system in question is nonclassical, but that this can only be witnessed by a more comprehensive experiment.

\subsection{Necessary conditions for proofs of nonclassicality to be robust to failures of tomographic completeness?}

Given this, the natural question becomes to characterize {\em exactly} when shadow maps do and do not cause simplex-embeddable fragments to become non-simplex-embeddable.\footnote{A still more ambitious question would be to characterize how taking the shadow map {\em quantitatively}~\cite{selby2024linear} affects simplex-embedding. } Can one give necessary and sufficient conditions (presumably geometric ones) for a GPT fragment to be simplex-embeddable if and only if its shadow is? This would constitute necessary and sufficient conditions for one's assessments of nonclassicality to be correct.

By Theorem~\ref{keyresult}, a {\em sufficient} condition for simplex-embedding to be preserved by a shadow map is that the fragment in question is relatively tomographic---a GPT subsystem of the true GPT system. But this is not a necessary condition, as evidenced by the example in Section~\ref{sec:redux}, where we exhibited a nontomographic fragment that was simplex-embeddable and whose shadow was also simplex-embeddable (despite being quite different from the fragment itself). 

We suspect that the results of Ref.~\cite{PuseydelRio} (showing that one can certify proofs of nonclassicality even in the presence of spurious operational equivalences) are also closely connected to the question raised in this section. We leave this connection to future work.

\subsection{Another useful perspective: operational identities for GPT subsystems are always genuine}

In this section, we essentially repeat the conclusions of Section~\ref{sec:mainresult}, but in a language that may be useful for readers who are more familiar with noncontextuality no-go theorems based on operational identities (or operational equivalences) than those based on simplex-embedding of GPTs.

Imagine one attempts to witness the failure of noncontextuality in an experiment on some particular physical system whose true GPT description is $G$. In any real experiment, one will only access a fragment of the states and effects of the system. From the statistics observed in the experiment, one must infer some operational identities.   If one wishes to use the statistics of such an experiment to prove the failure of noncontextuality, then it is often said that one must implement a set of states and effects that are tomographically complete for the system. 

But in fact, one can base such proofs on a strictly weaker assumption---that the operational identities used in one's proof are valid identities in the true GPT ($G$).  This is simply because Leibniz's principle (in the methodological form introduced by Spekkens~\cite{Leibniz}) can be applied to any processes that are indistinguishable in principle. Showing the indistinguishability of two processes relative to a set of procedures that are tomographically complete for the physical system in one's experiment is a {\em sufficient but not necessary} condition for in principle indistinguishability. 

For instance, consider any proof of nonclassicality built on top of operational equivalences holding among states of a single quantum bit. Such proofs do not suddenly fail to be valid simply because one implements them experimentally using a physical  qubit that is a tensor subsystem of a pair of qubits. This  is an obvious prerequisite to proofs of nonclassicality being meaningful in the first place, since all systems are subsystems of some larger system (e.g., of the universe). But there are less trivial examples as well. One could equally well implement a proof of nonclassicality using a virtual qubit defined as a coarse-grained degree of freedom of some large collection of physical qubits. 
Or, one could use only the states and effects lying in a plane of the Bloch sphere (i.e., using a rebit). 

Since the operational identities used in one's proof are necessarily extracted from one's experimental data, the question becomes, under what conditions are the apparent operational identities actually valid in the true GPT? Using the notion of a GPT subsystem,  one can give a simple necessary and sufficient condition.

\begin{lemma}
Consider a fragment $f$ of some fundamental GPT system $G$. If $f$ is a GPT subsystem, then all operational identities that hold for its states (effects) relative to its effects (states) are {\em genuine}. That is, they hold also relative to the full set of effects (states) in the true GPT system $G$.
\end{lemma}

This is because $f$ being a GPT subsystem implies that its states and effects are relatively tomographic, and so can be discriminated with {\em or without} the effects and states in $G$ that are not also in $f$.

For instance, consider again an experiment on a qubit that only accesses states and effects in the rebit. Although these states and effects would not be tomographically complete for the physical system (the qubit), this would not lead to any invalid operational equivalences, since the rebit states and rebit effects are relatively tomographic.

\subsection{Theories that hyperdecohere to quantum theory are nonclassical}\label{beyondqt}
 
State-of-the-art proofs of the failure of noncontextuality are theory-independent, in the sense that they do not assume the validity of quantum theory. Thus, like Bell's theorem, they provide a constraint on nature, independently of whether or not some future theory supersedes quantum theory. (See Section IV.B of Ref.~\cite{schmid2024addressing} for a more detailed discussion of the similarities between these two kinds of arguments.)

However, one might be concerned that this theory-independence is undermined by the fact that proofs of noncontextuality (unlike proofs of Bell nonclassicality) rely on a characterization of the GPT states and effects governing one's experiment---or equivalently, rely on particular operational identities that are known to hold among the GPT processes in one's experiment. If quantum theory is superseded by some future theory, then might it be that one's characterization of one's laboratory procedures as GPT processes are mistaken (and consequently, that the operational identities one believed held among them were also mistaken)? If this did happen, then current assessments of nonclassicality would not necessarily hold up as our theories of physics develop over time.

A competing intuition, however, is that any empirically successful future theory must contain quantum theory in some sense, and so must also exhibit the same forms of nonclassicality (and potentially more).  
In the following, we use the notion of GPT subsystems to show that this intuition can be made precise, and we show that as a consequence, current assessments of nonclassicality {\em will} remain unchanged even if our best theories of nature evolve in the future, if a well-motivated assumption holds.

The key argument was identified already by M\"uller and Garner in Lemma~11 of Ref.~\cite{muller2023testing}.

Any future successor to quantum theory will need to reproduce quantum theory in an appropriate limit, in order to explain why we have not seen beyond-quantum phenomena in any of our current experiments. Just as quantum phenomena are suppressed by decoherence in appropriate regimes, leading to the appearance of classical theory, so too we expect that such beyond-quantum phenomena would be suppressed (in appropriate regimes) by some generalized kind of decoherence process which leads to the appearance of quantum theory. A general notion of this kind, described within the framework of GPTs, was introduced in Refs.~\cite{zyczkowski2008quartic,Selby2017Entanglement,lee2018no,hefford2020hyper}, and is known as hyperdecoherence. 

Formally, a hyperdecoherence process on a GPT $G$ is an idempotent, discard-preserving linear map $H$ that takes states/effects in the GPT to states/effects in the GPT. Recall that idempotence of $H$ simply means that $H\circ H=H$. All of these features of a hyperdecoherence map are of course satisfied by a standard decoherence process taking quantum theory to classical theory.

Given some initial GPT system $G$ (where for simplicity we work with the standard represenation where effects live in the dual vector space), the hyperdecohered GPT system can be defined as the set of all states/effects in the image of the hyperdecoherence map $H$, namely
\beq\label{eq:hyperdecGPT1}
\left\{\btp
		\begin{pgfonlayer}{nodelayer}
			\node [style=Wsquareadj] (1) at (-0, -1) {$s$};
			\node [style=none] (2) at (-0, 1.2) {};
			\node [style=right label] (4) at (0, -0.5) {$G$};
			\node [style=small box] (5) at (0,0.3) {$H$};
			\node [style=right label] (6) at (0, 1.2) {$G$};
		\end{pgfonlayer}
		\begin{pgfonlayer}{edgelayer}
			\draw [qWire] (2.center) to (5);
			\draw [qWire] (5) to (1.center);
		\end{pgfonlayer}
		\etp\right\}_{s\in\Omega_G}
		\textrm{and}\quad
		\left\{\btp
		\begin{pgfonlayer}{nodelayer}
			\node [style=none] (1) at (-0, -1.2) {};
			\node [style=Wsquare] (2) at (-0, 1) {$e$};
			\node [style=right label] (4) at (0, -1.1) {$G$};
			\node [style=small box] (5) at (0,-0.3) {$H$};
			\node [style=right label] (6) at (0, 0.5) {$G$};
		\end{pgfonlayer}
		\begin{pgfonlayer}{edgelayer}
			\draw [qWire] (2.center) to (5);
			\draw [qWire] (5) to (1.center);
		\end{pgfonlayer}
		\etp\right\}_{e\in\mathcal{E}_G}.
\eeq
These hyperdecohered states and effects are always tomographic for each other, and so form a valid GPT system, as was proven, for example, in Appendix~D.2 of Ref.~\cite{centeno2024twirledworldssymmetryinducedfailures}.

Having defined the hyperdecohered theory, we now prove the key result of this section.
 \begin{theorem}\label{thm:DecoherenceIsSubsystem}
Consider a GPT system $G$ and the hyperdecohered GPT system $H(G)$ defined by any given hyperdecoherence process on it. Then $H(G)$ is a GPT subsystem of the original GPT system $G$.
\end{theorem}

The proof is simply that each process in the theory is already itself a process in $G$, by the assumption that $H$ is a physical transformation in the theory\footnote{The fact that a hyperdecoherence process takes states/effects in the GPT system to states/effects in the GPT system is motivated by the idea that the process must itself be a physical transformation in the GPT.}. Consequently, one can simply take the state map to be identity and the effect map to be identity in order to define the embedding of $H(G)$ into $G$.

As mentioned, the canonical example of hyperdecoherence is standard decoherence in quantum theory, e.g., taking a qubit to a classical bit. And as mentioned in Section~\ref{ncofcomp}, the classical bit is indeed a GPT subsystem of the qubit, in keeping with Theorem~\ref{thm:DecoherenceIsSubsystem}.

An immediate corollary of 
Theorem~\ref{thm:DecoherenceIsSubsystem}, first proved in Ref.~\cite{muller2023testing}, is the following:	
\begin{corollary}\label{cor:ContextualityofBeyondQuantum}
		Any theory which hyperdecoheres to quantum theory necessarily fails to admit of a generalised noncontextual ontological model.
	\end{corollary}

This is the sense in which every proof of the failure of noncontextuality that assumes quantum theory will remain valid even if quantum theory itself is superseded by a later successor theory.

If the manner by which quantum theory is recovered from some future successor theory is too exotic to fit the (quite general) mold of a hyperdecoherence process, then one would need to reassess whether such a beyond-quantum theory necessarily retains the same nonclassical features as quantum theory. This is analogous to how if one considers an exotic enough future theory (e.g., one which modified the nature of space and time), one would need to reassess the question of whether such a theory indeed was nonclassical in the sense seemingly implied by Bell's theorem. But in both cases, the accuracy of the conclusions one draws about nonclassicality are quite theory-independent---backed up by general assumptions that plausibly must hold in any reasonable future theory.

\section{Future directions}

It would be interesting to see whether GPT subsystems subsume still other notions of systems that have been considered. Perhaps of special interest are attempts to define a notion of a system (or subsystem) that is operationally defined rather than a primitive notion~\cite{Giuliosubsystems,kramer2018operational,gitton2022loopholegeneralizednoncontextuality}.

In forthcoming work, we generalize the idea of a GPT subsystem to the broader notion of a {\em GPT subtheory}, capturing the idea of what it means for a {\em theory} to live inside another. The definition is quite analogous to that given herein for GPT subsystems: one GPT is a GPT subtheory of another if there exists a structure-preserving embedding from the former into the latter. The novelty is that once one considers GPTs as full theories, there is  compositional structure in addition to convex structure to be preserved.

Another important direction for future work is to consider how the shadow map {\em quantitatively} affects simplex-embeddability. In this work, we focused merely on the qualitative divide between simplex-embeddable or not.
We also suspect that combining the current tools with those of Ref.~\cite{PuseydelRio} could lead to further progress on understanding when and how failures of relative tomography leads to mistaken assessments of nonclassicality. Note that our approach in this work was to investigate what happens to a true GPT description in the presence of limited experimental tests; Ref.~\cite{PuseydelRio}, in contrast, can be viewed as characterizing what can be said about the scope of all possible true GPT descriptions of an experiment, given a constraint on how far off one's set of processes is from being relatively tomographic.

Although we have given a simple sufficient condition (relative tomography of one's preparations and measurements) for assessments of noncontextuality to be correct, this condition is not something one can experimentally guarantee (although one can certainly gather experimental evidence for it~\cite{mazurek2021experimentally,grabowecky2021experimentally}). 
It remains an open question whether one can find other sufficient conditions, and which of these conditions might be easiest to test experimentally. Or, might there even be physical conditions under which relative tomography is guaranteed?
Very speculatively, one might obtain such a guarantee if one had a means to generate (all possible) physical symmetry transformations in some GPT.

Finally, we note that our conclusions regarding the relevant assumptions for witnessing the failure of generalized noncontextuality are also relevant to operational proofs of the failure of Kochen-Specker noncontextuality~\cite{KS}, since the Kochen-Specker assumptions are best motivated by the same considerations, namely appeal to Leibniz's principle~\cite{Leibniz}.

\section*{Acknowledgements}	
DS thanks Rob Spekkens and Elie Wolfe for useful discussions regarding when proofs of nonclassicality may or may not be robust. Early inspiration for this work came from discussions with Rob Spekkens relating to proofs of nonclassicality on a rebit embedded in a qubit. We thank Markus P.~M\"uller, Thomas Galley, Caroline L.~Jones, and Howard Barnum for a delightful and insightful visit to IQOQI-Vienna where many ideas related to this manuscript were discussed. We thank Markus P.~M\"uller especially for helping us recognize that relative tomography of states and effects cannot be captured by simply looking at the spans of their vector spaces, and in particular for coming up with the example in Section~\ref{twodisk}. We thank also Y{\`i}l{\`e} Y{\=\i}ng, Victor Gitton, and Marco Erba for relevant discussions. Finally, we thank the referees for carefully reading this work and providing helpful feedback.
JHS and DS were supported by the National Science Centre, Poland (Opus project, Categorical Foundations of the Non-Classicality of Nature, project no.~2021/41/B/ST2/03149). 
VR and RB acknowledge support by the Digital Horizon Europe project FoQaCiA, Foundations of quantum computational advantage, GA No.~101070558, funded by the European Union, NSERC (Canada), and UKRI (UK).
This work is partially carried out under IRA Programme, project no.~FENG.02.01-IP.05-0006/23, financed by the FENG program 2021-2027, Priority FENG.02, Measure FENG.02.01., with the support of the FNP. Research at Perimeter Institute is supported in part by the Government of Canada through the Department of Innovation, Science and Economic Development and by the Province of Ontario through the Ministry of Colleges and Universities.
Some diagrams were prepared using TikZit and Mathcha.

    \

 \bibliographystyle{quantum}
	\bibliography{bib.bib}
	
	\appendix

\section{Proof of Lemma~\ref{lem:AutomorphismsAreIso}}\label{sec:lemmaproof}

\AutomorphismsAreIso*  

 \proof  To begin, recall that every GPT embedding is faithful (Def.~\ref{defn:GPTembedding}), which means that $\iota$ and $\kappa$ are injective state and effect maps. Moreover, they are injective as linear maps $\iota:\mathsf{Span}[\Omega]\to\mathsf{Span}[\Omega]$ and $\kappa:\mathsf{Span}[\mathcal{E}]\to\mathsf{Span}[\mathcal{E}]$, so they automatically have a linear left inverse. All that remains to be shown is that they are also surjective (and hence invertible) as state and effect maps. In other words, $\iota(\Omega)=\Omega$ and $\kappa(\mathcal{E})=\mathcal{E}$.

First, we show that $\iota$ and $\kappa$ cannot change the volume of $\Omega$ and $\mathcal{E}$. Fix a Lebesgue measure on $\mathsf{Span}[\Omega]$ and write $\mathsf{vol}()$ for the corresponding volume (and analogously for $\mathsf{Span}[\mathcal{E}]$). Note that 
\begin{align}
    \mathsf{vol}[\iota(\Omega)]&=\det[\iota]\mathsf{vol}[\Omega],\\
    \mathsf{vol}[\kappa(\mathcal{E})]&=\det[\kappa]\mathsf{vol}[\mathcal{E}],
\end{align}  
since the determinant of a linear map specifies the scaling factor of volume for any measurable full-dimensional convex set(see~\cite[Lemma 40.4, p. 389]{Lesbegue}).\footnote{Here, our state and effect spaces are convex, and therefore measurable with respect to Lebesgue measure \cite[Thm.~2]{lang1986note}.} The conditions $\iota(\Omega)\subseteq \Omega$ and $\kappa(\mathcal{E})\subseteq\mathcal{E}$ imply that the volumes can only decrease, so $\det[\iota]\leq 1$ and $\det[\kappa]\leq 1$. 
Now,  recall that $p$ is tomographic, which means that it constitutes a nondegenerate bilinear pairing, $p:\mathsf{Span}[\Omega]\times\mathsf{Span}[\mathcal{E}]\to\mathds{R}$. That is,  
\begin{align}
    p(s,e)=0 \quad  \forall e \quad  \iff \, s=0,
\end{align}  
and  
\begin{align}
    p(s,e)=0\quad \forall s \quad \iff \, e=0.
\end{align}  
We can therefore define a transpose relative to $p$, $*_p$, such that for any linear map $f:\mathsf{Span}[\Omega]\to\mathsf{Span}[\Omega]$, there exists another linear map $f^{*_p}:\mathsf{Span}[\mathcal{E}]\to\mathsf{Span}[\mathcal{E}]$ such that  
\begin{align}
    p(f(s),e)=p(s,f^{*_p}(e)) \quad \forall \, s,e.
\end{align}  
Tomography of $p$ and probability preservation of embeddings lead to
\begin{align}
p(s,e)=p(\iota(s),\kappa(e))=p(s,\iota^{*_p}(\kappa(e))) \quad \forall \, s,e,
\end{align}  
which implies that  
\begin{align}
 \iota^{*_p}(\kappa(e))=e\,, \quad  \forall \, e\,.   
\end{align}  
From this it follows that  
\begin{equation} \label{invAdj}
\iota^{*_p}\circ\kappa=\mathsf{Id}_{\mathsf{Span}(\mathcal{E})}, \quad \text{and} \quad \kappa=(\iota^{*_p})^{-1}.
\end{equation}  
Therefore,  $\det[\kappa]=\det[(\iota^{*_p})^{-1}]=\det[\iota^{*_p}]^{-1}$, which together with the fact that $\det[\iota^{*_p}]=\det[\iota]$ (see Ref.~\cite[Theorem 3.21, p. 21]{LinAlg}) implies that $\det[\kappa]=\det[\iota]^{-1}$. Together with $\det[\iota]\leq 1$ and $\det[\kappa]\leq 1$, we get $\det[\iota]=\det[\kappa]=1$. So indeed $\mathsf{vol}[\iota(\Omega)]=\mathsf{vol}[\Omega]$ and $\mathsf{vol}[\kappa(\mathcal{E})]=\mathsf{vol}[\mathcal{E}]$.

The second step is to note that $\iota(\Omega)\supseteq\mathsf{int}(\Omega)$, where the interior is taken within $\mathsf{Span}[\Omega]$ (i.e., the relative interior). Indeed, if any relative interior point $x\in\Omega$ were not in $\iota(\Omega)$, there would exist a sufficiently small full-dimensional ball around $x$ contained in $\Omega$ but disjoint from $\iota(\Omega)$. This would imply a strictly smaller volume for $\iota(\Omega)$, contradicting the equality above. Thus every interior point of $\Omega$ lies in $\iota(\Omega)$. Moreover, since both $\iota(\Omega)$ and $\Omega$ are closed we have that $\iota(\Omega)=\overline{\iota(\Omega)}\supseteq \overline{\mathsf{int}(\Omega)}=\Omega$, where $\overline{\iota(\Omega)}$ and $\overline{\mathsf{int}(\Omega)}$ denote the closures of these sets.
Finally, recalling that  $\iota(\Omega)\subseteq\Omega$ (since $\iota$ is a state map from $\Omega$ to itself), we conclude $\iota(\Omega)=\Omega$. A similar argument shows that $\kappa(\mathcal{E})=\mathcal{E}$.

 Therefore, $\iota$ and $\kappa$ are surjective state and effect maps, and the pair $(\iota,\kappa)$ constitutes a GPT isomorphism.  

\section{Proof of Theorem \ref{shadow-equiv}} \label{secshadowequiv}

First, we state and prove three useful lemmas.

\begin{lemma}\label{lem:LinearityFromEmpAd}
Shadow maps are necessarily linear.
\end{lemma}
 \begin{proof}
    Consider first a state $s\in\Omega_f$ of a GPT fragment $f$, such that
    \beq
        \btp
        \tikzstate{s}{f}
    \etp
    =
    \sum_i\alpha_i\ \btp\begin{pgfonlayer}{nodelayer}
		\node [style=Wsquareadj] (1) at (-0, -0.4) {$s_i$};
		\node [style=none] (2) at (-0, 0.5) {};
		\node [style=right label] (3) at (0, 0.5) {$f$};
	\end{pgfonlayer}
	\begin{pgfonlayer}{edgelayer}
		\draw [qWire] (2.center) to (1.center);
\end{pgfonlayer}\etp,
    \eeq
     where $\{\alpha_i\}$ are real numbers  and $s_i$ are arbitrary states in $\Omega_f$. Then for any effect $e\in\mathcal{E}_f$, 
    \beq
        \btp
		\begin{pgfonlayer}{nodelayer}
			\node [style=Wsquareadj] (1) at (-0, -1.4) {$s$};
			\node [style=Wsquare] (2) at (-0, 1.4) {$e$};
			\node [style=right label] (4) at (0, 0.9) {$f$};
                \node [style=empty circle, fill=white] (5) at (0,0) {$p_f$};
			\node [style=right label] (6) at (0, -0.7) {$f$};
		\end{pgfonlayer}
		\begin{pgfonlayer}{edgelayer}
			\draw [qWire] (2) to (1.center);
		\end{pgfonlayer}
		\etp
  \ =\ 
  \btp
			\begin{pgfonlayer}{nodelayer}
				\node [style=Wsquareadj] (1) at (-0, -2.6) {$s$};
				\node [style=Wsquare] (2) at (-0, 2.6) {$e$};
				\node [style=right label] (4) at (0,2.2) {$f$};
				\node [style=small box,inner sep=1pt,fill=gray!30] (5) at (0, 1.4) {$\tau$};
				\node [style=small box,inner sep=1pt,fill=gray!30] (7) at (0, -1.4) {$\sigma$};
				\node [style=empty circle,fill=gray!30] (8) at (0,0) {$p_\mathcal{S}$};
				\node [style=right label] (6) at (0, -2) {$f$};
                \node [style=right label] (10) at (0,0.8) {$\mathcal{S}(f)$};
                \node [style=right label] (9) at (0,-0.8) {$\mathcal{S}(f)$};
			\end{pgfonlayer}
			\begin{pgfonlayer}{edgelayer}
				\draw [qWire] (2) to (5);
				\draw [sWire] (5) to (7);
				\draw [qWire] (7) to (1);
			\end{pgfonlayer}
			\etp.
    \eeq

Now, we appeal to the bilinearity of the map $p_f$. It will follow that
\begin{eqnarray}
    \btp
			\begin{pgfonlayer}{nodelayer}
				\node [style=Wsquareadj] (1) at (-0, -2.6) {$s$};
				\node [style=Wsquare] (2) at (-0, 2.6) {$e$};
				\node [style=right label] (4) at (0,2.2) {$f$};
				\node [style=small box,inner sep=1pt,fill=gray!30] (5) at (0, 1.4) {$\tau$};
				\node [style=small box,inner sep=1pt,fill=gray!30] (7) at (0, -1.4) {$\sigma$};
				\node [style=empty circle,fill=gray!30] (8) at (0,0) {$p_\mathcal{S}$};
				\node [style=right label] (6) at (0, -2) {$f$};
                \node [style=right label] (10) at (0,0.8) {$\mathcal{S}(f)$};
                \node [style=right label] (9) at (0,-0.8) {$\mathcal{S}(f)$};
			\end{pgfonlayer}
			\begin{pgfonlayer}{edgelayer}
				\draw [qWire] (2) to (5);
				\draw [sWire] (5) to (7);
				\draw [qWire] (7) to (1);
			\end{pgfonlayer}
			\etp
   &=&\quad\quad
    \btp
		\begin{pgfonlayer}{nodelayer}
			\node [style=Wsquareadj] (1) at (-0, -1.5) {$s$};
			\node [style=Wsquare] (2) at (-0, 1.5) {$e$};
			\node [style=right label] (4) at (0, 0.9) {$f$};
                \node [style=empty circle, fill=white] (5) at (0,0) {$p_f$};
			\node [style=right label] (6) at (0, -0.8) {$f$};
		\end{pgfonlayer}
		\begin{pgfonlayer}{edgelayer}
			\draw [qWire] (2) to (1.center);
		\end{pgfonlayer}
		\etp\\
  &=&
  \btp
		\begin{pgfonlayer}{nodelayer}
			\node [style=Wsquareadj] (1) at (-0, -1.7) {$s_i$};
			\node [style=Wsquare] (2) at (-0, 1.7) {$e$};
			\node [style=right label] (4) at (0, 1.2) {$f$};
                \node [style=empty circle,fill=white] (5) at (0,0.3) {$p_f$};
			\node [style=right label] (6) at (0,-0.4) {$f$};
            \node [style=none] (7) at (-1.8,-1.6) {$\sum_i\alpha_i$};
            \node [style=none] (8) at (-3,-.75) {};
            \node [style=none] (9) at (1,-.75) {};
            \node [style=none] (10) at (1,-2.75){};
            \node [style=none] (11) at (-3,-2.75) {};
		\end{pgfonlayer}
		\begin{pgfonlayer}{edgelayer}
			\draw [qWire] (2) to (1.center);
                \draw [thick gray dashed edge] (8.center) to (9.center) to (10.center) to (11.center) to cycle;
		\end{pgfonlayer}
		\etp\\
  &=& \sum_i\alpha_i\ \btp
		\begin{pgfonlayer}{nodelayer}
			\node [style=Wsquareadj] (1) at (-0, -1.6) {$s_i$};
			\node [style=Wsquare] (2) at (-0, 1.6) {$e$};
			\node [style=right label] (4) at (0, 0.9) {$f$};
                \node [style=empty circle,fill=white] (5) at (0,0) {$p_f$};
			\node [style=right label] (6) at (0, -0.8) {$f$};
		\end{pgfonlayer}
		\begin{pgfonlayer}{edgelayer}
			\draw [qWire] (2) to (1.center);
		\end{pgfonlayer}
		\etp\\
  &=&\sum_i\alpha_i\ \btp
			\begin{pgfonlayer}{nodelayer}
				\node [style=Wsquareadj] (1) at (-0, -2.8) {$s_i$};
				\node [style=Wsquare] (2) at (-0, 2.8) {$e$};
				\node [style=right label] (4) at (0,2.2) {$f$};
				\node [style=small box,inner sep=1pt,fill=gray!30] (5) at (0, 1.4) {$\tau$};
				\node [style=small box,inner sep=1pt,fill=gray!30] (7) at (0, -1.4) {$\sigma$};
				\node [style=empty circle,fill=gray!30] (8) at (0,0) {$p_\mathcal{S}$};
				\node [style=right label] (6) at (0, -2) {$f$};
                \node [style=right label] (10) at (0,0.8) {$\mathcal{S}(f)$};
                \node [style=right label] (9) at (0,-0.8) {$\mathcal{S}(f)$};
			\end{pgfonlayer}
			\begin{pgfonlayer}{edgelayer}
				\draw [qWire] (2) to (5);
				\draw [sWire] (5) to (7);
				\draw [qWire] (7) to (1);
			\end{pgfonlayer}
			\etp,
\end{eqnarray}
and since this holds for any $\tau(e)\in\mathcal{S}(\mathcal{E}_f)$, it follows from the fact that $p_\mathcal{S}$ is tomographic  that
\beq
   \begin{tikzpicture}
	\begin{pgfonlayer}{nodelayer}
		\node [style=none] (0) at (0, 1.7) {};
		\node [style=small box, inner sep=1pt, fill={gray!30}] (1) at (0, 0.8) {$\sigma$};
		\node [style=Wsquareadj] (2) at (0, -1.2) {$s_i$};
		\node [style=right label] (3) at (0, 1.7) {$\mathcal{S}(f)$};
		\node [style=right label] (4) at (0, -0.1) {$f$};
		\node [style=none] (5) at (-1.7, -1.3) {$\sum_i\alpha_i$};
		\node [style=none] (10) at (-2.75, -0.5) {};
		\node [style=none] (11) at (1, -0.5) {};
		\node [style=none] (12) at (1, -2.5) {};
		\node [style=none] (13) at (-2.75, -2.5) {};
	\end{pgfonlayer}
	\begin{pgfonlayer}{edgelayer}
		\draw [sWire] (0.center) to (1);
		\draw [qWire] (1) to (2);
		\draw [thick gray dashed edge] (11.center)
			 to (10.center)
			 to (13.center)
			 to (12.center)
			 to cycle;
	\end{pgfonlayer}
\end{tikzpicture}
\ =\ 
\sum_i\alpha_i\ 
\btp
    \begin{pgfonlayer}{nodelayer}
        \node [style=none] (0) at (0,1.2) {};
        \node [style=small box,inner sep=1pt,fill=gray!30] (1) at (0,0.4) {$\sigma$};
        \node [style=Wsquareadj] (2) at (0,-1.2) {$s_i$};
        \node [style=right label] (3) at (0,1.3) {$\mathcal{S}(f)$};
        \node [style=right label] (4) at (0,-0.4) {$f$};
    \end{pgfonlayer}
    \begin{pgfonlayer}{edgelayer}
        \draw [sWire] (0.center) to (1);
        \draw [qWire] (1) to (2);
    \end{pgfonlayer}
\etp .
\eeq
 This must be true for any $s\in \Omega_f$ and any decomposition of it in terms of some $s_i$, so $\sigma$ is a linear map. 
		
  The proof of the linearity of $\tau$ follows a similar reasoning, relying now on the fact that $\mathcal{S}(\Omega_f)$ spans the dual of $\mathcal{S}(\mathcal{E}_f)$.
\end{proof}

\begin{lemma}
\label{lemma:kerShadowmaps} Consider a fragment $f=(\Omega_f,\mathcal{E}_f,p_f)$. For any shadow $\mathcal{S}(f)=(\sigma(\Omega_f),\tau(\mathcal{E}_f),p_\mathcal{S})$ of the fragment, the maps $\sigma$ and $\tau$ obey
    \begin{align}
        \mathsf{Ker}[\sigma] = \mathsf{Ker}_L[p_f]\\
        \mathsf{Ker}[\tau] = \mathsf{Ker}_R[p_f].
    \end{align}
\end{lemma}
\begin{proof}
    The proof is as follows, where the first equivalence is from the definition of the kernel, the second is from the fact that the effects of $\mathcal{S}(f)$ are (by definition) all of the form $\tau(e)$ and these (by definition) span the right argument of $p_\mathcal{S}$, the third is from the fact that the shadow map preserves the probabilistic predictions of the fragment, the fourth is from the fact that the effects $e$ (by definition) span the left argument of $p_f$, and the final equivalence is the definition of the left kernel. 
    	\begin{eqnarray}
	v \in\mathsf{Ker}(\sigma)&\iff& 
        \btp
	\begin{pgfonlayer}{nodelayer}
		\node [style=Wsquareadj] (0) at (0,-1) {$v$};
        \node [style=none] (1) at (0,1) {};
        \node [style=small box,inner sep=1pt,fill=gray!30] (2) at (0,0.2) {$\sigma$};
        \node [style=right label] (3) at (0,-0.4) {$f$};
        \node [style=right label] (4) at (0,0.8) {$\mathcal{S}(f)$};
        \end{pgfonlayer}
        \begin{pgfonlayer}{edgelayer}
            \draw [qWire] (0) to (2);
            \draw [sWire] (2) to (1);
        \end{pgfonlayer}
	\etp = 0\\
	&\iff&  \btp
			\begin{pgfonlayer}{nodelayer}
				\node [style=Wsquareadj] (1) at (-0, -2.6) {$v$};
				\node [style=Wsquare] (2) at (-0, 2.6) {$e$};
				\node [style=right label] (4) at (0,2.2) {$f$};
				\node [style=small box,inner sep=1pt,fill=gray!30] (5) at (0, 1.4) {$\tau$};
				\node [style=small box,inner sep=1pt,fill=gray!30] (7) at (0, -1.4) {$\sigma$};
				\node [style=empty circle,fill=gray!30] (8) at (0,0) {$p_\mathcal{S}$};
				\node [style=right label] (6) at (0, -2) {$f$};
                \node [style=right label] (10) at (0,0.8) {$\mathcal{S}(f)$};
                \node [style=right label] (9) at (0,-0.8) {$\mathcal{S}(f)$};
			\end{pgfonlayer}
			\begin{pgfonlayer}{edgelayer}
				\draw [qWire] (2) to (5);
				\draw [sWire] (5) to (7);
				\draw [qWire] (7) to (1);
			\end{pgfonlayer}
			\etp =0,\quad\forall e\in\mathcal{E}_f\\
	&\iff& \btp
	\begin{pgfonlayer}{nodelayer}
		\node [style=Wsquareadj] (1) at (-0, -1.4) {$v$};
		\node [style=Wsquare] (2) at (-0, 1.4) {$e$};
		\node [style=right label] (4) at (0, 1) {$f$};
		\node [style=empty circle,fill=white] (8) at (0,0) {$p_f$};
		\node [style=right label] (6) at (0, -0.8) {$f$};
	\end{pgfonlayer}
	\begin{pgfonlayer}{edgelayer}
		\draw [qWire] (2) to (1.center);
	\end{pgfonlayer}
	\etp\ =0,\quad\forall e\in\mathcal{E}_f\\
	&\iff& \btp
	\begin{pgfonlayer}{nodelayer}
		  \node [style=Wsquareadj] (0) at (0,-1.2) {$v$};
            \node [style=none] (1) at (0,1.2) {};
            \node [style=right label] (2) at (0,1) {$f$};
            \node [style=empty circle,inner sep=1pt,fill=white] (3) at (0,0.1) {$p_f$};
            \node [style=right label] (4) at (0,-0.55) {$f$};
	\end{pgfonlayer}
 \begin{pgfonlayer}{edgelayer}
        \draw [qWire] (0) to (1);
 \end{pgfonlayer}
	\etp\ =0\\
        &\iff& v\in \mathsf{Ker}_L[p_f].
    \end{eqnarray}
    
        A similar argument shows that $\mathsf{Ker}[\tau]=\mathsf{Ker}_R[p_f]$. 
\end{proof}

 In particular, if we consider some $v\in\mathsf{Ker}[\sigma]$ such that $v=s-s'$ then this means that $\sigma(s)=\sigma(s')\iff s\backsim s'$ and similarly,  for $w \in \mathsf{Ker}[\tau]$ such that $w=e-e'$, one has  $\tau(e)=\tau(e')\iff$ $e\backsim e$ where the relations $\backsim$ are the ones in Def.~\ref{defn:method4}.

\begin{lemma}[Shadow maps split]
\label{lemma:ShadowMapsSplit} Consider a fragment $f=(\Omega_f,\mathcal{E}_f,p_f)$. Given any shadow $\mathcal{S}(f)=(\sigma(\Omega_f),\tau(\mathcal{E}_{f}),p_\mathcal{S})$ for $f$, the maps $\sigma$ and $\tau$ split as
    \begin{align}
    \label{shadpwmapssplit}
        \sigma = \tilde{\sigma}\circ\backsim,\\
        \tau = \tilde{\tau}\circ\backsim,
    \end{align}
    where $\backsim$ are the quotienting maps defining $f\ns/\ns\backsim$ (as in Def.~\ref{defn:method4}) and where  $\tilde{\sigma}:\Omega\ns/\ns\backsim \ \to U$ and $\tilde{\tau}: \mathcal{E}\ns/\ns\backsim \ \to V$  are unique \emph{injective} linear maps.  Diagrammatically,
    \beq
        \btp
            \begin{pgfonlayer}{nodelayer}
                \node[style=small box,inner sep=1pt,fill=gray!30] (0) at (0,0) {$\sigma$};
                \node [style=none] (1) at (0,0.8) {};
                \node [style=none] (2) at (0,-0.8) {};
                \node [style=right label] (3) at (0,0.7) {$\mathcal{S}(f)$};
                \node[style=right label] (4) at (0,-0.7) {$f$};
            \end{pgfonlayer}
            \begin{pgfonlayer}{edgelayer}
                \draw [qWire] (2) to (0);
                \draw [sWire] (0) to (1);
            \end{pgfonlayer}
        \etp
        \ =\ 
        \btp
            \begin{pgfonlayer}{nodelayer}
                \node [style=right label] (0) at (0,-0.1) {$f\ns/\ns\backsim$};
                \node [style=coinc,inner sep=1pt, fill=gray!30] (1) at (0,-0.6) {$\backsim$};
                \node [style=small box,inner sep=1pt, fill=gray!30] (2) at (0,0.5) {$\tilde{\sigma}$};
                \node [style=right label] (3) at (0,-1.2) {$f$};
                \node [style=right label] (4) at (0,1.2) {$\mathcal{S}(f)$};
                \node [style=none] (5) at (0,-1.3) {};
                \node [style=none] (6) at (0, 1.3) {};
            \end{pgfonlayer}
            \begin{pgfonlayer}{edgelayer}
                \draw [qWire] (5) to (1);
                \draw [sWire] (1) to (6);
            \end{pgfonlayer}
        \etp,
    \eeq
        \beq
        \btp
            \begin{pgfonlayer}{nodelayer}
                \node[style=small box,inner sep=1pt,fill=gray!30] (0) at (0,0) {$\tau$};
                \node [style=none] (1) at (0,0.8) {};
                \node [style=none] (2) at (0,-0.8) {};
                \node [style=right label] (3) at (0,-0.7) {$\mathcal{S}(f)$};
                \node[style=right label] (4) at (0,0.7) {$f$};
            \end{pgfonlayer}
            \begin{pgfonlayer}{edgelayer}
                \draw [qWire] (1) to (0);
                \draw [sWire] (0) to (2);
            \end{pgfonlayer}
        \etp
        \ =\ 
        \btp
            \begin{pgfonlayer}{nodelayer}
                \node [style=right label] (0) at (0,0.1) {$f\ns/\ns\backsim$};
                \node [style=inc,inner sep=1pt, fill=gray!30] (1) at (0,0.5) {$\backsim$};
                \node [style=small box,inner sep=1pt, fill=gray!30] (2) at (0,-0.6) {$\tilde{\tau}$};
                \node [style=right label] (3) at (0,1.2) {$f$};
                \node [style=right label] (4) at (0,-1.2) {$\mathcal{S}(f)$};
                \node [style=none] (5) at (0,-1.3) {};
                \node [style=none] (6) at (0, 1.3) {};
            \end{pgfonlayer}
            \begin{pgfonlayer}{edgelayer}
                \draw [qWire] (6) to (1);
                \draw [sWire] (1) to (5);
            \end{pgfonlayer}
        \etp.
    \eeq
\end{lemma}
\begin{proof}
      This is a consequence of the universal property of quotients. Consider a linear map $\mu: T\to W$ where $T$ and $W$ are vector spaces, and consider a subspace $A\subseteq T$ such that $A\subseteq \mathsf{Ker}[\mu]$. The universal property of the quotient (for vector spaces) ensures that there exists a unique $\tilde{\mu}:T/A\to W$ such that $\mu=\tilde{\mu}\circ \sim_A$, where $\sim_A$ is the quotienting map $\sim: T\to T/A$. Note that $\mathsf{Ker}[\tilde{\mu}] = \mathsf{Ker}[\mu]/A$. Therefore, if $A=\mathsf{Ker}[\mu]$, then $\mathsf{Ker}[\tilde{\mu}]=0_{T/A}$, so $\tilde{\mu}$ is injective. In other words, the universal property tells us that there exists a unique map $\tilde{\mu}$ such that the following diagram commutes:

      \beq
        \btp
            \begin{pgfonlayer}{nodelayer}
                \node[style=none] (0) at (-1.5,1.5) {$T$};
                \node [style=none] (1) at (1.5,1.5) {$W$};
                \node[style=none] (2) at (1.5,-1.5) {$T\ns/\ns A$};
                \node [style=none] (3) at (0,1.9) {$\mu$};
                \node [style=none] (4) at (1.9,0) {$\tilde{\mu}$};
                \node [style=none] (5) at (-0.5,-0.5) {$\sim_A$};
            \end{pgfonlayer}
            \begin{pgfonlayer}{edgelayer}
                \draw [->] (0) to (1);
                \draw [->] (2) to (1);
                \draw [->] (0) to (2);
            \end{pgfonlayer}
        \etp
      \eeq

      \sloppy Let us apply this theorem to our case, with the substitution $T=\mathsf{Span}[\Omega_f]$, $W=U$, and $A=\mathsf{Ker}[\sigma]$ for splitting $\sigma$. Since the quotienting map is the canonical projection onto $\mathsf{Span}[\Omega_{f}]/\mathsf{Ker}[\sigma]$, in our case this equals $\mathsf{Span}[\Omega_{f}]/\mathsf{Ker}_L[p_f]$, given Lemma~\ref{lemma:kerShadowmaps}. This is exactly the quotienting relation for states in Def.~\ref{defn:method4} (and therefore denoted by the backwards symbol $\backsim$). Next, we apply the universal property  with the substitution $T=\mathsf{Span}[\mathcal{E}_f]$, $W=V$, and $A=\mathsf{Ker}[\tau]$ in order to split $\tau$. Here, we get that $\sim$ is the quotienting of $\mathcal{E}_f$ with respect to  $\mathsf{Ker}_R[p_f]$  (again using Lemma~\ref{lemma:kerShadowmaps}), therefore also denoted by the backwards symbol `$\backsim$'. In other words, the quotienting identifies any pair of states $s,s'$ such that $p_f(s,e)=p_f(s',e)$ for all effects $e\in \mathcal{E}_f$ and similarly for effects. 
      
     The universal property of quotients then implies that there exists a unique map ${\tilde{\sigma}:\mathsf{Span}[\Omega_{f}/\backsim]\to U}$ such that $\sigma = \tilde{\sigma}\circ \backsim$ and, similarly, there exists a unique map ${\mathsf{Span}[\mathcal{E}_{ f}/\backsim]\to V}$ such that $\tau = \tilde{\tau}\circ \backsim$. Since the quotienting maps are implemented with respect to the kernels of $\sigma$ and $\tau$, both maps $\tilde{\sigma}$ and $\tilde{\tau}$ are injective.  
      \end{proof}

      Lemmas~\ref{lem:LinearityFromEmpAd} and ~\ref{lemma:ShadowMapsSplit} can be used to prove that any shadow with the form prescribed by Definition~\ref{defn:shadow} is equivalent to the  particular (canonical) shadow $f/\backsim$ constructed in Definition~\ref{defn:method4}.
      
      \begin{lemma}[All shadows $\mathcal{S}(f)$ are equivalent to $f/\sim$]
      \label{lemma:AllShadowsEquivQuotientShadow}
           \sloppy Consider a fragment $f=(\Omega_f,\mathcal{E}_f,p_f)$ and an arbitrary shadow $\mathcal{S}(f)=(\sigma(\Omega_f),\tau(\mathcal{E}_f),p_\mathcal{S})$ of that fragment (as defined in Definition~\ref{defn:shadow}).  Then, the quotient shadow $f\ns/\ns\backsim$ (as defined in Definition~\ref{defn:method4}) is equivalent to $\mathcal{S}(f)$ . That is,  there exists  an \emph{invertible} GPT fragment embedding of $f\ns/\ns\backsim$ into $\mathcal{S}(f)$.
      \end{lemma}
      \begin{proof}
     \sloppy  We will show that $(\tilde{\sigma},\tilde{\tau})$ provides a faithful and full embedding of $f\ns/\ns\backsim$ into $\mathcal{S}(f)$. Recall that this implies that the embedding is invertible. 
     First, note that $\tilde{\sigma}$ inherits linearity from $\sigma$ (which must be linear, by Lemma~\ref{lem:LinearityFromEmpAd}). Second, given Lemma~\ref{lemma:ShadowMapsSplit}, the domain of $\tilde{\sigma}$ is $\mathsf{Span}[\Omega_{ f}/\backsim]$, and it is an injective state map, obeying $\tilde{\sigma}(\Omega_f/\backsim)\subseteq \sigma(\Omega_{f})$ and mapping any distinct pair of states $\btilde{s},\btilde{s}'\in\Omega_{f}/\backsim$ to distinct states in $\sigma(\Omega_{f})$. We now show  $\tilde{\sigma}$ is surjective (full). Any state in $\sigma(\Omega_f)$ can be written as $\sigma(s)$ with $s\in\Omega_f$, and using the splitting of $\sigma$, we get
    \beq
\btp
	\begin{pgfonlayer}{nodelayer}
		\node [style=Wsquareadj] (0) at (0,-1) {$s$};
        \node [style=none] (1) at (0,1) {};
        \node [style=small box,inner sep=1pt,fill=gray!30] (2) at (0,0.2) {$\sigma$};
        \node [style=right label] (3) at (0,-0.3) {$f$};
        \node [style=right label] (4) at (0,0.8) {$\mathcal{S}(f)$};
        \end{pgfonlayer}
        \begin{pgfonlayer}{edgelayer}
            \draw [qWire] (0) to (2);
            \draw [sWire] (2) to (1);
        \end{pgfonlayer}
	\etp
 \ =\ 
        \btp
        \begin{pgfonlayer}{nodelayer}
            \node [style=right label] (0) at (0,0) {$f\ns/\ns\backsim$};
            \node [style=Wsquareadj] (1) at (0,-1.6) {$s$};
            \node [style=none] (2) at (0,1.6) {};
            \node [style=small box,inner sep=1pt,fill=gray!30] (3) at (0,0.6) {$\tilde{\sigma}$};
            \node [style=coinc,inner sep=1pt,fill=gray!30] (4) at (0,-0.6) {$\backsim$};
            \node[style=right label] (5) at (0,-1) {$f$};
            \node [style=right label] (6) at (0,1.3) {$\mathcal{S}(f)$};
        \end{pgfonlayer}
        \begin{pgfonlayer}{edgelayer}
            \draw [qWire] (1) to (4);
            \draw [sWire] (4) to (2);
        \end{pgfonlayer}
        \etp
        \ =\ 
                \btp
	\begin{pgfonlayer}{nodelayer}
		\node [style=Wsquareadj] (0) at (0,-1) {$\tilde{s}$};
        \node [style=none] (1) at (0,1) {};
        \node [style=small box,inner sep=1pt,fill=gray!30] (2) at (0,0.2) {$\tilde{\sigma}$};
        \node [style=right label] (3) at (0,-0.3) {$f\ns/\ns\backsim$};
        \node [style=right label] (4) at (0,0.8) {$\mathcal{S}(f)$};
        \end{pgfonlayer}
        \begin{pgfonlayer}{edgelayer}
            \draw [sWire] (0) to (1);
        \end{pgfonlayer}
	\etp
        \forall \sigma(s)\in\sigma(\Omega_f),
    \eeq
    \sloppy which implies that ${\sigma(\Omega_f)\subseteq\tilde{\sigma}(\Omega_f/\backsim)}$. Therefore, ${\tilde{\sigma}(\Omega_f/\backsim)=\sigma(\Omega_f)}$, and $\tilde{\sigma}$ is a full state map.
  
    Analogously for effects, $\tilde{\tau}$ is linear (again as a consequence of Lemma~\ref{lem:LinearityFromEmpAd}), its domain is $\mathcal{E}_f/\backsim$, and it acts as an injective effect map, as a consequence of Lemma~\ref{lemma:ShadowMapsSplit}. Finally, it is also a full effect map, since for any $\tau(e)\in\tau(\mathcal{E}_f)$ one has 
    \beq
    \btp
	\begin{pgfonlayer}{nodelayer}
		\node [style=Wsquare] (0) at (0,1) {$e$};
        \node [style=none] (1) at (0,-1) {};
        \node [style=small box,inner sep=1pt,fill=gray!30] (2) at (0,-0.2) {$\tau$};
        \node [style=right label] (3) at (0,0.4) {$f$};
        \node [style=right label] (4) at (0,-0.8) {$\mathcal{S}(f)$};
        \end{pgfonlayer}
        \begin{pgfonlayer}{edgelayer}
            \draw [qWire] (0) to (2);
            \draw [sWire] (2) to (1);
        \end{pgfonlayer}
	\etp
 \ =\ 
 \btp
        \begin{pgfonlayer}{nodelayer}
            \node [style=right label] (0) at (0,0.1) {$f\ns/\ns\backsim$};
            \node [style=Wsquare] (1) at (0,1.6) {$e$};
            \node [style=none] (2) at (0,-1.6) {};
            \node [style=small box,inner sep=1pt,fill=gray!30] (3) at (0,-0.6) {$\tilde{\tau}$};
            \node [style=inc,inner sep=1pt,fill=gray!30] (4) at (0,0.6) {$\backsim$};
            \node[style=right label] (5) at (0,1.1) {$f$};
            \node [style=right label] (6) at (0,-1.3) {$\mathcal{S}(f)$};
        \end{pgfonlayer}
        \begin{pgfonlayer}{edgelayer}
            \draw [qWire] (1) to (4);
            \draw [sWire] (4) to (2);
        \end{pgfonlayer}
        \etp
        \ =\ 
                \btp
	\begin{pgfonlayer}{nodelayer}
		\node [style=Wsquare] (0) at (0,1) {$\tilde{e}$};
        \node [style=none] (1) at (0,-1) {};
        \node [style=small box,inner sep=1pt,fill=gray!30] (2) at (0,-0.2) {$\tilde{\tau}$};
        \node [style=right label] (3) at (0,0.4) {$f\ns/\ns\backsim$};
        \node [style=right label] (4) at (0,-0.8) {$\mathcal{S}(f)$};
        \end{pgfonlayer}
        \begin{pgfonlayer}{edgelayer}
            \draw [sWire] (0) to (1);
        \end{pgfonlayer}
	\etp.
    \eeq

    That $(\tilde{\sigma},\tilde{\tau})$ preserve the probabilities is a consequence of the fact that both $\mathcal{S}(f)$ and $f/\backsim$ reproduce the probabilities of the fragment $f$.
    
      \end{proof}
        Having established that any given shadow $\mathcal{S}(f)$ of a fragment $f$ is equivalent to the quotient shadow $f/\backsim$,
we can now prove 
Theorem~\ref{shadow-equiv}.

\shadowequiv*

\begin{proof}
If $\mathcal{S}_1(f)$ and $\mathcal{S}_2(f)$ are two shadows for $f$, Lemma~\ref{lemma:AllShadowsEquivQuotientShadow} tell us that  $\mathcal{S}_1(f)$ is equivalent to $f\ns/\ns\backsim$, as well as $\mathcal{S}_2(f)$ is equivalent to $f\ns/\ns\backsim$. Since equivalence is transitive and symmetric, we get that $\mathcal{S}_2(f)$ is equivalent to $\mathcal{S}_1(f)$.
\end{proof}

The following Lemma applies to any fragment, and will be useful below.
\begin{lemma}(Probability rules are unique)
\label{lemma:ProbRuleUnique}
    Consider a fragment $f:=(\Omega_f,\mathcal{E}_f,p_f)$, and suppose that $(\Omega_f,\mathcal{E}_f,p'_f)$ is such that $p'_f$ also reproduces the probabilities of the fragment. Then, $p'_f=p_f$. 
\end{lemma}
\begin{proof}
    Suppose $p'_f(s,e)=p_f(s,e)$ for all $s\in\Omega_f$ and $e\in\mathcal{E}_f$. Since probability rules are bilinear, then $p'_f(u,v)=p_f(u,v)$ for all $u\in \mathsf{Span}[\Omega_f]$ and $v\in\mathsf{Span}[\mathcal{E}_f]$. Since probability rules are  defined with domains $\mathsf{Span}[\Omega_f]\times\mathsf{Span}[\mathcal{E}_f]$ this implies $p'_f=p_f$. 
\end{proof}
As an instance of Lemma~\ref{lemma:ProbRuleUnique}, we get the following.

\begin{corollary}[The only freedom in a shadow map is given by $\sigma$ and $\tau$] 
\label{corollary:SigmaTaufixProbRule}
Consider a fragment $f=(\Omega_f,\mathcal{E}_f,p_f)$ and two shadows for this fragment, $\mathcal{S}_1(f)=(\sigma_1(\Omega_f),\tau_1(\mathcal{E}_f),p_{s_1})$ and $\mathcal{S}_2(f)=(\sigma_2(\Omega_f),\tau_2(\mathcal{E}_f),p_{s_2})$. If $\sigma_1=\sigma_2$ and $\tau_1=\tau_2$, then $p_{s_1}=p_{s_2}$. That is, given $\sigma$ and $\tau$, the probability rule $p_\mathcal{S}$ of the shadow $\mathcal{S}(f)$ is fixed.
\end{corollary}

	%\tocless
    \section{Proof of Lemma~\ref{fragsub}}\label{LinearityFromEmpAd}

We now prove Lemma~\ref{fragsub}. 

\fragsub*
 
 \begin{proof}
To see that the shadow embedding $(\sigma,\tau)$ is injective if and only if the fragment is relatively tomographic, notice the following: if the fragment is relatively tomographic, then for $s_1,s_2\in\Omega_f$ one has
\begin{eqnarray}
    \btp\tikzstate{s_1}{f}\etp\ \neq\ \btp\tikzstate{s_2}{f}\etp &\iff& 
    \btp
		\begin{pgfonlayer}{nodelayer}
			\node [style=Wsquareadj] (1) at (-0, -1.6) {$s_1$};
			\node [style=Wsquare] (2) at (-0, 1.6) {$e$};
			\node [style=right label] (4) at (0, 1) {$f$};
			\node [style=empty circle, fill=white] (8) at (0,0) {$p_f$};
			\node [style=right label] (6) at (0, -0.8) {$f$};
		\end{pgfonlayer}
		\begin{pgfonlayer}{edgelayer}
			\draw [qWire] (2) to (1.center);
		\end{pgfonlayer}
		\etp
  \ \neq\ 
  \btp
		\begin{pgfonlayer}{nodelayer}
			\node [style=Wsquareadj] (1) at (-0, -1.6) {$s_2$};
			\node [style=Wsquare] (2) at (-0, 1.6) {$e$};
			\node [style=right label] (4) at (0, 1) {$f$};
			\node [style=empty circle, fill=white] (8) at (0,0) {$p_f$};
			\node [style=right label] (6) at (0, -0.8) {$f$};
		\end{pgfonlayer}
		\begin{pgfonlayer}{edgelayer}
			\draw [qWire] (2) to (1.center);
		\end{pgfonlayer}
		\etp,\quad\exists e\in\mathcal{E}_f\\
  &\iff& \btp
		\begin{pgfonlayer}{nodelayer}
			\node [style=Wsquareadj] (1) at (-0, -2.8) {$s_1$};
			\node [style=Wsquare] (2) at (-0, 2.8) {$e$};
                \node [style=right label] (12) at (0,2.3) {$f$};
                \node [style=right label] (13) at (0,-2) {$f$};
			\node [style=empty circle, fill=gray!30] (8) at (0,0) {$p_\mathcal{S}$};
			\node [style=right label] (8) at (0, -0.7) {$\mathcal{S}(f)$};
			\node [style=right label] (9) at (0, 0.9) {$\mathcal{S}(f)$};
                \node [style=small box,inner sep=1pt,fill=gray!30] (10) at (0,-1.5) {$\sigma$};
                \node [style=small box, inner sep=1pt, fill=gray!30] (11) at (0,1.5) {$\tau$};
		\end{pgfonlayer}
		\begin{pgfonlayer}{edgelayer}
			\draw [qWire] (2) to (11);
                \draw [qWire] (10) to (1);
                \draw [sWire] (10) to (11);
		\end{pgfonlayer}
		\etp
  \ \neq\ 
  \btp
		\begin{pgfonlayer}{nodelayer}
			\node [style=Wsquareadj] (1) at (-0, -2.8) {$s_2$};
			\node [style=Wsquare] (2) at (-0, 2.8) {$e$};
                \node [style=right label] (12) at (0,2.3) {$f$};
                \node [style=right label] (13) at (0,-2) {$f$};
			\node [style=empty circle, fill=gray!30] (8) at (0,0) {$p_\mathcal{S}$};
			\node [style=right label] (8) at (0, -0.7) {$\mathcal{S}(f)$};
			\node [style=right label] (9) at (0, 0.9) {$\mathcal{S}(f)$};
                \node [style=small box,inner sep=1pt,fill=gray!30] (10) at (0,-1.5) {$\sigma$};
                \node [style=small box,inner sep=1pt,fill=gray!30] (11) at (0,1.5) {$\tau$};
		\end{pgfonlayer}
		\begin{pgfonlayer}{edgelayer}
			\draw [qWire] (2) to (11);
                \draw [qWire] (10) to (1);
                \draw [sWire] (10) to (11);
		\end{pgfonlayer}
		\etp,\ \exists e\in\mathcal{E}_f\\
  &\iff& \btp
		\begin{pgfonlayer}{nodelayer}
			\node [style=Wsquareadj] (1) at (-0, -1.2) {$s_1$};
			\node [style=none] (2) at (-0, 1.2) {};
			\node [style=right label] (4) at (0, -0.4) {$f$};
			\node [style=small box, inner sep=1pt, fill=gray!30] (5) at (0,0.2) {$\sigma$};
			\node [style=right label] (6) at (0,1) {$\mathcal{S}(f)$};
		\end{pgfonlayer}
		\begin{pgfonlayer}{edgelayer}
			\draw [sWire] (2.center) to (5);
			\draw [qWire] (5) to (1.center);
		\end{pgfonlayer}
		\etp
  \ \neq\ 
  \btp
		\begin{pgfonlayer}{nodelayer}
			\node [style=Wsquareadj] (1) at (-0, -1.2) {$s_2$};
			\node [style=none] (2) at (-0, 1.2) {};
			\node [style=right label] (4) at (0, -0.4) {$f$};
			\node [style=small box, inner sep=1pt, fill=gray!30] (5) at (0,0.2) {$\sigma$};
			\node [style=right label] (6) at (0,1) {$\mathcal{S}(f)$};
		\end{pgfonlayer}
		\begin{pgfonlayer}{edgelayer}
			\draw [sWire] (2.center) to (5);
			\draw [qWire] (5) to (1.center);
		\end{pgfonlayer}
		\etp.
\end{eqnarray}
Here, the first equivalence follows from the fragment being relatively tomographic, the second from Eq.~\eqref{probruleshadow}, and the third from the fact that the probability rule of a shadow is tomographic. Consequently, the map $\sigma$ is injective.
A similar argument gives $e_1\neq e_2\implies \tau(e_1)\neq\tau(e_2)$ for any $e_1,e_2\in\mathcal{E}_f$ of a tomographic fragment, so $\tau$ is injective. Since $\sigma$ and $\tau$ are injective, the shadow map is injective. 

Conversely, notice that for any $s_1,s_2\in\Omega_f$ such that $\sigma(s_1)=\sigma(s_2)$, injectivity of $\sigma$ gives us that
\begin{eqnarray}
    \btp\tikzstate{s_1}{f}\etp\ =\ \btp\tikzstate{s_2}{f}\etp &\impliedby& \btp
		\begin{pgfonlayer}{nodelayer}
			\node [style=Wsquareadj] (1) at (-0, -1.2) {$s_2$};
			\node [style=none] (2) at (-0, 1.2) {};
			\node [style=right label] (4) at (0, -0.4) {$f$};
			\node [style=small box, inner sep=1pt, fill=gray!30] (5) at (0,0.2) {$\sigma$};
			\node [style=right label] (6) at (0,1) {$\mathcal{S}(f)$};
		\end{pgfonlayer}
		\begin{pgfonlayer}{edgelayer}
			\draw [sWire] (2.center) to (5);
			\draw [qWire] (5) to (1.center);
		\end{pgfonlayer}
		\etp\ =\ \btp
		\begin{pgfonlayer}{nodelayer}
			\node [style=Wsquareadj] (1) at (-0, -1.2) {$s_2$};
			\node [style=none] (2) at (-0, 1.2) {};
			\node [style=right label] (4) at (0, -0.4) {$f$};
			\node [style=small box, inner sep=1pt, fill=gray!30] (5) at (0,0.2) {$\sigma$};
			\node [style=right label] (6) at (0,1) {$\mathcal{S}(f)$};
		\end{pgfonlayer}
		\begin{pgfonlayer}{edgelayer}
			\draw [sWire] (2.center) to (5);
			\draw [qWire] (5) to (1.center);
		\end{pgfonlayer}
		\etp\\
  &\iff& \btp
		\begin{pgfonlayer}{nodelayer}
			\node [style=Wsquareadj] (1) at (-0, -2.8) {$s_1$};
			\node [style=Wsquare] (2) at (-0, 2.8) {$e$};
                \node [style=right label] (12) at (0,2.3) {$f$};
                \node [style=right label] (13) at (0,-2) {$f$};
			\node [style=empty circle, fill=gray!30] (8) at (0,0) {$p_\mathcal{S}$};
			\node [style=right label] (8) at (0, -0.7) {$\mathcal{S}(f)$};
			\node [style=right label] (9) at (0, 0.9) {$\mathcal{S}(f)$};
                \node [style=small box,inner sep=1pt,fill=gray!30] (10) at (0,-1.5) {$\sigma$};
                \node [style=small box, inner sep=1pt, fill=gray!30] (11) at (0,1.5) {$\tau$};
		\end{pgfonlayer}
		\begin{pgfonlayer}{edgelayer}
			\draw [qWire] (2) to (11);
                \draw [qWire] (10) to (1);
                \draw [sWire] (10) to (11);
		\end{pgfonlayer}
		\etp
  \ =\ 
  \btp
		\begin{pgfonlayer}{nodelayer}
			\node [style=Wsquareadj] (1) at (-0, -2.8) {$s_2$};
			\node [style=Wsquare] (2) at (-0, 2.8) {$e$};
                \node [style=right label] (12) at (0,2.3) {$f$};
                \node [style=right label] (13) at (0,-2) {$f$};
			\node [style=empty circle, fill=gray!30] (8) at (0,0) {$p_\mathcal{S}$};
			\node [style=right label] (8) at (0, -0.7) {$\mathcal{S}(f)$};
			\node [style=right label] (9) at (0, 0.9) {$\mathcal{S}(f)$};
                \node [style=small box,inner sep=1pt,fill=gray!30] (10) at (0,-1.5) {$\sigma$};
                \node [style=small box,inner sep=1pt,fill=gray!30] (11) at (0,1.5) {$\tau$};
		\end{pgfonlayer}
		\begin{pgfonlayer}{edgelayer}
			\draw [qWire] (2) to (11);
                \draw [qWire] (10) to (1);
                \draw [sWire] (10) to (11);
		\end{pgfonlayer}
		\etp,\ \forall e\in\mathcal{E}_f\\
  &\iff& \btp
		\begin{pgfonlayer}{nodelayer}
			\node [style=Wsquareadj] (1) at (-0, -1.6) {$s_1$};
			\node [style=Wsquare] (2) at (-0, 1.6) {$e$};
			\node [style=right label] (4) at (0, 1) {$f$};
			\node [style=empty circle, fill=white] (8) at (0,0) {$p_f$};
			\node [style=right label] (6) at (0, -0.8) {$f$};
		\end{pgfonlayer}
		\begin{pgfonlayer}{edgelayer}
			\draw [qWire] (2) to (1.center);
		\end{pgfonlayer}
		\etp
  \ =\ 
  \btp
		\begin{pgfonlayer}{nodelayer}
			\node [style=Wsquareadj] (1) at (-0, -1.6) {$s_2$};
			\node [style=Wsquare] (2) at (-0, 1.6) {$e$};
			\node [style=right label] (4) at (0, 1) {$f$};
			\node [style=empty circle, fill=white] (8) at (0,0) {$p_f$};
			\node [style=right label] (6) at (0, -0.8) {$f$};
		\end{pgfonlayer}
		\begin{pgfonlayer}{edgelayer}
			\draw [qWire] (2) to (1.center);
		\end{pgfonlayer}
		\etp,\quad\forall e\in\mathcal{E}_f,
\end{eqnarray}
with a similar argument for effects $e_1,e_2\in\mathcal{E}_f$ such that $\tau(e_1)=\tau(e_2)$. This is the definition of a relatively tomographic fragment.
\end{proof}

%\tocless
\section{Proof of Theorem~\ref{thm:Holevo}}\label{sec:Holevogeneral}

We now prove Theorem~\ref{thm:Holevo}.

\holevo*

 \begin{proof}
Consider an arbitrary polytopic GPT system $G$---that is, any GPT system in which the state space $\Omega_G$ is a polytope and the effect space $\mathcal{E}_G$ is a polytope.\footnote{In fact, the proof does not rely on the effect space being a polytope.}  As $G$ is a proper GPT system rather than a generic fragment, its states and effects are tomographic for each other. For simplicity, here we will work with the representation in which states span a vector space $V$,  effects live in the dual space $V^*$, and the probability rule is simply the evaluation map. The wire labels in the diagrams below describe these vector spaces instead of the GPT systems and fragments, since we consider maps that a priori might not be state/effect maps.

Because $\Omega_G$ is a polytope, it has a finite set of pure (extremal) states, $\{s_\lambda\}_{\lambda=1}^n\in\Omega_G$, which we order in some arbitrary way as $s_1,\ldots,s_n$. Denote the pure states of a classical (simplicial) system as $\{\lambda\}_{\lambda=1}^n$, a unit basis of $\mathds{R}^n$, and order them as well.  Then, consider the linear map $L:\mathds{R}^n\to V$ defined by the action
\beq \label{Lmapdefn}
\btp
    \begin{pgfonlayer}{nodelayer}
        \node [style=none] (0) at (0,1.2) {};
        \node [style=inc,fill=gray!30] (1) at (0,0.3) {$L$};
        \node [style=Wsquareadj] (2) at (0,-1.2) {$\lambda$};
        \node [style=right label] (3) at (0,1.2) {$V$};
        \node [style=right label] (4) at (0,-0.6) {$\mathds{R}^n$};
    \end{pgfonlayer}
    \begin{pgfonlayer}{edgelayer}
        \draw [qWire] (0.center) to (1);
        \draw (1) to (2);
    \end{pgfonlayer}
\etp
\quad:=\quad
\btp
        \begin{pgfonlayer}{nodelayer}
        \node [style=Wsquareadj] (0) at (0,-0.6) {$s_\lambda$};
        \node [style=none] (1) at (0,0.6) {};
        \node [style=right label] (2) at (0,0.4) {$V$};
    \end{pgfonlayer}
    \begin{pgfonlayer}{edgelayer}
        \draw [qWire] (0) to (1.center);
    \end{pgfonlayer}
\etp,
\quad \forall s_\lambda\in\Omega_G.
\eeq

Since the classical states are linearly independent---as they form a basis for $\mathds{R}^n$---this linear map always exists.
This transformation $L$ takes the classical pure states to the pure states in $\Omega_G$, and so can be thought of as a classically controlled preparation procedure with a setting variable that ranges over all of the pure states in the GPT system $G$. (Thus, it can be thought of as a physical transformation.)

One can represent $L$ in matrix form in an intuitive way using the fact that $\{\lambda\}_{\lambda}$ is a unit basis of $\mathds{R}^n$. Then, if $\mathbf{s}_{\lambda}$ is the GPT vector representing state $s_\lambda\in\Omega_G$ in some basis, one has
\begin{align}
\label{eq:LMatrix}
    L= \left[
  \begin{array}{cccc}
    \vertbar & \vertbar &        & \vertbar \\
  {\bf  s}_{1}    & {\bf s}_{2}    & \ldots & {\bf s}_{n}    \\
    \vertbar & \vertbar &        & \vertbar 
  \end{array}
\right].
\end{align}

Consider now the contravariant\footnote{Diagrammatically, this just means one reads the diagram top to bottom instead of bottom to top.} action of $L$ on the effects in the original GPT system:
\beq \label{xidefn}
    \btp
        \tikzeffect{e}{V}
    \etp
    \quad\mapsto\quad
    \btp
        \begin{pgfonlayer}{nodelayer}
            \node [style=none] (0) at (0,-0.8) {};
            \node [style=Wsquare] (1) at (0,0.8) {$\xi_e$};
            \node [style=right label] (2) at (0,0) {$\mathds{R}^n$};
        \end{pgfonlayer}
        \begin{pgfonlayer}{edgelayer}
            \draw (0.center) to (1);
        \end{pgfonlayer}
    \etp
    \quad:=\quad
    \btp
    \begin{pgfonlayer}{nodelayer}
        \node [style=Wsquare] (0) at (0,1.2) {$e$};
        \node [style=inc,fill=gray!30] (1) at (0,-0.1) {$L$};
        \node [style=none] (2) at (0,-1.2) {};
        \node [style=right label] (3) at (0,0.7) {$V$};
        \node [style=right label] (4) at (0,-0.9) {$\mathds{R}^n$};
    \end{pgfonlayer}
    \begin{pgfonlayer}{edgelayer}
        \draw [qWire] (0.center) to (1);
        \draw (1) to (2);
    \end{pgfonlayer}
    \etp.
\eeq
As $L$ is a physical transformation, it is both a state map and an effect map, and so the $\xi_e$ are necessarily valid effects on a classical system.

We can use this map to define a fragment of the simplicial GPT system  $\Lambda$ of dimension $n$, which we call the Holevo fragment for $G$ and which we denote by $\mathscr{h}_G$. It contains all the states in the simplicial theory's state space, namely ${\omega\in\Delta}$, but only contains the effects of the form $\xi_e=e\circ L$ for $e\in\mathcal{E}_G$. 
Diagrammatically \footnote{ The dashed box here is just the diagrammatic representation of the evaluation map (the probability rule).}:
\beq
    \mathscr{h}_G:=\left(\Delta,\quad
    \left\{\btp
        \begin{pgfonlayer}{nodelayer}
            \node [style=none] (0) at (0,-0.8) {};
            \node [style=Wsquare] (1) at (0,0.4) {$\xi_e$};
            \node [style=right label] (2) at (0,-.4) {$\mathds{R}^n$};
        \end{pgfonlayer}
        \begin{pgfonlayer}{edgelayer}
            \draw (0.center) to (1);
        \end{pgfonlayer}
    \etp\right\}_{e\in\mathcal{E}_G},
   \quad
    \btp
		\begin{pgfonlayer}{nodelayer}
			\node [style=none] (1) at (-0, -0.6) {};
			\node [style=none] (2) at (-0, 0.6) {};
			\node [style=right label] (4) at (0, -0.6) {$\mathds{R}^n$};
			\node [style=empty diagram small] (5) at (0,0) {};
			\node [style=right label] (6) at (0, 0.6) {$\mathds{R}^n$};
		\end{pgfonlayer}
		\begin{pgfonlayer}{edgelayer}
			\draw (2) to (1.center);
		\end{pgfonlayer}
		\etp\right).
    \eeq

Next we will construct the GPT shadow $\mathcal{S}(\mathscr{h}_G)$ of $\mathscr{h}_G$ and show that it is equal to $G$.  

For that, consider a right inverse of $L$, namely $L_r^{-1}:V\to \mathds{R}^n$ satisfying
\beq\label{eq:defn-inverse-L}
 \btp
    \begin{pgfonlayer}{nodelayer}
        \node [style=none] (0) at (0,2) {};
        \node [style=none] (1) at (0,-2) {};
        \node [style=inc,fill=gray!30] (2) at (0,0.8) {$L$};
        \node [style=proj,fill=gray!30] (3) at (0,-0.8) {$L_r^{-1}$};
        \node [style=right label] (4) at (0,0) {$\mathds{R}^n$};
        \node [style=right label] (5) at (0,1.7) {$V$};
        \node [style=right label] (6) at (0,-1.7) {$V$};
    \end{pgfonlayer}
    \begin{pgfonlayer}{edgelayer}
        \draw [qWire] (0.center) to (2);
        \draw [qWire] (1.center) to (3);
        \draw (2) to (3);
    \end{pgfonlayer}
 \etp
 \quad=\quad
 \btp
    \begin{pgfonlayer}{nodelayer}
        \node [style=none] (0) at (0,1.3) {};
        \node [style=none] (1) at (0,-1.3) {};
        \node [style=right label] (2) at (0,1) {$V$};
    \end{pgfonlayer}
    \begin{pgfonlayer}{edgelayer}
        \draw [qWire] (0.center) to (1.center);
    \end{pgfonlayer}
 \etp.
\eeq
This map always exists, due to the fact that $L$ is surjective (i.e., any element $v\in V$ can be written as $L(w)$ for some $w\in\mathds{R}^n$) and that all surjective functions have right inverses. To see that $L$ is surjective, note that since $G$ is a proper GPT system, $\mathsf{Span}(\Omega_G)=V$, so any vector $v\in V$ can be written as \begin{align}
v=\sum_{\lambda}r_{\lambda}s_{\lambda}=\sum_{\lambda}r_{\lambda}L(\lambda) = L\left(\sum_{\lambda}r_{\lambda}\lambda\right) = L(w),
\end{align}
where the last equality follows from the fact that the deterministic distributions $\{\lambda\}_{\lambda}$ live in $\mathds{R}^n$. 

We now claim that we can define a shadow $\mathcal{S}(\mathscr{h}_G)$ of $\mathscr{h}_G$ by taking $\tau:= L_r^{-1}$, taking $\sigma:= L$, and taking $p_\mathcal{S}$ to be the evaluation map in $V$. With these definitions, the shadow is
\beq\label{eq:defn-shadow-L}
\mathcal{H}:=\left(
    \left\{
\btp
    \begin{pgfonlayer}{nodelayer}
        \node [style=none] (0) at (0,1.2) {};
        \node [style=inc,fill=gray!30] (1) at (0,0.3) {$L$};
        \node [style=Wsquareadj] (2) at (0,-1.2) {$\omega$};
        \node [style=right label] (3) at (0,1.2) {$V$};
        \node [style=right label] (4) at (0,-0.6) {$\mathds{R}^n$};
    \end{pgfonlayer}
    \begin{pgfonlayer}{edgelayer}
        \draw [qWire] (0.center) to (1);
        \draw (1) to (2);
    \end{pgfonlayer}
\etp
    \right\}_{\omega\in\Delta},
    \left\{
        \btp
    \begin{pgfonlayer}{nodelayer}
        \node [style=Wsquare] (0) at (0,1.2) {$\xi_e$};
        \node [style=proj,fill=gray!30] (1) at (0,-0.2) {$L_r^{-1}$};
        \node [style=none] (2) at (0,-1.2) {};
        \node [style=right label] (3) at (0,0.55) {$\mathds{R}^n$};
        \node [style=right label] (4) at (0,-1) {$V$};
    \end{pgfonlayer}
    \begin{pgfonlayer}{edgelayer}
        \draw (0.center) to (1);
        \draw [qWire] (1) to (2);
    \end{pgfonlayer}
    \etp
    \right\}_{e\in\mathcal{E}_G},
    \quad
    \btp
		\begin{pgfonlayer}{nodelayer}
			\node [style=none] (1) at (-0, -1.2) {};
			\node [style=none] (2) at (-0, 1.2) {};
			\node [style=right label] (4) at (0, -1) {$V$};
			\node [style=empty diagram small] (5) at (0,0) {};
			\node [style=right label] (6) at (0, 1) {$V$};
		\end{pgfonlayer}
		\begin{pgfonlayer}{edgelayer}
			\draw [qWire] (2) to (1.center);
		\end{pgfonlayer}
		\etp
\right).
\eeq

It is not hard to see that $\mathcal{H} = G$. That the effect space of Eq.~\eqref{eq:defn-shadow-L} is just $\mathcal{E}_G$ follows from the definition of $\xi_e$ (Eq.~\eqref{xidefn}). That the state space of Eq.~\eqref{eq:defn-shadow-L} is just $\Omega_G$ follows from the fact that the extremal elements of the former are of the form $L \circ \lambda$, and hence are extremal elements of $\Omega_G$ (by Eq.~\eqref{Lmapdefn}). 

It follows that $\mathcal{H}$ is a valid GPT system. Consequently, to prove that it is indeed a shadow of $\mathscr{h}_G$, all that we need to demonstrate is that $L$ and $L_r^{-1}$ are full state and effect maps, respectively. Clearly, $L$ is a full state map, since by ranging over the extremal states in $\mathds{R}^n$, one can reach every extremal state in $\Omega_G$. Moreover, $L_r^{-1}$ is a full effect map, since varying over $\xi_e$ is the same as varying $e\circ L$ over $e$, and $L_r^{-1}$ takes this to $e$.

Thus we see that $\mathcal{H}$ is indeed the shadow of $\mathscr{h}_G$, and that it is indeed equal to $G$. That is, $\mathcal{H}=\mathcal{S}(\mathscr{h}_G)=G$. 

Since this construction is independent of any particularities of the polytopic GPT system $G$, this establishes Theorem~\ref{thm:Holevo}.
\end{proof}

In the rest of the section, we essentially provide an alternative proof to the above, one which makes use of the quotienting definition of a shadow, and so is more explicit and may be more insightful to some readers. 

First we prove that $L_r^{-1}$ identifies effects that are not distinguished by the states in $\Omega_{\mathscr{h}}$, and that $L$ identifies states that are not distinguished by the effects in $\mathcal{E}_{\mathscr{h}}$. For the case of $L^{-1}_r$, it should identify effects $\xi\in\mathcal{E}_{\mathscr{h}}$ that are not distinguished by the states in the fragment; since the fragment contains all states in the simplex $\Delta$, $L^{-1}_r$ should not identify any two different effects. Indeed, consider two Holevo effects $\xi_e$ and $\xi'_{e'}$ in $\mathcal{E}^{\mathscr{h}}$. Then:

\begin{eqnarray}
    \btp
    \begin{pgfonlayer}{nodelayer}
        \node [style=Wsquare] (0) at (0,1.2) {$\xi_e$};
        \node [style=proj,fill=gray!30] (1) at (0,-0.2) {$L_r^{-1}$};
        \node [style=none] (2) at (0,-1.2) {};
        \node [style=right label] (3) at (0,0.55) {$\mathds{R}^n$};
        \node [style=right label] (4) at (0,-1) {$V$};
    \end{pgfonlayer}
    \begin{pgfonlayer}{edgelayer}
        \draw (0.center) to (1);
        \draw [qWire] (1) to (2);
    \end{pgfonlayer}
    \etp
    \ =\ 
    \btp
    \begin{pgfonlayer}{nodelayer}
        \node [style=Wsquare] (0) at (0,1.4) {$\xi'_{e'}$};
        \node [style=proj,fill=gray!30] (1) at (0,-0.2) {$L_r^{-1}$};
        \node [style=none] (2) at (0,-1.2) {};
        \node [style=right label] (3) at (0,0.55) {$\mathds{R}^n$};
        \node [style=right label] (4) at (0,-1) {$V$};
    \end{pgfonlayer}
    \begin{pgfonlayer}{edgelayer}
        \draw (0.center) to (1);
        \draw [qWire] (1) to (2);
    \end{pgfonlayer}
    \etp
    &\iff& 
    \btp
    \begin{pgfonlayer}{nodelayer}
        \node [style=Wsquare] (0) at (0,2) {$e$};
        \node [style=none] (1) at (0,-2) {};
        \node [style=inc,fill=gray!30] (2) at (0,0.8) {$L$};
        \node [style=proj,fill=gray!30] (3) at (0,-0.8) {$L_r^{-1}$};
        \node [style=right label] (4) at (0,0) {$\mathds{R}^n$};
        \node [style=right label] (5) at (0,1.5) {$V$};
        \node [style=right label] (6) at (0,-1.7) {$V$};
    \end{pgfonlayer}
    \begin{pgfonlayer}{edgelayer}
        \draw [qWire] (0.center) to (2);
        \draw [qWire] (1.center) to (3);
        \draw (2) to (3);
    \end{pgfonlayer}
 \etp
 \ =\ 
 \btp
    \begin{pgfonlayer}{nodelayer}
        \node [style=Wsquare] (0) at (0,2) {$e'$};
        \node [style=none] (1) at (0,-2) {};
        \node [style=inc,fill=gray!30] (2) at (0,0.8) {$L$};
        \node [style=proj,fill=gray!30] (3) at (0,-0.8) {$L_r^{-1}$};
        \node [style=right label] (4) at (0,0) {$\mathds{R}^n$};
        \node [style=right label] (5) at (0,1.5) {$V$};
        \node [style=right label] (6) at (0,-1.7) {$V$};
    \end{pgfonlayer}
    \begin{pgfonlayer}{edgelayer}
        \draw [qWire] (0.center) to (2);
        \draw [qWire] (1.center) to (3);
        \draw (2) to (3);
    \end{pgfonlayer}
 \etp\\
 &\iff&
 \btp
 \tikzeffect{e}{V}
 \etp\ =\ 
 \btp\tikzeffect{e'}{V}\etp.    
\end{eqnarray}

That is, two effects in the Holevo fragment are mapped to the same effect in the purported shadow (i.e., mapped by $L_r^{-1}$ to the same effect) if and only if they came from the same effect $e=e'$ in the original theory. But for $e=e'$ one has $\xi_e = e\circ L = e'\circ L= \xi'_{e'}$, and so
\beq
    \btp
    \begin{pgfonlayer}{nodelayer}
        \node [style=Wsquare] (0) at (0,1.2) {$\xi_e$};
        \node [style=proj,fill=gray!30] (1) at (0,-0.2) {$L_r^{-1}$};
        \node [style=none] (2) at (0,-1.2) {};
        \node [style=right label] (3) at (0,0.55) {$\mathds{R}^n$};
        \node [style=right label] (4) at (0,-1) {$V$};
    \end{pgfonlayer}
    \begin{pgfonlayer}{edgelayer}
        \draw (0.center) to (1);
        \draw [qWire] (1) to (2);
    \end{pgfonlayer}
    \etp
    \ =\ 
    \btp
    \begin{pgfonlayer}{nodelayer}
        \node [style=Wsquare] (0) at (0,1.4) {$\xi'_{e'}$};
        \node [style=proj,fill=gray!30] (1) at (0,-0.2) {$L_r^{-1}$};
        \node [style=none] (2) at (0,-1.2) {};
        \node [style=right label] (3) at (0,0.55) {$\mathds{R}^n$};
        \node [style=right label] (4) at (0,-1) {$V$};
    \end{pgfonlayer}
    \begin{pgfonlayer}{edgelayer}
        \draw (0.center) to (1);
        \draw [qWire] (1) to (2);
    \end{pgfonlayer}
    \etp
    \ \iff\ 
    \btp
        \begin{pgfonlayer}{nodelayer}
            \node [style=none] (0) at (0,-0.8) {};
            \node [style=Wsquare] (1) at (0,0.8) {$\xi_e$};
            \node [style=right label] (2) at (0,0) {$\mathds{R}^n$};
        \end{pgfonlayer}
        \begin{pgfonlayer}{edgelayer}
            \draw (0.center) to (1);
        \end{pgfonlayer}
    \etp\ =\ 
    \btp
        \begin{pgfonlayer}{nodelayer}
            \node [style=none] (0) at (0,-0.8) {};
            \node [style=Wsquare] (1) at (0,1) {$\xi'_{e'}$};
            \node [style=right label] (2) at (0,0) {$\mathds{R}^n$};
        \end{pgfonlayer}
        \begin{pgfonlayer}{edgelayer}
            \draw (0.center) to (1);
        \end{pgfonlayer}
    \etp
\eeq
so no distinct effects in $\mathscr{h}(G)$ are mapped to the same effect in the shadow $\mathcal{H}$ by the action of $L^{-1}_r$. Now, for the case of states, $L$ must identify those states in $\Delta$ that are not distinguished by the effects in $\mathscr{h}(G)$. Let us show this is indeed the case. First, we have that for all $\xi\in\mathcal{E}_{\mathscr{h}}$,
\beq
    \btp
        \begin{pgfonlayer}{nodelayer}
            \node [style=Wsquare] (0) at (0,0.7) {$\xi_e$};
            \node [style=Wsquareadj] (1) at (0,-0.7) {$\omega$};
            \node [style=right label] (2) at (0,0) {$\mathbbm{R}^n$};
        \end{pgfonlayer}
        \begin{pgfonlayer}{edgelayer}
            \draw (0) to (1);
        \end{pgfonlayer}
    \etp
\ = \ 
    \btp
        \begin{pgfonlayer}{nodelayer}
            \node [style=Wsquare] (0) at (0,0.7) {$\xi_{e}$};
            \node [style=Wsquareadj] (1) at (0,-0.7) {$\gamma$};
            \node [style=right label] (2) at (0,0) {$\mathbbm{R}^n$};
        \end{pgfonlayer}
        \begin{pgfonlayer}{edgelayer}
            \draw (0) to (1);
        \end{pgfonlayer}
    \etp
    \ \implies\ 
    \btp
        \begin{pgfonlayer}{nodelayer}
            \node [style=Wsquare] (0) at (0,1.4) {$e$};
            \node [style=Wsquareadj] (1) at (0,-1.4) {$\omega$};
            \node [style=inc, fill=gray!30] (2) at (0,0) {$L$};
            \node [style=right label] (3) at (0,0.7) {$V$};
            \node [style=right label] (4) at (0,-0.7) {$\mathbbm{R}^n$};
        \end{pgfonlayer}
        \begin{pgfonlayer}{edgelayer}
            \draw [qWire] (0) to (2);
            \draw (2) to (1);
        \end{pgfonlayer}
    \etp
    \ =\ 
    \btp
        \begin{pgfonlayer}{nodelayer}
            \node [style=Wsquare] (0) at (0,1.4) {$e$};
            \node [style=Wsquareadj] (1) at (0,-1.4) {$\gamma$};
            \node [style=inc, fill=gray!30] (2) at (0,0) {$L$};
            \node [style=right label] (3) at (0,0.7) {$V$};
            \node [style=right label] (4) at (0,-0.7) {$\mathbbm{R}^n$};
        \end{pgfonlayer}
        \begin{pgfonlayer}{edgelayer}
            \draw [qWire] (0) to (2);
            \draw (2) to (1);
        \end{pgfonlayer}
    \etp,\quad\forall e\in\mathcal{E}_G.
\eeq
\sloppy Above, the condition $\forall e\in\mathcal{E}_{G}$ follows from the fact that ${\{e\circ L\}_{e\in\mathcal{E}_G}=\mathcal{E}_{\mathscr{h}}}$, given the definition of the Holevo fragment $\mathscr{h}(G)$. Now, since $G$ is tomographically complete (as it is a proper GPT system), we get 
\beq
\btp
        \begin{pgfonlayer}{nodelayer}
            \node [style=Wsquare] (0) at (0,1.4) {$e$};
            \node [style=Wsquareadj] (1) at (0,-1.4) {$\omega$};
            \node [style=inc, fill=gray!30] (2) at (0,0) {$L$};
            \node [style=right label] (3) at (0,0.7) {$V$};
            \node [style=right label] (4) at (0,-0.7) {$\mathbbm{R}^n$};
        \end{pgfonlayer}
        \begin{pgfonlayer}{edgelayer}
            \draw [qWire] (0) to (2);
            \draw (2) to (1);
        \end{pgfonlayer}
    \etp
    \ =\ 
    \btp
        \begin{pgfonlayer}{nodelayer}
            \node [style=Wsquare] (0) at (0,1.4) {$e$};
            \node [style=Wsquareadj] (1) at (0,-1.4) {$\gamma$};
            \node [style=inc, fill=gray!30] (2) at (0,0) {$L$};
            \node [style=right label] (3) at (0,0.7) {$V$};
            \node [style=right label] (4) at (0,-0.7) {$\mathbbm{R}^n$};
        \end{pgfonlayer}
        \begin{pgfonlayer}{edgelayer}
            \draw [qWire] (0) to (2);
            \draw (2) to (1);
        \end{pgfonlayer}
    \etp,\quad\forall e\in\mathcal{E}_G\implies\ 
    \btp
    \begin{pgfonlayer}{nodelayer}
        \node [style=none] (0) at (0,1.2) {};
        \node [style=inc,fill=gray!30] (1) at (0,0.3) {$L$};
        \node [style=Wsquareadj] (2) at (0,-1.2) {$\omega$};
        \node [style=right label] (3) at (0,1.2) {$V$};
        \node [style=right label] (4) at (0,-0.6) {$\mathds{R}^n$};
    \end{pgfonlayer}
    \begin{pgfonlayer}{edgelayer}
        \draw [qWire] (0.center) to (1);
        \draw (1) to (2);
    \end{pgfonlayer}
\etp
\ =\ 
\btp
    \begin{pgfonlayer}{nodelayer}
        \node [style=none] (0) at (0,1.2) {};
        \node [style=inc,fill=gray!30] (1) at (0,0.3) {$L$};
        \node [style=Wsquareadj] (2) at (0,-1.2) {$\gamma$};
        \node [style=right label] (3) at (0,1.2) {$V$};
        \node [style=right label] (4) at (0,-0.6) {$\mathds{R}^n$};
    \end{pgfonlayer}
    \begin{pgfonlayer}{edgelayer}
        \draw [qWire] (0.center) to (1);
        \draw (1) to (2);
    \end{pgfonlayer}
\etp
\eeq
showing that all states that are not distinguished by the effects in $\mathscr{h}(G)$ get identified by $L$. This implies that $\mathsf{Ker}_{L}[p_{\mathscr{h}}]\subseteq\mathsf{Ker}[L]$, where $p_{\mathscr{h}}$ is the probability rule of the Holevo fragment (the evaluation map). In other words, all states that are not distinguished by the effects in $\mathscr{h}$ get mapped to the same state by the map $L$. 

Finally, note that the probabilities predicted by $\mathcal{H}$ are the same as those predicted by $\mathscr{h}(G)$, since
\begin{eqnarray}
    \btp
        \begin{pgfonlayer}{nodelayer}
            \node [style=Wsquare] (0) at (0,2) {$\xi_e$};
            \node [style=Wsquareadj] (1) at (0,-2) {$\omega$};
            \node [style=proj,fill=gray!30] (2) at (0,0.6) {$L_r^{-1}$};
            \node [style=inc,fill=gray!30] (3) at (0,-0.8) {$L$};
            \node [style=right label] (4) at (0,-1.4) {$\mathbbm{R}^n$};
            \node [style=right label] (5) at (0, 1.3) {$\mathbbm{R}^n$};
            \node [style=right label] (6) at (0,-0.1) {$V$};
        \end{pgfonlayer}
        \begin{pgfonlayer}{edgelayer}
            \draw (0) to (2);
            \draw (1) to (3);
            \draw [qWire] (2) to (3);
        \end{pgfonlayer}
    \etp
    &=&
    \btp
        \begin{pgfonlayer}{nodelayer}
            \node [style=Wsquare] (0) at (0,2.5) {$e$};
            \node [style=Wsquareadj] (1) at (0,-2.5) {$\omega$};
            \node [style=inc,fill=gray!30] (2) at (0,1.3) {$L$};
            \node [style=proj,fill=gray!30] (3) at (0,0) {$L_r^{-1}$};
            \node [style=inc,fill=gray!30] (4) at (0,-1.3) {$L$};
            \node [style=right label] (5) at (0,0.7) {$\mathbbm{R}^n$};
            \node [style=right label] (6) at (0,-0.7) {$V$};
            \node [style=right label] (7) at (0,-1.9) {$\mathbbm{R}^n$};
            \node [style=right label] (8) at (0,2) {$V$};
        \end{pgfonlayer}
        \begin{pgfonlayer}{edgelayer}
            \draw [qWire] (0) to (2);
            \draw (2) to (3);
            \draw [qWire] (3) to (4);
            \draw (4) to (1);
        \end{pgfonlayer}
    \etp\\
    &=&
    \btp
        \begin{pgfonlayer}{nodelayer}
            \node [style=Wsquare] (0) at (0,1.4) {$e$};
            \node [style=Wsquareadj] (1) at (0,-1.4) {$\omega$};
            \node [style=inc, fill=gray!30] (2) at (0,0) {$L$};
            \node [style=right label] (3) at (0,0.7) {$V$};
            \node [style=right label] (4) at (0,-0.7) {$\mathbbm{R}^n$};
        \end{pgfonlayer}
        \begin{pgfonlayer}{edgelayer}
            \draw [qWire] (0) to (2);
            \draw (2) to (1);
        \end{pgfonlayer}
    \etp\\
    &=&
    \btp
        \begin{pgfonlayer}{nodelayer}
            \node [style=Wsquare] (0) at (0,0.7) {$\xi_e$};
            \node [style=Wsquareadj] (1) at (0,-0.7) {$\omega$};
            \node [style=right label] (2) at (0,0) {$\mathbbm{R}^n$};
        \end{pgfonlayer}
        \begin{pgfonlayer}{edgelayer}
            \draw (0) to (1);
        \end{pgfonlayer}
    \etp.
\end{eqnarray}
In the first and last equality, we used that $\xi_e=e\circ L$ and in the second equality we used that $L^{-1}_r\circ L=\mathds{1}_V$. 

 This also implies that any two states that are distinguished by the effects in $\mathscr{h}$ should be distinct in $\Omega_{\mathcal{H}}$. This means that $\mathsf{Ker}[L]\subseteq\mathsf{Ker}_{L}[p_{\mathscr{h}}]$ which then implies $\mathsf{Ker}[L]=\mathsf{Ker}_{ L}[p_{\mathscr{h}}]$. By the universal property of quotienting, $L=\tilde{L}\circ\backsim$, where $\backsim$ is the quotienting relation of Def.~\ref{defn:method4}, and $\tilde{L}$ is some injective map. What is important here is that if $L(\omega)\neq L(\gamma)$, then there exists some effect $f\in\mathcal{E}_{\mathcal{H}}$ that distinguishes them. Therefore, the left kernel of $p_{\mathcal{H}}$ is trivial. Given that the right kernel is also trivial, the evaluation map $p_{\mathcal{H}}$ is tomographic. Therefore, $\mathcal{H}$ is indeed a shadow of $\mathscr{h}(G)$.

We now prove that $\mathcal{H}=G$. We start by showing that the state space $\Omega_{\mathcal{H}}$ of the shadow is equal to $\Omega_G$ and that the effect space  $\mathcal{E}_{\mathcal{H}}$ of the shadow is equal to $\mathcal{E}_G$. 
First, note that all states $\nu\in\Omega_H$ can be written as $L(\omega)$ where $\omega\in\Delta$. Since a classical state $\omega\in\Delta$ is always a convex combination of the deterministic ones, $\{\lambda\}_{\lambda}\in\Delta$, we get
\beq
\btp
        \begin{pgfonlayer}{nodelayer}
        \node [style=Wsquareadj] (0) at (0,-0.6) {$\nu$};
        \node [style=none] (1) at (0,0.6) {};
        \node [style=right label] (2) at (0,0.4) {$V$};
    \end{pgfonlayer}
    \begin{pgfonlayer}{edgelayer}
        \draw [qWire] (0) to (1.center);
    \end{pgfonlayer}
\etp\ =\ 
\btp
    \begin{pgfonlayer}{nodelayer}
        \node [style=none] (0) at (0,1.2) {};
        \node [style=inc,fill=gray!30] (1) at (0,0.3) {$L$};
        \node [style=Wsquareadj] (2) at (0,-1.2) {$\omega$};
        \node [style=right label] (3) at (0,1.2) {$V$};
        \node [style=right label] (4) at (0,-0.6) {$\mathds{R}^n$};
    \end{pgfonlayer}
    \begin{pgfonlayer}{edgelayer}
        \draw [qWire] (0.center) to (1);
        \draw (1) to (2);
    \end{pgfonlayer}
\etp
\ =\ 
\sum_\lambda\alpha_\lambda 
\btp
    \begin{pgfonlayer}{nodelayer}
        \node [style=none] (0) at (0,1.2) {};
        \node [style=inc,fill=gray!30] (1) at (0,0.3) {$L$};
        \node [style=Wsquareadj] (2) at (0,-1.2) {$\lambda$};
        \node [style=right label] (3) at (0,1.2) {$V$};
        \node [style=right label] (4) at (0,-0.6) {$\mathds{R}^n$};
    \end{pgfonlayer}
    \begin{pgfonlayer}{edgelayer}
        \draw [qWire] (0.center) to (1);
        \draw (1) to (2);
    \end{pgfonlayer}
\etp
\ =\ 
\sum_\lambda\alpha_\lambda\ 
\btp
        \begin{pgfonlayer}{nodelayer}
        \node [style=Wsquareadj] (0) at (0,-0.6) {$s_\lambda$};
        \node [style=none] (1) at (0,0.6) {};
        \node [style=right label] (2) at (0,0.4) {$V$};
    \end{pgfonlayer}
    \begin{pgfonlayer}{edgelayer}
        \draw [qWire] (0) to (1.center);
    \end{pgfonlayer}
\etp
\in\Omega_G,
\eeq
\noindent where $\{\alpha_{\lambda}\}$ are coefficients of the convex combination. This tell us that every $\nu\in\Omega_{\mathcal{H}}$ belongs to $\Omega_G$. Now, any $\omega\in\Omega_G$ is a convex combination of the pure states $\{s_{\lambda}\}_{\lambda}$, so

\beq
\btp
        \begin{pgfonlayer}{nodelayer}
        \node [style=Wsquareadj] (0) at (0,-0.6) {$\omega$};
        \node [style=none] (1) at (0,0.6) {};
        \node [style=right label] (2) at (0,0.4) {$V$};
    \end{pgfonlayer}
    \begin{pgfonlayer}{edgelayer}
        \draw [qWire] (0) to (1.center);
    \end{pgfonlayer}
\etp\ =\ 
\sum_\lambda\beta\lambda\ 
\btp
        \begin{pgfonlayer}{nodelayer}
        \node [style=Wsquareadj] (0) at (0,-0.6) {$s_\lambda$};
        \node [style=none] (1) at (0,0.6) {};
        \node [style=right label] (2) at (0,0.4) {$V$};
    \end{pgfonlayer}
    \begin{pgfonlayer}{edgelayer}
        \draw [qWire] (0) to (1.center);
    \end{pgfonlayer}
\etp\ =\ 
\sum_\lambda\beta_\lambda 
\btp
    \begin{pgfonlayer}{nodelayer}
        \node [style=none] (0) at (0,1.2) {};
        \node [style=inc,fill=gray!30] (1) at (0,0.3) {$L$};
        \node [style=Wsquareadj] (2) at (0,-1.2) {$\lambda$};
        \node [style=right label] (3) at (0,1.2) {$V$};
        \node [style=right label] (4) at (0,-0.6) {$\mathds{R}^n$};
    \end{pgfonlayer}
    \begin{pgfonlayer}{edgelayer}
        \draw [qWire] (0.center) to (1);
        \draw (1) to (2);
    \end{pgfonlayer}
\etp\ =\ 
\btp
    \begin{pgfonlayer}{nodelayer}
        \node [style=none] (0) at (0,1.2) {};
        \node [style=inc,fill=gray!30] (1) at (0,0.3) {$L$};
        \node [style=Wsquareadj] (2) at (0,-1.2) {$\gamma$};
        \node [style=right label] (3) at (0,1.2) {$V$};
        \node [style=right label] (4) at (0,-0.6) {$\mathds{R}^n$};
    \end{pgfonlayer}
    \begin{pgfonlayer}{edgelayer}
        \draw [qWire] (0.center) to (1);
        \draw (1) to (2);
    \end{pgfonlayer}
\etp\in\Omega_\mathcal{H},
\eeq
where we have defined
\beq
 \btp
    \begin{pgfonlayer}{nodelayer}
        \node [style=Wsquareadj] (0) at (0,-0.6) {$\gamma$};
        \node [style=none] (1) at (0,0.6) {};
        \node [style=right label] (2) at (0,0.4) {$\mathds{R}^n$};
    \end{pgfonlayer}
    \begin{pgfonlayer}{edgelayer}
        \draw (0) to (1.center);
    \end{pgfonlayer}
 \etp
 \quad:=\quad
 \sum_\lambda\beta_\lambda\ 
 \btp
    \begin{pgfonlayer}{nodelayer}
        \node [style=Wsquareadj] (0) at (0,-0.6) {$\lambda$};
        \node [style=none] (1) at (0,0.6) {};
        \node [style=right label] (2) at (0,0.4) {$\mathds{R}^n$};
    \end{pgfonlayer}
    \begin{pgfonlayer}{edgelayer}
        \draw (0) to (1.center);
    \end{pgfonlayer}
 \etp.
\eeq
This shows that every $\omega\in\Omega_G$ is also an element of $\Omega_H$. Putting both facts together gives $\Omega_H=\Omega_G$. 

Now, let us turn to the effects. 
First, using Eq.~\ref{eq:defn-inverse-L}, any $e\in\mathcal{E}_{G}$ can be written as
\beq
    \btp
        \tikzeffect{e}{V}
    \etp
    \ =\ 
    \btp
    \begin{pgfonlayer}{nodelayer}
        \node [style=Wsquare] (0) at (0,2) {$e$};
        \node [style=none] (1) at (0,-2) {};
        \node [style=inc,fill=gray!30] (2) at (0,0.8) {$L$};
        \node [style=proj,fill=gray!30] (3) at (0,-0.8) {$L_r^{-1}$};
        \node [style=right label] (4) at (0,0) {$\mathds{R}^n$};
        \node [style=right label] (5) at (0,1.5) {$V$};
        \node [style=right label] (6) at (0,-1.7) {$V$};
    \end{pgfonlayer}
    \begin{pgfonlayer}{edgelayer}
        \draw [qWire] (0.center) to (2);
        \draw [qWire] (1.center) to (3);
        \draw (2) to (3);
    \end{pgfonlayer}
 \etp
 \ =\ 
  \btp
    \begin{pgfonlayer}{nodelayer}
        \node [style=Wsquare] (0) at (0,1.2) {$\xi_e$};
        \node [style=proj,fill=gray!30] (1) at (0,-0.2) {$L_r^{-1}$};
        \node [style=none] (2) at (0,-1.2) {};
        \node [style=right label] (3) at (0,0.55) {$\mathds{R}^n$};
        \node [style=right label] (4) at (0,-1) {$V$};
    \end{pgfonlayer}
    \begin{pgfonlayer}{edgelayer}
        \draw (0.center) to (1);
        \draw [qWire] (1) to (2);
    \end{pgfonlayer}
    \etp\in\mathcal{E}_{\mathcal{H}}.
\eeq
Conversely, for every effect $f\in\mathcal{E}_\mathcal{H}$ in the shadow,  there is (by Eq.~\ref{eq:defn-shadow-L}) a unique effect $e\in\mathcal{E}_{G}$ in the original GPT system such that $f = \xi_e\circ L^{-1}_r$, so we denote the shadow effects as $f_e$.  Given the definition of $\xi_e$ as $e\circ L$, this implies that
\beq
    \btp
        \tikzeffect{f_e}{V}
    \etp
    \ =\ 
    \btp
    \begin{pgfonlayer}{nodelayer}
        \node [style=Wsquare] (0) at (0,1.2) {$\xi_e$};
        \node [style=proj,fill=gray!30] (1) at (0,-0.2) {$L_r^{-1}$};
        \node [style=none] (2) at (0,-1.2) {};
        \node [style=right label] (3) at (0,0.55) {$\mathds{R}^n$};
        \node [style=right label] (4) at (0,-1) {$V$};
    \end{pgfonlayer}
    \begin{pgfonlayer}{edgelayer}
        \draw (0.center) to (1);
        \draw [qWire] (1) to (2);
    \end{pgfonlayer}
    \etp\ =\ 
    \btp
    \begin{pgfonlayer}{nodelayer}
        \node [style=Wsquare] (0) at (0,2) {$e$};
        \node [style=none] (1) at (0,-2) {};
        \node [style=inc,fill=gray!30] (2) at (0,0.8) {$L$};
        \node [style=proj,fill=gray!30] (3) at (0,-0.8) {$L_r^{-1}$};
        \node [style=right label] (4) at (0,0) {$\mathds{R}^n$};
        \node [style=right label] (5) at (0,1.5) {$V$};
        \node [style=right label] (6) at (0,-1.7) {$V$};
    \end{pgfonlayer}
    \begin{pgfonlayer}{edgelayer}
        \draw [qWire] (0.center) to (2);
        \draw [qWire] (1.center) to (3);
        \draw (2) to (3);
    \end{pgfonlayer}
 \etp\ =\ 
 \btp
        \tikzeffect{e}{V}
    \etp\in\mathcal{E}_{G},
\eeq
which proves that $f\in\mathcal{E}_{\mathcal{H}}\implies f\in\mathcal{E}_{G}$. 
This proves that $\mathcal{E}_{\mathcal{H}}=\mathcal{E}_{G}$ and that $f_u = u$.

In addition to having identical state spaces and effect spaces and unit effects, the probability rules for $\mathcal{H}$ and $G$ are identical, so they are in fact identical GPT systems.

%\tocless
\section{Proof of Theorem~\ref{thm:TheoryAgnosticTomography}}\label{sec:TheoryAgnosticTomographyProof}

We now prove Theorem~\ref{thm:TheoryAgnosticTomography}. 

\theoryagnostictomography*

\begin{proof}
	Note that in this proof, the diagram wires are labelled by the domain (codomain) of the map instead of by the GPT systems and fragments.
	
		In theory-agnostic tomography, one factorizes the matrix of empirical data $D\in M_{m\times n}(\mathds{R})$ generated by the states and effects in one's fragment $f$ by finding the smallest integer $k$ such that there exist $E\in M_{m\times k}(\mathds{R})$ and $S\in M_{k\times n}(\mathds{R})$ satisfying
		
		\beq \label{probruleTATdecomp}
\begin{tikzpicture}
	\begin{pgfonlayer}{nodelayer}
		\node [style=small map] (0) at (-0, 0) {$D$};
		\node [style=none] (1) at (-0, 1) {};
		\node [style=none] (2) at (-0, -1) {};
		\node [style=right label] (0) at (-0, 1) {$\mathbb{R}^m$};
		\node [style=right label] (4) at (-0, -1) {$\mathbb{R}^n$};
	\end{pgfonlayer}
	\begin{pgfonlayer}{edgelayer}
		\draw (1.center) to (2.center);
	\end{pgfonlayer}
\end{tikzpicture}}  \quad = \quad %
\begin{tikzpicture}
	\begin{pgfonlayer}{nodelayer}
		\node [style=proj,fill=gray!30] (0) at (0, -0.75) {$S$};
		\node [style=none] (1) at (0, 1.5) {};
		\node [style=none] (2) at (0, -1.5) {};
		\node [style=right label] (3) at (0, 1.5) {$\mathbb{R}^m$};
		\node [style=right label] (4) at (0, -1.5) {$\mathbb{R}^n$};
		\node [style=inc,fill=gray!30] (5) at (0,0.75) {$E$};
		\node [style=right label] (6) at (0,0) {$\mathbb{R}^k$};
	\end{pgfonlayer}
	\begin{pgfonlayer}{edgelayer}
		\draw (1.center) to (5);
		\draw (5) to (0);
		\draw (0) to (2.center);
	\end{pgfonlayer}
\end{tikzpicture}}.
		\eeq
	
Consider now the states $\Theta:=\{s_1,\dots,s_n\}$ and effects $\Phi:=\{e_1,\dots,e_m\}$ defining\footnote{Recall from Section~\ref{sec:TheoryAgnosticTomography} that these may just be taken to be the convexly extremal states and effects in the fragment.} the original (unknown) fragment $f$, and define a map $d:\Theta\to\mathds{R}^n$ mapping every $s_i\in\Theta$ to one of the $n$ orthonormal basis vectors in $\mathds{R}^n$:
	\beq
\begin{tikzpicture}
	\begin{pgfonlayer}{nodelayer}
		\node [style=small map, inner sep=1pt, fill=gray!30] (0) at (-0, 0) {$d$};
		\node [style=none] (1) at (-0, 1) {};
		\node [style=Wsquareadj] (2) at (-0, -1.5) {$s_i$};
		\node [style=right label] (3) at (-0, 1) {$\mathbb{R}^n$};
		\node [style=right label] (4) at (0, -0.75) {$\Theta$};
	\end{pgfonlayer}
	\begin{pgfonlayer}{edgelayer}
		\draw (1.center) to (0);
        \draw [qWire] (0) to (2.center);
	\end{pgfonlayer}
\end{tikzpicture}}
	 \ :=\quad
\begin{tikzpicture}
	\begin{pgfonlayer}{nodelayer}
		\node [style=none] (1) at (-0, 0.5) {};
		\node [style=Wsquareadj] (2) at (-0, -0.5) {$i$};
		\node [style=right label] (3) at (-0, 0.5) {$\mathbb{R}^n$};
	\end{pgfonlayer}
	\begin{pgfonlayer}{edgelayer}
		\draw (1.center) to (2.center);
	\end{pgfonlayer}
\end{tikzpicture}}.
	\eeq
Note that this map has a set ($\Theta$), not a vector space, as its domain; consequently, one cannot speak about whether or not it is linear. The map $D\circ d:\Theta\to\mathbbm{R}^m$ therefore takes states in $\Theta$ to columns of the data table $D$.

Due to the minimality of $k$, the process $E$ must admit of a left-inverse, i.e., there is $E^{-1}_l$ such that
	\beq
\begin{tikzpicture}
	\begin{pgfonlayer}{nodelayer}
		\node [style=inc,fill=gray!30] (0) at (0, -0.75) {$E$};
		\node [style=none] (1) at (0, 1.5) {};
		\node [style=none] (2) at (0, -1.5) {};
		\node [style=right label] (3) at (0, 1.5) {$\mathbb{R}^k$};
		\node [style=right label] (4) at (0, -1.5) {$\mathbb{R}^k$};
		\node [style=proj,fill=gray!30] (5) at (0,0.75) {$E^{-1}_l$};
		\node [style=right label] (6) at (0,0) {$\mathbb{R}^m$};
	\end{pgfonlayer}
	\begin{pgfonlayer}{edgelayer}
		\draw (1.center) to (5);
		\draw (5) to (0);
		\draw (0) to (2.center);
	\end{pgfonlayer}
\end{tikzpicture}}
		\quad = \quad
\begin{tikzpicture}
	\begin{pgfonlayer}{nodelayer}
		\node [style=none] (1) at (0, 1.5) {};
		\node [style=none] (2) at (0, -1.5) {};
		\node [style=right label] (3) at (0, 1.5) {$\mathbb{R}^k$};
	\end{pgfonlayer}
	\begin{pgfonlayer}{edgelayer}
		\draw (1.center) to (2.center);
	\end{pgfonlayer}
\end{tikzpicture}}.
	\eeq
Consequently, let us define the map
	\beq
\begin{tikzpicture}
	\begin{pgfonlayer}{nodelayer}
		\node [style=none] (1) at (0, 1) {};
		\node [style=none] (2) at (0, -1) {};
		\node [style=right label] (3) at (0, 1) {$\mathbb{R}^k$};
		\node [style=right label] (4) at (0,-1) {$\Theta$};
		\node [style=small map,inner sep=1pt,fill=gray!30] (5) at (0,0) {$\hat{\sigma}$};
	\end{pgfonlayer}
	\begin{pgfonlayer}{edgelayer}
		\draw (1.center) to (5);
		\draw [qWire] (5) to (2.center);
	\end{pgfonlayer}
\end{tikzpicture}}
	\ :=\quad
\begin{tikzpicture}
	\begin{pgfonlayer}{nodelayer}
		\node [style=none] (1) at (0,2.4) {};
		\node [style=none] (2) at (0, -2.4) {};
		\node [style=right label] (3) at (0, 2.4) {$\mathbb{R}^k$};
		\node [style=right label] (4) at (0,-2.3) {$\Theta$};
		\node [style=small map] (5) at (0,0) {$D$};
		\node [style=right label] (6) at (0,0.8) {$\mathbb{R}^m$};
		\node [style= right label] (7) at (0,-0.8) {$\mathbb{R}^n$};
		\node [style= small map,inner sep=1pt, fill=gray!30] (8) at (0,-1.6) {$d$};
		\node [style=proj,fill=gray!30] (9) at (0,1.6) {$E^{-1}_l$};
	\end{pgfonlayer}
	\begin{pgfonlayer}{edgelayer}
		\draw (1.center) to (8);
		\draw [qWire] (8) to (2.center);
	\end{pgfonlayer}
\end{tikzpicture}}
        \ =\quad
\begin{tikzpicture}
	\begin{pgfonlayer}{nodelayer}
		\node [style=proj, fill={gray!30}] (0) at (0, 0.75) {$S$};
		\node [style=none] (1) at (0, -1.5) {};
		\node [style=none] (2) at (0, 1.5) {};
		\node [style=right label] (3) at (0, -1.5) {$\Theta$};
		\node [style=right label] (4) at (0, 1.5) {$\mathbb{R}^k$};
		\node [style=small map, inner sep=1pt, fill={gray!30}] (5) at (0, -0.75) {$d$};
		\node [style=right label] (6) at (0, 0) {$\mathbb{R}^n$};
		\node [style=none] (7) at (0, 1.75) {};
	\end{pgfonlayer}
	\begin{pgfonlayer}{edgelayer}
		\draw (1.center) to (5);
		\draw (5) to (0);
		\draw [qWire] (5) to (1.center);
		\draw (7.center) to (0);
	\end{pgfonlayer}
\end{tikzpicture}},
	\eeq
which takes states $s_i\in\Omega$ to columns of the matrix $S$.

We now define a map $\sigma$, which (as we will show later) is the linear extension of $\hat{\sigma}$, acting on $\mathsf{Span}[\Theta]$. For each vector $s\in\mathsf{Span}[\Theta]$, pick one specific linear combination $s=\sum_i\alpha_is_i$ of that vector in terms of the states $s_i$, and define the action of $\sigma$ by
\beq
\begin{tikzpicture}
	\begin{pgfonlayer}{nodelayer}
		\node [style=none] (0) at (0, 1.7) {};
		\node [style=small box, fill={gray!30}] (1) at (0, 0.8) {$\sigma$};
		\node [style=Wsquareadj] (2) at (0, -1.2) {$s_i$};
		\node [style=right label] (3) at (0, 1.7) {$\mathbbm{R}^k$};
		\node [style=right label] (4) at (0, -0.1) {$S_\Theta$};
		\node [style=none] (5) at (-1.7, -1.3) {$\sum_i\alpha_i$};
		\node [style=none] (6) at (-2.7, -2.4) {};
		\node [style=none] (7) at (-2.7, -0.45) {};
		\node [style=none] (8) at (1, -0.45) {};
		\node [style=none] (9) at (1, -2.4) {};
	\end{pgfonlayer}
	\begin{pgfonlayer}{edgelayer}
		\draw (0.center) to (1);
		\draw [qWire] (1) to (2);
		\draw [thick gray dashed edge] (8.center)
			 to (9.center)
			 to (6.center)
			 to (7.center)
			 to cycle;
	\end{pgfonlayer}
\end{tikzpicture}
\ :=\ 
\sum_i\alpha_i\ 
\btp
    \begin{pgfonlayer}{nodelayer}
        \node [style=none] (0) at (0,1.2) {};
        \node [style=small box,fill=gray!30] (1) at (0,0.4) {$\hat{\sigma}$};
        \node [style=Wsquareadj] (2) at (0,-1.2) {$s_i$};
        \node [style=right label] (3) at (0,1.3) {$\mathbbm{R}^k$};
        \node [style=right label] (4) at (0,-0.4) {$\Theta$};
    \end{pgfonlayer}
    \begin{pgfonlayer}{edgelayer}
        \draw (0.center) to (1);
        \draw [qWire] (1) to (2);
    \end{pgfonlayer}
\etp,\quad\alpha_i\in\mathbbm{R},\ \forall i
\eeq
In principle, this map could be non-linear, as it might depend on the particular choice of decomposition for each vector $s$. However, below we will show that it forms part of a valid shadow map taking $f$ to $G_{D_f}$, so
by Lemma~\ref{lem:LinearityFromEmpAd}, it  is indeed linear.

Similarly to what was done for the states, one can define a map $d'$ taking effects $e_i\in\Phi$ to one of the $m$ orthonormal basis vectors in $\mathbbm{R}^m$, such that the composition $d'\circ D$ maps the effects $e_i\in\Phi$ to rows of the matrix $D$:
	\beq
\begin{tikzpicture}
	\begin{pgfonlayer}{nodelayer}
		\node [style=small map,inner sep=1pt, fill=gray!30] (0) at (-0, 0) {$d'$};
		\node [style=none] (1) at (-0, -1) {};
		\node [style=Wsquare] (2) at (-0, 1.5) {$e_i$};
		\node [style=right label] (3) at (-0, -1) {$\mathbb{R}^m$};
		\node [style=right label] (4) at (0, 0.8) {$\Phi$};
	\end{pgfonlayer}
	\begin{pgfonlayer}{edgelayer}
		\draw (1.center) to (0);
        \draw [qWire] (0) to (2.center);
	\end{pgfonlayer}
\end{tikzpicture}}
	\ :=\quad
\begin{tikzpicture}
	\begin{pgfonlayer}{nodelayer}
		\node [style=none] (1) at (-0, -0.5) {};
		\node [style=Wsquare] (2) at (-0, 0.5) {$i$};
		\node [style=right label] (3) at (-0, -0.5) {$\mathbb{R}^m$};
	\end{pgfonlayer}
	\begin{pgfonlayer}{edgelayer}
		\draw (1.center) to (2.center);
	\end{pgfonlayer}
\end{tikzpicture}}.
	\eeq
And by minimality of $k$, the map $S$ must admit of a right-inverse, $S^{-1}_r$, such that
	\beq
\begin{tikzpicture}
	\begin{pgfonlayer}{nodelayer}
		\node [style=proj,fill=gray!30] (0) at (0, 0.75) {$S$};
		\node [style=none] (1) at (0, -1.5) {};
		\node [style=none] (2) at (0, 1.5) {};
		\node [style=right label] (3) at (0, -1.5) {$\mathbb{R}^k$};
		\node [style=right label] (4) at (0, 1.5) {$\mathbb{R}^k$};
		\node [style=inc,fill=gray!30] (5) at (0,-0.75) {$S^{-1}_r$};
		\node [style=right label] (6) at (0,0) {$\mathbb{R}^n$};
	\end{pgfonlayer}
	\begin{pgfonlayer}{edgelayer}
		\draw (1.center) to (5);
		\draw (5) to (0);
		\draw (0) to (2.center);
	\end{pgfonlayer}
\end{tikzpicture}}
	\quad =\quad
\begin{tikzpicture}
	\begin{pgfonlayer}{nodelayer}
		\node [style=none] (1) at (0, 1.5) {};
		\node [style=none] (2) at (0, -1.5) {};
		\node [style=right label] (3) at (0, 1.5) {$\mathbb{R}^k$};
	\end{pgfonlayer}
	\begin{pgfonlayer}{edgelayer}
		\draw (1.center) to (2.center);
	\end{pgfonlayer}
\end{tikzpicture}}.
	\eeq
So, analogous to the above, we define the map
	\beq
\begin{tikzpicture}
	\begin{pgfonlayer}{nodelayer}
		\node [style=none] (1) at (0, 1) {};
		\node [style=none] (2) at (0, -1) {};
		\node [style=right label] (3) at (0, -1) {$\mathbb{R}^k$};
		\node [style=right label] (4) at (0,1) {$\Phi$};
		\node [style=small map,inner sep=1pt,fill=gray!30] (5) at (0,0) {$\hat{\tau}$};
	\end{pgfonlayer}
	\begin{pgfonlayer}{edgelayer}
		\draw [qWire] (1.center) to (5);
		\draw (5) to (2.center);
	\end{pgfonlayer}
\end{tikzpicture}}
	\ :=\quad
\begin{tikzpicture}
	\begin{pgfonlayer}{nodelayer}
		\node [style=none] (1) at (0, -2.5) {};
		\node [style=none] (2) at (0,2.5) {};
		\node [style=right label] (3) at (0, -2.5) {$\mathbb{R}^k$};
		\node [style=right label] (4) at (0,2.5) {$\Phi$};
		\node [style=small map] (5) at (0,0) {$D$};
		\node [style=right label] (6) at (0,0.8) {$\mathbb{R}^m$};
		\node [style= right label] (7) at (0,-0.8) {$\mathbb{R}^n$};
		\node [style= small map,inner sep=1pt,fill=gray!30] (8) at (0,1.6) {$d'$};
		\node [style=proj,fill=gray!30] (9) at (0,-1.6) {$S^{-1}_r$};
	\end{pgfonlayer}
	\begin{pgfonlayer}{edgelayer}
		\draw [qWire] (2.center) to (8);
		\draw (8) to (1.center);
	\end{pgfonlayer}
\end{tikzpicture}}
	\eeq
which takes effects $e_i\in\Phi$ to rows of the matrix $E$. 
Finally, we define the linear extension of $\hat{\tau}$, namely $\tau:\mathsf{Span}(\Phi)=:S_\Phi\to\mathds{R}^k$, just as we did for the map $\sigma$ above. (Here too, defining $\tau$ by a particular choice for the decomposition of each vector in its domain might give the appearance that the map could be nonlinear, but we will justify below that it is in fact linear.)

It now follows that $\sigma(\mathsf{Conv}[\Theta])$ and $\tau(\mathsf{Conv}[\Phi])$, together with the evaluation map on $\mathds{R}^k$, constitutes a shadow for the fragment $f$ that generated the data table. Indeed, the evaluation map in $\mathds{R}^k$ is tomographic, due to minimality of $k$, and the probabilities of $f$ (encoded in $D$) are reproduced:
\beq
   \btp
			\begin{pgfonlayer}{nodelayer}
				\node [style=Wsquareadj] (1) at (-0, -2.4) {$s_i$};
				\node [style=Wsquare] (2) at (-0, 2.4) {$e_j$};
				\node [style=right label] (4) at (0, 1.7) {$S_{\Phi}$};
				\node [style=small box,fill=gray!30] (5) at (0, 0.8) {$\tau$};
				\node [style=small box, fill=gray!30] (7) at (0, -0.8) {$\sigma$};
				\node [style=right label] (8) at (0,0) {$\mathbbm{R}^k$};
				\node [style=right label] (6) at (0, -1.7) {$S_{\Theta}$};
			\end{pgfonlayer}
			\begin{pgfonlayer}{edgelayer}
				\draw [qWire] (2) to (5);
				\draw (5) to (7);
				\draw [qWire] (7) to (1);
			\end{pgfonlayer}
			\etp
   \ =\ 
   \btp
			\begin{pgfonlayer}{nodelayer}
				\node [style=Wsquareadj] (1) at (-0, -3.2) {$s_i$};
				\node [style=Wsquare] (2) at (-0, 3.2) {$e_j$};
				\node [style=right label] (4) at (0, 2.5) {$\Phi$};
				\node [style=small box,fill=gray!30] (5) at (0, 1.6) {$d'$};
				\node [style=small box, fill=gray!30] (7) at (0, -1.6) {$d$};
				\node [style=right label] (8) at (0,0.8) {$\mathbbm{R}^m$};
				\node [style=right label] (6) at (0, -2.4) {$\Theta$};
                \node [style=right label] (9) at (0,-0.8) {$\mathbbm{R}^n$};
                \node [style=small map] (10) at (0,0) {$D$};
			\end{pgfonlayer}
			\begin{pgfonlayer}{edgelayer}
				\draw [qWire] (2) to (5);
				\draw (5) to (7);
				\draw [qWire] (7) to (1);
			\end{pgfonlayer}
			\etp
   \ =\ 
   \btp
			\begin{pgfonlayer}{nodelayer}
				\node [style=Wsquareadj] (1) at (-0, -1.4) {$i$};
				\node [style=Wsquare] (2) at (-0, 1.4) {$j$};
				\node [style=right label] (4) at (0, 0.8) {$\mathbbm{R}^m$};
				\node [style=small map] (8) at (0,0) {$D$};
				\node [style=right label] (6) at (0, -0.8) {$\mathbbm{R}^n$};
			\end{pgfonlayer}
			\begin{pgfonlayer}{edgelayer}
				\draw (2) to (1);
			\end{pgfonlayer}
			\etp
   \ =\ D(i,j).
   \eeq
Since $\sigma$ and $\tau$ constitute shadow maps, Lemma~\ref{lem:LinearityFromEmpAd} tell us that they are linear. Therefore, the particular choice of linear combination for a vector $s\in\mathsf{Span}[\Theta]$ in terms of $\{s_i\}$, made when defining $\sigma$, is irrelevant (and similarly for $\tau$).

Therefore, we have established that the output of theory-agnostic tomography applied to the data table generated by a GPT fragment gives the shadow of that fragment as its output.

\end{proof}

\end{document}